\documentclass[runningheads,orivec]{llncs}
\usepackage[T1]{fontenc}
\usepackage{graphicx}
\usepackage[utf8]{inputenc}

\usepackage{amssymb}    
\usepackage{marvosym}
\usepackage{stmaryrd}
\usepackage{bussproofs}
\usepackage{mathtools}
\usepackage{bbold}
\usepackage{enumerate}
\usepackage{qtree}
\usepackage{tikz-cd}
\usetikzlibrary{calc}

\usetikzlibrary{decorations.pathreplacing, calligraphy}

\usepackage{ifthen}
\usepackage{xspace}
\usepackage{fancybox}
\usepackage{hhline}

\usepackage{soul}

\usepackage{adjustbox}

\usepackage{xcolor}
\usepackage{multirow}
\usepackage{microtype}

\usepackage{hyperref}

\hypersetup{colorlinks=true}

\usepackage[colorinlistoftodos]{todonotes}
\usepackage{proof}
\usepackage{pifont}

\usepackage{appendix}
\capturecounter{theorem}
\capturecounter{lemma}
\capturecounter{proposition}
\usepackage{marginnote}

\newboolean{appendix}
\setboolean{appendix}{false}
\newboolean{withlinkedproofs}
\setboolean{withlinkedproofs}{true}

\newcommand{\NoteProof}[1]{
        \ifthenelse{\boolean{withproofs}}{\ifthenelse{\boolean{appendix}}{
        \marginnote{Originally at p. \pageref{#1}}
        }{
        \marginnote{{Proof at p.\,{\pageref{app:#1}}}}
        }
        }{}
}

\newcommand{\NoteFullProof}[1]{
	\ifthenelse{\boolean{withproofs}}{\ifthenelse{\boolean{appendix}}{
			\marginnote{Originally at p. \pageref{#1}}
		}{
			\marginnote{{Full Proof at p.\,{\pageref{app:#1}}}}
		}
	}{}
}

\newcommand{\applabel}[1]{$\phantomsection\label{app:#1}$}

\newcommand{\sem}[1]{\interp{#1}}

\newcommand{\semint}[1]{\interp{#1}^{\intsym}}

\newcommand{\ignore}[1]{}
 
\newcommand{\hcolspace}{@{\hspace{.25cm}}}
\newcommand{\myinput}[1]{\ifthenelse{\boolean{withimages}}{\input{#1}}{}}

\newcommand{\reflemma}[1]{Lemma~\ref{l:#1}}

\newcommand{\reflemmaeq}[1]{{L.\ref{l:#1}}}

\newcommand{\refprop}[1]{Prop.~\ref{prop:#1}}
\newcommand{\refsect}[1]{Section~\ref{sect:#1}}

\renewcommand{\refeq}[1]{(\ref{eq:#1})} 
\newcommand{\reffig}[1]{Fig.~\ref{fig:#1}}

\newcommand{\ie}{\textit{i.e.}\xspace}
\newcommand{\eg}{\textit{e.g.}\xspace}
\newcommand{\ih}{\textit{i.h.}\xspace}

\newcommand{\defeq}{\coloneqq} 

\newcommand{\grameq}{\Coloneqq} 
\newcommand{\set}[1]{\{#1\}}

\newcommand{\nat}{\mathbb{N}}
\newcommand{\size}[1]{|#1|}

\renewcommand{\l}{\lambda}
\newcommand{\isub}[2]{\{#1/#2\}}
\renewcommand{\isub}[2]{\{#1\,{:=}\,#2\}}
\newcommand{\esub}[2]{[#1/#2]}
\renewcommand{\esub}[2]{[#1{\shortleftarrow}#2]}
\newcommand{\fv}[1]{{\tt fv}(#1)}

\newcommand{\rootRew}[1]{\mapsto_{#1}}
\newcommand{\Rew}[1]{\rightarrow_{#1}}

\newcommand{\shufeqext}{\shufeqext} 

\newcommand{\tm}{t}
\newcommand{\tmtwo}{u}
\newcommand{\tmthree}{s}
\newcommand{\tmfour}{r}

\newcommand{\tmp}{\tm'}
\newcommand{\tmtwop}{\tmtwo'}

\newcommand{\var}{x}
\newcommand{\vartwo}{y}
\newcommand{\varthree}{z}
\newcommand{\varfour}{w}

\newcommand{\ctxholep}[1]{\langle #1\rangle}
\newcommand{\ctxhole}{\ctxholep{\cdot}}

\newcommand{\ctx}{C}

\newcommand{\ctxp}[1]{\ctx\ctxholep{#1}}

\newcommand{\hctx}{H}

\newcommand{\hauxctx}{P}

\newcommand{\hauxctxp}[1]{\hauxctx\ctxholep{#1}}

\newcommand{\arbctxp}[1]{\arbctxp{#1}}
\newcommand{\arbctxtwop}[1]{\arbctxtwop{#1}}

\newcommand{\la}[1]{\lambda #1.}

\newcommand{\myproof}[1]{
\ifthenelse{\boolean{omitproofs}}{\begin{IEEEproof} Proof available but omitted for readability. \end{IEEEproof}}{#1}}

\newcommand{\node}{\mathtt{n}}

\newcommand{\withproofs}[1]{\ifthenelse{\boolean{withproofs}}{#1}{}}

\newcommand{\withoutproofs}[1]{\ifthenelse{\boolean{withproofs}}{}{#1}}

\newcommand{\doubt}[1]{}

\newcounter{numberone}
\newcounter{numberoneroman}
\newcounter{numberonealph}

\newcommand{\mset}[1]{[#1]}

\newcommand{\zero}{\mathbf{0}}

\newcommand{\typctx}{\Gamma}
\newcommand{\typctxtwo}{\Delta}

\newcommand{\tder}{\pi}
\newcommand{\tdertwo}{\sigma}

\newcommand{\ruleAp}{@}

\newcommand{\hastype}{\!:\!}

\makeatletter
\newsavebox{\@brx}
\newcommand{\llangle}[1][]{\savebox{\@brx}{\(\m@th{#1\langle}\)}%
  \mathopen{\copy\@brx\kern-0.7\wd\@brx\usebox{\@brx}}}
\newcommand{\rrangle}[1][]{\savebox{\@brx}{\(\m@th{#1\rangle}\)}%
  \mathclose{\copy\@brx\kern-0.7\wd\@brx\usebox{\@brx}}}
\makeatother

\newcommand{\symfont}[1]{\mathtt{#1}}

\newcommand{\Id}{\symfont{I}}

\usepackage{proof}

\newcommand{\FV}[1]{\mathsf{fv}(#1)}

\newcommand{\bshs}{\Downarrow_{\mathrm{h}}\,} 
\newcommand{\bsh}[1]{\Downarrow_{\mathrm{h}}^{#1}\,} 
\newcommand{\bshdiv}{\not\bsh{}}

\newcommand{\bohm}{B{\"o}hm\xspace}

\newcommand{\mtype}{\typefont{M}}
\newcommand{\mtypetwo}{\typefont{N}}
\newcommand{\mtypethree}{\typefont{O}}

\newcommand{\emptytype}{[~]}
\renewcommand{\emptytype}{\zero}

\newcommand{\multitype}[2]{[{#2}_{1},\ldots,{#2}_{#1}]}

\newcommand{\ltype}{\typefont{L}}
\newcommand{\ltypetwo}{\ltype'}

\newcommand{\gtype}{\typefont{T}}
\newcommand{\gtypetwo}{\gtype'}

\newcommand{\vartype}{X}
\renewcommand{\vartype}{\typefont{A}}

\newcommand{\typectx}{\Gamma}
\newcommand{\typectxtwo}{\Delta}

\newcommand{\emptytypectx}{\emptyset}

\newcommand{\typingruleApp}{@}
\newcommand{\typingruleAx}{\mathsf{ax}}
\newcommand{\typingruleAbs}{\lambda}
\newcommand{\typingruleMany}{\mathsf{many}}

\newcommand{\interp}[1]{\llbracket #1 \rrbracket}

\newcommand{\tderiv}{\tder}
\newcommand{\tderivp}{\tderiv'} 
\newcommand{\tderivtwo}{\tdertwo}

\newcommand{\derives}{\vartriangleright}

\newcommand{\htm}{h}

\newcommand{\blueclr}{\blue{\mathsf{b}}}
\newcommand{\redclr}{\red{\mathsf{r}}}

\newcommand{\bla}[1]{\blue{\lambda_{\blueclr}} #1.}

\newcommand{\rapp}[2]{#1 \red{\bullet_{\redclr}} #2}
\newcommand{\bapp}[2]{#1 \blue{\bullet_{\blueclr}} #2}

\newcommand{\hsym}{\symfont{h}}

\newcommand{\toh}{\Rew{\hsym}}

\newcommand{\bshcols}{\Downarrow_{\hchsym}\,} 
\newcommand{\bshcol}[1]{\Downarrow_{\hchsym}^{\intsym#1}\,} 
\newcommand{\bshcoldiv}{\not\bshcols}

\newcommand{\Lambdac}{\Lambda_{\RB}}

\newcommand{\equivchcol}{\sqsubseteq_{\hcolsym}^{\qcontextualsym}}
\newcommand{\leqchcol}{\sqsubseteq_{\hcolsym}^{\qcontextualsym}}
\renewcommand{\equivchcol}{\equiv_{\intsym}^{\mathrm{ctx}}}
\renewcommand{\leqchcol}{\sqsubseteq_{\intsym}^{\mathrm{ctx}}}

\newcommand{\intleq}{\sqsubseteq^{\mathrm{int}}}

\newcommand{\abscolorone}{c}
\newcommand{\abscolortwo}{d}
\newcommand{\clr}[1]{\mathsf{\abscolorone}_{#1}}
\newcommand{\clrtwo}[1]{\mathsf{\abscolortwo}_{#1}}
\newcommand{\clrp}[1]{\mathsf{\abscolorone}'_{#1}}

\newcommand{\colr}{\mathsf{\abscolorone}}
\newcommand{\colrtwo}{\mathsf{\abscolortwo}}
\newcommand{\colrp}{\mathsf{\abscolorone}'}

\newcommand{\cla}[2]{\lambda_{\clr{#1}} #2.}

\newcommand{\ccapp}[3]{#2 \bullet_{\clr {#1}} #3}
\newcommand{\capp}[3]{#2 \bullet_{\clrtwo {#1}} #3}

\newcommand{\appsym}{\bullet}
\renewcommand{\capp}[3]{#2 \appsym_{\clrtwo {#1}} #3}

\newcommand{\manyclam}[2]{\lambda_{\clr{1}\cdots\clr{#1}} #2_1\ldots#2_{#1}.\,}
\newcommand{\manyblam}[2]{\lambda_{\blueclr\cdots \blueclr}#2_1\ldots#2_{#1}.\,}
\newcommand{\manyrlam}[3][1]{\lambda_{\redclr\cdots \redclr}#3_{#1}\ldots#3_{#2}.\,}

\newcommand{\manycapp}[3]{\capp{#1}{\capp{1}{#2}{#3_1} \cdots}{#3_{#1}}}

\newcommand{\clavec}[2]{\lambda_{\vec{#1}\phantom{.}}\vec{#2}.}
\newcommand{\cappvec}[3]{#2\appsymp{\vec{#1}}\vec{#3}}

\newcommand{\ctypes}[1]{\vdash^{#1}_{\intsym}}

\newcommand{\typearrowp}[1]{\xrightarrow{{#1}}}

\newcommand{\monoToPlayer}[2]{\overline{#2}^{#1}}

\newcommand{\monoToBlue}[1]{\monoToPlayer{\blueclr}{#1}}
\newcommand{\monoToRed}[1]{\monoToPlayer{\redclr}{#1}}

\newcommand{\spsym}{\mathrm{sp}}

\newcommand{\laxsym}{\eta^\infty}
\newcommand{\laxbsym}{\cB\laxsym}

\newcommand{\etabtle}{\sqsubseteq_{\laxbsym}}

\newcommand{\laxredsym}{\eta_{\symfont{red}}^\infty}
\newcommand{\laxbredsym}{\cB\laxredsym}
\newcommand{\etaredbtleq}{\sqsubseteq_{\laxbredsym}}

\newcommand{\slbtleq}{\etaredbtleq}

\newcommand{\exder}{\,\triangleright}

\newcommand\mplus{\uplus}

\newcommand{\bsub}{\begin{enumerate}[(i)]}
\renewcommand{\esub}{\end{enumerate}}
\newcommand{\obsle}[1][]{\sqsubseteq_{#1}^{\mathrm{ctx}}}

\newcommand{\lam}{\ensuremath{\lambda}}
\newcommand{\las}[2]{\lambda #1_{1}\dots #1_{#2}.}
\newcommand{\apps}[2]{{#1_1}\cdots #1_{#2}}
\newcommand{\Tupler}[1]{{\sf T}_{#1}}
\newcommand{\Tuple}[1]{\langle #1\rangle}
\newcommand{\Proj}[2]{{\sf S}^{#1}_{#2}}
\newcommand{\Var}{\textsc{Var}}

\newcommand{\convc}[2]{\!\Downarrow_{#1}^{#2}\,} 
\newcommand{\relation}[1]{{\sf #1}}

\newcommand{\iset}{I}

\newcommand{\yinyang}[1][1]{%
    \begin{tikzpicture}[scale=#1*0.07]
      \draw[line width = #1*0.05mm,transform canvas={yshift=0.02cm}] (0,0) circle (1cm);
      \path[fill=black,transform canvas={yshift=0.02cm}] (90:1cm) arc (90:-90:0.5cm)
                        (0,0)    arc (90:270:0.5cm)
                        (0,-1cm) arc (-90:-270:1cm);

    \end{tikzpicture}}
\newcommand{\intsym}{\yinyang}
\newcommand{\nointsym}{\tau}

\renewcommand{\appsym}{\cdot}
\newcommand{\appsymp}[1]{\appsym^{#1}}

\renewcommand{\blueclr}{\bullet}
\renewcommand{\redclr}{\circ}

\newcommand{\blackcol}{\blueclr}

\newcommand{\bappsym}{\blueclr}
\newcommand{\rappsym}{\redclr}

\renewcommand{\bla}[1]{\lambda_{\blueclr} #1.}

\renewcommand{\rapp}[2]{#1 \rappsym #2}
\renewcommand{\bapp}[2]{#1 \bappsym #2}

\newcommand{\bnointsym}{\beta_{\nointsym}}
\newcommand{\bintsym}{\beta_{\intsym}}
\newcommand{\bchsym}{\beta_{\checkerssym}}

\newcommand{\rtobnoint}{\rootRew{\bnointsym}}
\newcommand{\rtobint}{\rootRew{\bintsym}}
\newcommand{\tobnoint}{\Rew{\bnointsym}}
\newcommand{\tobint}{\Rew{\bintsym}}

\newcommand{\checkerssym}{\redclr \blueclr}
\renewcommand{\Lambdac}{{\Lambda_{\checkerssym}}}

\newcommand{\lap}[2]{\lambda_{#1} #2.}
\newcommand{\appp}[3]{#2 \appsymp{#1} #3}
\renewcommand{\ccapp}[3]{\appp{\clr{#1}}{#2}{#3}}
\renewcommand{\capp}[3]{\appp{\clrtwo{#1}}{#2}{#3}}

\newcommand{\typingruleAppInt}{@_{\intsym}}
\newcommand{\typingruleAppNoInt}{@_{\nointsym}}

\newcommand{\hnointsym}{\hsym_{\nointsym}}
\newcommand{\hintsym}{\hsym_{\intsym}}
\newcommand{\hchsym}{\hsym_{\checkerssym}}

\newcommand{\tohint}{\Rew{\hintsym}} 
\newcommand{\tohnoint}{\Rew{\hnointsym}}

\newcommand{\restr}{\!\!\upharpoonright}
\newcommand{\head}{\relation{h}}
\newcommand{\comb}[1]{\symfont{#1}}
\newcommand{\emptyseq}{\langle\rangle}

\newcommand{\cB}{\mathcal{B}}

\newcommand{\clrd}[1]{\clr{#1}^\bot}

\newcommand{\bId}{\Id_{\blueclr}}
\newcommand{\bOne}{\comb{1}_{\blueclr}}
\newcommand{\rId}{\Id_{\redclr}}

\newcommand{\tobch}{\Rew{\bchsym}}
\newcommand{\tohch}{\Rew{\hchsym}}

\newcommand{\speq}{=_{\spsym}}

\newcommand{\hnf}{h}

\newcommand{\insize}[1]{\size{#1}_{@}}

\newcommand{\typefont}[1]{{\mathsf{#1}}}

\tikzset{
ocenter/.style={baseline={([yshift=-.5ex, xshift=-.5ex]current bounding box)}},  
}

\newcommand{\leqchcolr}{\sqsubseteq_{\intsym}^{\mathrm{ctx\cdot imp}}}
\newcommand{\intrleq}{\sqsubseteq^{\mathrm{int\cdot imp}}}

\newcommand{\chcontexts}{\mathcal{C}_{\checkerssym}}
\newcommand{\chhcontexts}{\mathcal{H}_{\checkerssym}}

\newcommand{\wash}[1]{\underline{#1}}
\newcommand{\Tr}[3]{(#1,#2,#3)}
\newcommand{\Pair}[2]{\langle #1,#2\rangle}
\newcommand{\Kr}[2]{\boldsymbol{\delta}^\bot_{#1,#2}}
\newcommand{\lecol}{\sqsubseteq_{\mathrm{col}}}
\newcommand{\lecolpos}{\sqsubseteq^{+}_{\mathrm{col}}}

\newcommand{\hasstype}[1][k]{\vdash^{#1}_s} 
\renewcommand{\hasstype}[1][k]{\ctypes{#1}}
\newcommand{\whiterpos}{\le^+}
\newcommand{\whiterneg}{\le^-}
\newcommand{\abs}[1]{|#1|}

\newcommand{\leqimp}{\sqsubseteq^{\mathrm{ctx\cdot imp}}}
\newcommand{\leqctx}{\sqsubseteq^{\mathrm{ctx}}}
\renewcommand{\leqimp}{\intrleq}
\renewcommand{\leqctx}{\intleq}

\renewcommand{\intrleq}{\sqsubseteq^{\mathrm{imp}}}
\renewcommand{\leqchcolr}{\sqsubseteq^{\mathrm{imp}}_{\intsym}}
\renewcommand{\leqchcol}{\sqsubseteq^{\mathrm{int}}_{\intsym}}
\renewcommand{\equivchcol}{\equiv^{\mathrm{int}}_{\intsym}}

\newcommand{\leqpwcplain}{\sqsubseteq^{\mathrm{pwc}}}
\renewcommand{\lecol}{\leqpwcplain}

\newcommand{\leqpwc}{\sqsubseteq^{\mathrm{pwc}}_{\intsym}}

\renewcommand{\lecolpos}{\leqpwc}

\newcommand{\whiternegp}[1]{\whiterneg_{#1}}
\newcommand{\whiterposp}[1]{\whiterpos_{#1}}
\newcommand{\whiterpp}[2]{\le^{#1}_{#2}}
\newcommand{\pol}{\mathsf{p}}
\newcommand{\polop}{\lnot\mathsf{p}}
\newcommand{\whiterpol}[2]{\le^{#1}_{#2}}
\newcommand{\posorneg}{\pol}
\newcommand{\posornegop}{\polop}
\newcommand{\possym}{+}
\newcommand{\negsym}{-}

\newcommand{\madewhiterneg}[1]{{{#1}_{\ominus}}}
\newcommand{\madewhiterpos}[1]{{{#1}^{\oplus}}}
\newcommand{\madewhiternegandpos}[1]{{{#1}^{\oplus}_{\ominus}}}

\newboolean{withproofs}
\setboolean{withproofs}{true}

\newcommand{\techreport}[2]{\ifthenelse{\boolean{withproofs}}{#1}{#2}}

\begin{document}
\title{Interaction Improvement}
%
%
\author{Adrienne Lancelot\inst{1,2,3}\orcidID{0009-0009-5481-5719}\thanks{\protect\techreport{The omitted proofs can be found in the appendix. This document is an extended version of the conference publication \cite{conference-version}.}{The omitted proofs can be found in the technical report \cite{lancelot2026interactionimprovement}.}}
\and Giulio~Manzonetto\inst{3}\orcidID{0000-0003-1448-9014} 
\and Guy McCusker\inst{4}\orcidID{0000-0002-0305-6398}
\and Gabriele Vanoni\inst{3}\orcidID{0000-0001-8762-8674}
}
\authorrunning{A. Lancelot, G. Manzonetto, G. McCusker, G. Vanoni}
%
\institute{Inria, France and LIX, Ecole Polytechnique, Palaiseau
\and Università di Bologna, Italy \email{adrienne.lancelot@unibo.it} 
	\and
Université Paris Cité, CNRS, IRIF, F-75013, Paris, France
\email{\{gmanzone,gabriele.vanoni\}@irif.fr}
\and
Department of Computer Science, University of Bath, United Kingdom
\email{G.A.McCusker@bath.ac.uk}
}
\maketitle              
\begin{abstract}

The relational semantics of linear logic is a powerful framework for defining \emph{resource-aware} models of the $\lambda$-calculus. 
However, its \emph{quantitative} aspects are not reflected in the preorders and equational theories induced by these models. Indeed, they can be characterized in terms of (in)equalities between Böhm trees up to extensionality, which are \emph{qualitative} in nature. We employ the recently introduced checkers calculus to define a quantitative contextual preorder on $\lambda$-terms, and demonstrate that it coincides with the preorder associated to the relational semantics. 

\end{abstract}
%
%
%

\section{Introduction}

Much of programming language theory is based upon the powerful notion of \emph{observational equivalence}~\cite{morris1968lambda}. A program phrase may be placed in an execution context, yielding a whole program which can be executed. By specifying what \emph{observations} one may make of a whole program, we obtain a theory of observational equivalence: two program phrases are said to be observationally equivalent if, in any execution context, they give rise to the same observations. 
That is to say,  one cannot distinguish them, no matter what context they are placed in. 
Typical observations might be whether the programs terminate, return a particular integer or string, and so on. This notion of equivalence is both intuitively appealing and naturally \emph{compositional}: a program occurring as a fragment of a larger piece of code can always be replaced by an equivalent one, yielding equivalent code. 

It is natural to refine such a notion of equivalence to an observational \emph{preorder}, typically defined by asking that the observations made of one program are a subset of those made of another. For example, if one observes only termination, then a program \emph{simulates another one} if, whenever they are both placed in the same execution context, if the simulated program terminates, then the simulator does as well. When the simulated program does not terminate, the simulator is free to display any behavior. The observational preorder allows semanticists to talk about how to refine a program to handle more cases of execution contexts and which \emph{updates} of code are sound. Formally, taking as our notion of observation the termination predicate $\convc{}{}$, we define the observational preorder as:


\[\begin{array}{cccc}
        t \obsle u &\mbox{ if, for all contexts }C,  & 
        C\ctxholep{t}\!\mathop{\Downarrow} \Rightarrow C\ctxholep{u}\!\mathop{\Downarrow}
\end{array}
\]

Observational equivalences and preorders provide a robust methodology that may be applied to a wide variety of programming languages, from $\lambda$-calculus~\cite{Barendregt84} to sophisticated higher-order languages with computational effects~\cite{DBLP:journals/jacm/EhrhardPT18,DBLP:journals/csur/PatrignaniAC19,lassen+strovring-bisimilarity-eta,DBLP:conf/fossacs/BiernackiLP19,biernacki_et_al:LIPIcs:2020:12329}, with relatively little change to the basic definitions. A large body of theory supports reasoning about these theories: they may be captured by denotational semantics~\cite{Hyland76,Wadsworth76}, simplified by means of \emph{context lemmas} that characterize a subset of contexts sufficient to make all possible distinctions~\cite{berry:inria-00076481,DBLP:journals/iandc/AbramskyO93,Howe1996method}, or analyzed by coinductive bisimulation-style techniques~\cite{lassen1999bisimulation}. However, in most formulations, the available observations are \emph{extensional}: they report \emph{what} programs compute but not \emph{how}.
Intensional information such as how many computation steps are used, and what resources are consumed, are ignored. 
There is a challenging need for a robust way to express the fact that one program actually \emph{optimizes} another.

\paragraph{Improvements.}
Sands proposed \emph{improvement theory}~\cite{SandsImprovementTheory,SandsToplas,DBLP:journals/tcs/Sands96} as an execution-time-sensitive refinement of the observational preorder. He formulated a notion of observational improvement, where the observation is given not by mere termination $\Downarrow$  but by a quantitative notion $\Downarrow^k$ which reports that the number of steps required to terminate was $k$. We say that $u$ (observationally) improves $t$, i.e.\
\newcommand{\impleq}{\sqsubseteq^{\mathrm{imp}}}
\[\begin{array}{cccc}
        t \impleq u &\mbox{ if for all contexts }C,\forall k\ge 0\,. &
        C\ctxholep{t}\!\Downarrow^k \Rightarrow
        \exists k'\leq k.~C\ctxholep{u}\!\Downarrow^{k'}.
\end{array}
\]
Despite its intuitive appeal, this notion of improvement is considerably less studied than the standard preorders above. Perhaps one reason for this is the absence of a rich theory of denotational semantics capturing intensional information.


\paragraph{Denotational Semantics \& Equational Theories.} Denotational semantics gives an interpretation to program phrases by representing them as elements of a mathematical structure~\cite{ScottS71}, often a morphism in some category. This representation induces an equational theory on program phrases. In almost all cases, the interpretation of a phrase in the model is defined in a \emph{compositional} fashion, so that the interpretation of a program is constructed from those of its constituent phrases. Moreover, the interpretation is typically \emph{invariant under computation}: in the $\lambda$-calculus, for example, terms related by $\beta$-reductions will receive equal denotations. Thus the induced inequational theory is not sensitive to reduction steps, and does not capture notions of improvement. Nevertheless some denotational models---notably games models and relational models---do contain intensional information beyond what is captured by the contextual preorder. 


\paragraph{Intensional Semantics \& Quantitative Information.} Game semantics and relational semantics therefore offer the tantalising possibility of capturing quantitative information in a model and hence in the induced theory. Relational models record the interaction of a function with its arguments via a \emph{multiset} of input items; game models go further and record the full sequence of interactions between a program and its execution context. Despite this, the application of these models to quantitative analysis of program execution has been limited. The original impetus behind work on game semantics of programming languages was the so-called full abstraction problem for \textsf{PCF}, so attention was firmly on the \emph{extensional} behaviour of programs: in the fully-abstract models of \textsf{PCF}~\cite{DBLP:journals/iandc/AbramskyJM00,DBLP:journals/iandc/HylandO00}, all intensional information is explicitly quotiented out, to achieve the desired full abstraction theorems. 
In games models of more powerful languages, such as those involving stateful computation, the intensional information in the model turns out to be exactly the same as what is captured by the standard contextual preorder. These lines of work therefore do not appear to offer a denotational account of an improvement-style ordering. Indeed, in both game and relational semantics, the intensional information revealed by the semantics of a term is invariant under reduction of that term: the exposed information concerns only 
the interactions that a term may have with its context, and not what happens inside the term. One may draw an analogy with communications in distributed systems and client-server protocols~\cite{Fokkink2007}: game and relational semantics focus their attention on external communications. Purely internal steps are not exposed, thus allowing the semantics to remain invariant under reductions. This distinction mirrors the observation that local computations within a client or a server are comparatively fast or easily optimized, whereas the communication between them is the true performance bottleneck. Nevertheless this focus appears to make the idea of an improvement ordering inapplicable. 

One approach to recovering an improvement-style ordering is demonstrated by Ghica's \emph{slot games}~\cite{DBLP:conf/popl/Ghica05} which introduce a global counter, exposing internal reductions at the top level, and as a result achieve full abstraction with respect to the improvement ordering. A similar example of instrumenting the semantics globally, in the case of relational semantics, appears in~\cite{LairdMMP13}, though no improvement-style characterization is given. The question of understanding the intensional information available in the models (without such global counters) remains unanswered. 

Our focus in this paper is the relational semantics of the $\lambda$-calculus, which we study in this paper in the form of a system of non-idempotent intersection types. Non-idempotent intersection type systems are able to provide complexity analyses of $\lambda$-terms, yielding evaluation time bounds, as pioneered by de Carvalho \cite{Carvalho07} and later refined by Accattoli et al. \cite{DBLP:journals/jfp/AccattoliGK20}, as well as execution space bounds \cite{DBLP:journals/pacmpl/AccattoliLV22}. This quantitative information, for the most part, only appears at the level of type derivations---the construction of the denotation of a term---and are invisible in the final denotation and hence in the inequational theory. However, this inequational theory in general does not coincide with the observational preorder. The following questions then arise naturally: 
\begin{center}
What is the inequational theory induced by these semantics?\\ And is it cost-aware as the type system suggests?
\end{center}

\paragraph{Counting Interactions.} This paper is about reconciling the two notion of program optimizations and denotational semantics.
Specifically, we characterize the preorder induced by the relational semantics as an improvement-style preorder. 
Our approach makes use of the quantitative analysis of programs introduced in \cite{InteractionEquivalence}. This novel theory aims to account for the distinction between internal reductions and external interactions, as described above. Rather than counting \emph{all} computations steps, as Sands's improvement theory would, it only counts the \emph{interaction} steps between a term and its context. 
To achieve the distinction between interaction steps and internal steps, we use the \emph{checkers calculus} of \cite{InteractionEquivalence} which is a bichromatic version of the untyped $\lambda$-calculus. 
Checkers terms are color-annotated $\lambda$-terms, where the abstraction and the application constructors are tagged by a color, either black or white. Intuitively, programs are players uniformly painted black, and contexts are opponents uniformly painted white. As a term is executed in a context, the black and white parts meet and mix. 
This allows us to define notions of quantitative observational equivalence, preorder and \emph{interaction improvement}, akin to Sands's improvement, but counting only computation steps that involve a white constructor interacting with a black one. 

\paragraph{Characterizing The Costs at Play Behind Relational Semantics.}
Interaction improvement makes explicit, in particular, the quantitative aspects underlying the preorder induced by the relational model of \cite{DBLP:journals/entcs/HylandNPR06}.
The interpretation in this model can be described via a type system based on multi types, also known as non-idempotent intersection types \cite{DBLP:journals/mscs/PaoliniPR17,BKV17}.
The preorder induced by this model has been characterized as the \bohm tree preorder up to $\eta$-reductions \cite{BreuvartMR18}.
However, the execution costs involved in this relational model have so far remained unclear.

In this paper, we introduce a colored relational semantics for the checkers calculus (Figure~\ref{fig:colored-types-for-cbn}) and define an ordering on interpretations that compares their elements in a refined manner which is both qualitative and quantitative. 
In our main contribution (Theorem~\ref{thm:main}) we show that, when restricted to regular $\lambda$-terms, this ordering coincides with both the relational preorder---equivalently, with the \bohm tree preorder up to $\eta$-reductions---and interaction improvement.

It turns out to be quite challenging to show that the ordering we introduce on interpretations is \emph{compositional}. Though denotational approaches are typically compositional by definition, the fact that it presents difficulties here is perhaps not surprising given the intensional information we are tracking: similar challenges arise in intensional models such as game semantics \cite{DBLP:journals/tcs/KerNO02,DBLP:journals/tcs/KerNO03} and operational approaches such as applicative and normal form bisimulations ~\cite{BIERNACKI+LENGLET+POLESIUK-eta-nf-bisimulation,dallago+gavazzo-Enf-compatible-relation}
(or so-called operational game semantics~\cite{LevyLICSGames,Koutavas-LICS2023}). Here we establish compositionality using a novel technique which bears comparison to game semantics.

Our relational semantics annotates the ordinary relational model in two ways: numerical annotations track the steps needed to reduce a term, while color annotations on the types track the colorings of terms and their potential contexts. An unannotated entry in the relational semantics of a term may have several valid annotations in our system.  \emph{Improvements} result not only from lower numerical annotations, but also from better matching of color annotations between a term and its context. The delicate interplay between colors and counts requires a careful analysis. Once the appropriate preorder has been defined (Definition~\ref{def:PWCpreorder}), we must show that it is \emph{compositional} (Proposition~\ref{prop:leqpwc-contextual}). To prove this property, we treat typings dynamically: we study how modifications to the context colorings of a typing affect the corresponding typings of a term---a process we call \emph{repainting} (Lemma~\ref{l:repaint}). When an improved term is placed in a context, these color changes propagate between the term and its context, in a manner reminiscent of the composition of strategies in game semantics. Once this propagation stabilizes, a new typing emerges, and the preorder is preserved.


\section{The Checkers Calculus}
\label{sec:checkers}

This section is devoted to presenting the \emph{checkers calculus} from \cite{InteractionEquivalence}, a bichromatic variant of the \lam-calculus designed to count the reductions that occur during the interaction between a \lam-term (the \emph{player}) and a context (the \emph{opponent}). 

\subsection{Its Syntax and Operational Semantics}
In the checkers calculus, each application $\cdot^\colr$ and each abstraction $\lam_\colr$ is assigned a color $\colr\in\set{\redclr,\blueclr}$, either white or black, depending on whether the constructor `belongs' to the player or the opponent. 
We use $\colr^\bot$ to denote the opposite color, namely $\redclr^\perp \defeq \blueclr$ and $\blueclr^\perp \defeq\redclr $. 
We consider fixed a countable set $\Var$ of \emph{variables}.

\begin{definition} The set $\Lambdac$ of \emph{checkers terms} is inductively defined as follows:
\[
        \Lambdac\quad\ni \quad \tm,\tmtwo,\tmthree  \grameq \var \mid \lap\colr{\var}\tm  \mid \appp\colrtwo\tm\tmtwo,\qquad \textrm{ for }x\in\Var\textrm{ and }\clr{},\clrtwo{}\in\set{\redclr,\blueclr}.
\]
\end{definition}
We do not color variables, as they can be substituted by arbitrary terms.
As usual, we assume that application associates to the left and has a higher precedence than abstraction.
When the colors are actually specified, we simply write $\rapp\tm\tmtwo$ (and $\bapp\tm\tmtwo$) for $\appp\redclr\tm\tmtwo$ (and $\appp\blueclr\tm\tmtwo$).
In case of many consecutive abstractions or applications, we shorten the notations to $\clavec{\colr}{\var}\tm$ and $\cappvec{\colr}{\tm}{\tmtwo}$, respectively. 
If needed, we expand the former to $\manyclam{k}{\var}\tm$, and the latter to $\appp{\clr{1}\cdots\clr{k}}{\tm}{\tmtwo_1\cdots \tmtwo_{k}}$.

The checkers calculus inherits a number of notions from the usual \lam-calculus. 
We consider checkers terms modulo $\alpha$-conversion.
We denote by $\FV{\tm}$ the set of \emph{free variables} of $\tm$, and by $\tm\isub\var{\tmtwo}$ the capture-free substitution of $\tmtwo$ for all free occurrences of $x$ in $\tm$. The operational semantics of the checkers calculus is defined by taking appropriate context closure of basic rewriting rules. 
Intuitively, a checkers context is a checkers term containing one occurrence of a `hole' $\ctxhole$.

\begin{definition}
\bsub
\item The set $\chcontexts$ of \emph{checkers contexts} and its subset $\chhcontexts$ of \emph{checkers head contexts} are defined by the following grammars:
\[
        \begin{array}{rrllllll}
                        \chcontexts\ni& \ctx & \grameq& \ctxhole \mid \cla{}{\var}\ctx \mid \ccapp{}\ctx\tmtwo \mid \ccapp{}\tm\ctx 
                        \\[1ex]
                        \chhcontexts\ni& \hctx & \grameq & \manyclam{n}\var\appp{\clrtwo{1}\cdots\clrtwo{k}}\ctxhole{\tm_1 \cdots \tm_k}
\\
        \end{array}
\]
\item Given $\ctx\in\chcontexts$ and $\tm\in\Lambdac$, we denote by $\ctxp{{\tm}}$ the checkers term obtained by substituting $\tm$ for $\ctxhole$ in $\ctx$, possibly capturing free variables in $\tm$.
\item The \emph{contextual closure} of a binary relation $\mathcal{R}$ on $\Lambdac$ is the least relation $\chcontexts\ctxholep{\mathcal{R}}$ such that $\tm\ \chcontexts\ctxholep{\mathcal{R}}\ \tmtwo$ entails $\ctxp{{\tm}}\ \chcontexts\ctxholep{\mathcal{R}}\ \ctxp{{\tmtwo}}$, for all $\ctx\in\chcontexts$.
Its \emph{head-contextual closure} $\chhcontexts\ctxholep{\mathcal{R}}$ is defined analogously, by taking $\ctx\in\chhcontexts$.
\esub
\end{definition}
\begin{figure}
\arraycolsep=3.5pt
\tabcolsep=3.5pt
\centering
		\begin{tabular}{c}
		$\begin{array}{r\hcolspace rll}
			\textsc{Silent} & \ccapp{}{(\cla{}\var\tm)}{\tmtwo} &\rtobnoint& \tm\isub\var{\tmtwo}
			\\
			\textsc{Interaction} & \appp{\clrd{}}{(\cla{}\var\tm)}\tmtwo &\rtobint& \tm\isub\var\tmtwo
			\end{array}$
			\\[9pt]
			$\begin{array}{r@{\hspace{.25cm}} rll}
			\textsc{Silent $\beta$} & \tobnoint &\defeq& \chcontexts\ctxholep\rtobnoint
			\\
			\textsc{Interaction $\beta$} & \tobint &\defeq& \chcontexts\ctxholep\rtobint
			\\
			\textsc{Checkers $\beta$} & \tobch &\defeq& \tobnoint \cup \tobint
			\\
			\textsc{Silent head} & \tohnoint &\defeq& \chhcontexts\ctxholep\rtobnoint
			\\
			\textsc{Interaction head} & \tohint &\defeq& \chhcontexts\ctxholep\rtobint
			\\
			\textsc{Checkers head} & \tohch &\defeq& \tohnoint \cup \tohint

		\end{array}$
		\end{tabular}
	\caption{The operational semantics of the checkers calculus.}
	\label{fig:col-def}	
\end{figure}
There are two kinds of colored $\beta$-redexes $\capp{}{(\cla{}\var\tm)}{\tmtwo}$, the \emph{silent} ones and the \emph{interaction} ones. 
Silent redexes are $\beta$-redexes where the color $\clr{}$ of the abstraction matches the color $\clrtwo{}$ of the application. Intuitively, these steps are internal to each player's world. 
In interaction redexes, instead, the color of the abstraction and the color of the application are different, \ie $\clr{}\neq \clrtwo{}$. 
This represents the scenario where the two players interact with each other, which, from each player's perspective, amounts to interacting with the external world. 
When the contracted redex occurs in head position, we say that it is a \emph{head interaction step}. 

\begin{definition} The reductions of the checkers calculus are given in Figure~\ref{fig:col-def}.
Given a reduction $\to_R$, we say that: a checkers term $\tm$ is in \emph{$R$-normal form} ($R$-nf) if $\tm$ does not contain any $R$-redex;
$t$ \emph{has a $R$-nf} if it reduces to a $\tmtwo$ in $R$-nf.
\end{definition}
Despite our earlier intuition of associating black with the player and white with the opponent, notice that the checkers calculus is entirely symmetrical. 


\begin{example}\label{ex:delta} 
Let $\bId \defeq \bla\var\var$ be the black identity and $\comb{D}_\circ \defeq \lam_\circ y.\lam_\circ x.x\circ (y\circ x)$. 
\begin{enumerate}
\item
Their black-application gives rise to a silent step $\bapp{\bId} \comb{D}_\circ \tobnoint \comb{D}_\circ$, while their white-application gives rise to an interaction step, namely $\rapp{\bId} \comb{D}_\circ \tobint \comb{D}_\circ$. 
\item $\bapp{\bapp{\comb{D}_\circ}{\bId}}{\bId} \tohint \bapp{(\lam_\circ x.x\circ (\bId\circ x))}{\bId} \tobint \bapp{(\lam_\circ x.x\circ  x)}{\bId}\tohint \rapp{\bId}{\bId}\tohint\bId$.
\item A monochromatic checkers term displays the same behavior as the underlying term: $ \comb{D}_\circ\circ \comb{D}_\circ\tobnoint \lam_\circ x.x\circ (\comb{D}_\circ\circ x)\tobnoint\lam_\circ x.x\circ (\lam_\circ z.z\circ (x\circ z))$ and $\bId\bullet\bId\tobnoint\bId$.
\end{enumerate}
\end{example}
There are in fact two ``copies'' of the set $\Lambda$ of \lam-terms within $\Lambdac$: one obtained by painting all \lam-terms black, and the other by painting them white. Formally:

\begin{definition}\label{def:paintinglambdas} 
        Given a color $\colr\in\{\redclr,\blueclr\}$, the \emph{$\colr$-painting of ordinary $\l$-terms}, is the mapping $\monoToPlayer{\colr}{\cdot} : \Lambda \rightarrow \Lambdac$ defined by induction as follows:
        \[
        \monoToPlayer{\colr}{\var} \defeq \var, 
        \qquad
        \monoToPlayer{\colr}{\la\var\tm}  \defeq  \lap\colr\var\monoToPlayer{\colr}\tm, 
        \qquad
        \monoToPlayer{\colr}{\tm\tmtwo} \defeq  \appp{\colr}{\monoToPlayer{\colr}{\tm}}{\monoToPlayer{\colr}{\tmtwo}}.
        \]
        Given a checkers term $\tm\in\Lambdac$, the \emph{color washing} map $\wash{\cdot} :  \Lambdac\to\Lambda$ is defined by:
        \[
        \wash{\var} \defeq \var, 
        \qquad
        \wash{\lap\colr\var\tm}  \defeq  \la\var\wash\tm,       
        \qquad
        \wash{\appp{\colr}\tm\tmtwo} \defeq  \wash{\tm}\cdot \wash{\tmtwo}.
        \]
\end{definition}
Note that $\beta$-reductions $\to_\beta$ (resp.\ head reductions $\Rew{\relation{h}}$) on ordinary \lam-terms correspond to silent (head) reductions.

\begin{lemma}[Correspondence {\cite[Prop.\,3.13-14]{InteractionEquivalence}}]
        Let $\tm\in\Lambdac$ and $\tmtwo,\tmtwop\in\Lambda$.
        $\label{l:correspondence}$
        \bsub
        
        \item If $\tmtwo\Rew{\beta}\tmtwop$, then $\monoToPlayer{\colr}{\tmtwo}\Rew{\beta_\tau}\monoToPlayer{\colr}{\tmtwop}$. Similarly, if $\tmtwo\Rew{\relation{h}}\tmtwop$, then $\monoToPlayer{\colr}{\tmtwo}\Rew{\relation{h}_\tau}\monoToPlayer{\colr}{\tmtwop}$.
        \item If $\wash\tm\Rew{\relation R}\tmtwo$ then $\exists\tmp\in\Lambdac$ such that $\wash{\tmp}=\tmtwo$ and $\tm\Rew{\relation R_{\bullet\circ}}\tmp$ for $\relation{R}\in\{\beta,\hsym\}$.
        \esub 
\end{lemma}
So, \emph{checkers head normal forms} ($\hchsym$-nfs) have the shape $\htm = \clavec\colr\var\appp{\clrtwo{1}\cdots\clrtwo{k}}\vartwo{\tm_1 \cdots \tm_k}$.

\begin{theorem}[Confluence~\cite{InteractionEquivalence}]
Reductions $\tobch$, $\tobnoint$, and $\tobint$ are confluent.
\end{theorem}

\begin{remark}\label{rem:eta}
The checkers calculus has no analogue of the $\eta$-reduction $\lam x.\tm x \to_\eta \tm$, with $x\notin\FV{\tm}$.
Consider the candidate black $\eta$-expansion  $\bOne \defeq\bla\var\bla\vartwo{\bapp\var\vartwo}$ of the black identity $\bId$.
Then, assuming $\bOne \to_{\blackcol\eta}\bId$, would break confluence: 
$\rapp{\rapp{\bOne}\varthree}\varfour\to_{\hintsym} \rapp{(\bla\vartwo\bapp\varthree\vartwo)}\varfour\to_{\hintsym}\bapp\varthree\varfour$, while
$\rapp{\rapp{\bOne}\varthree}\varfour\to_{\blackcol\eta} \rapp{\rapp{\bId}\varthree}\varfour\to_{\hintsym} \rapp\varthree\varfour$.
\end{remark}

\subsection{Interaction Improvement}

Observational preorders have been introduced to express that one program is more defined than another: when placed in any context, if the former terminates, then so does the latter. By varying the notion of termination being observed, one obtains different preorders that nonetheless share the same qualitative nature.
To achieve a \emph{quantitative} notion of observation, one can compare two programs by considering the number of steps required for termination (see Sands \cite{SandsImprovementTheory,SandsToplas,DBLP:journals/tcs/Sands96}). 
However, this approach is often too fine-grained, since it also distinguishes terms that represent the same program at different stages of evaluation.
In the checkers calculus, the introduction of colors makes it possible to compare programs by counting only the interaction steps between the program and its context, ignoring silent steps. In this paper, as in~\cite{InteractionEquivalence}, we adopt head-reduction as our operational semantics, leading to the following definitions.

Given $\tm\in\Lambdac$, we write $\bshcol{k}$ whenever $\tm$ has a $\hchsym$-nf $h$, and the reduction sequence $\tm\tohch\cdots\tohch h$ contains $k$ interaction head steps $\tohint$. Since $\tohch$ is a deterministic strategy, such a sequence is unique and so is the number $k$. 
Silent head steps are ignored in the calculation of $k$, as expected. 

\begin{definition}[Checkers interaction improvement and equivalence \cite{InteractionEquivalence}]\label{def:checkersinteractionpreorder}\
         On checkers terms $\tm,\tmtwo\in\Lambdac$, we define the following relations:
        \bsub
        \item \emph{Interaction preorder $\leqchcol$}.\\ 
                $\tm \leqchcol \tmtwo $ if $\forall \ctx\in\chcontexts,k\in\nat \,.\,[\,\ctxp{{\tm}}\bshcol{k} \ \Rightarrow\ \ctxp{\tmtwo}\bshcol{k}]$.
        \item \emph{Interaction improvement (preorder) $\leqchcolr$}.\\
                $\tm \leqchcolr \tmtwo$ if $\forall \ctx\in\chcontexts,k\in\nat \,.\,[\,\ctxp{{\tm}}\bshcol{k}  \Rightarrow \ctxp{\tmtwo}\bshcol{k'}\textrm{ for some } k' \leq k]$;
        \item \emph{Interaction equivalence $\equivchcol$}. It is the equivalence relation induced by $\leqchcol$:\\ 
                ${\tm} \equivchcol {\tmtwo}$ if ${\tm} \leqchcol {\tmtwo}$ and ${\tmtwo} \leqchcol {\tm}$.
        \esub
\end{definition}
It is easy to see that $\leqchcol\,\subsetneq\, \leqchcolr$. The strictness of the inclusion is due to the fact that $\rapp\bId\rId \leqchcolr\rId :=\lam_\circ x.x$, but $\rapp\bId\rId \not\leqchcol\rId$, as both checkers terms have $\rId$ as $\hchsym$-nf, but $\rapp\bId\rId$ requires an additional interaction step to reach it.
Observe that the equivalence relation induced by $\leqchcolr$ coincides with $\equivchcol$.

\begin{example}\label{ex:preorders} 
Consider $\bId$ and $\comb{D}_\circ$ from Example~\ref{ex:delta}, and $\bOne$ from Remark~\ref{rem:eta}.
\begin{itemize}
        \item Since $\comb{\Omega} := (\lam x.xx)(\lam y.yy)$ has no hnf, any of its colorings will be a bottom element w.r.t.\ $\leqchcol$ and $\leqchcolr$.
        E.g. $\monoToPlayer{\bullet}{\comb{\Omega}}\leqchcol\bId$ and $\monoToPlayer{\circ}{\comb{\Omega}}\leqchcol\bOne$.
        \item More generally, $\tm\equivchcol \tmtwo$ holds whenever $\tm,\tmtwo\in\Lambdac$ have no $\hchsym$-nf.
        \item To see that $\lam_\circ x.x\not\leqchcolr \lam_\bullet x.x$, just take the context $C := \ctxhole\bullet\bId$. Indeed, $\ctxp{\lam_\circ x.x}\bshcol{1}$ while $\ctxp{\lam_\bullet x.x}\bshcol{0}$ because the head step is silent.
        \item\label{ex:whitedelta2} $\bOne\not\leqchcol\bId$, as they are separated by $C = \ctxhole\circ\varthree\circ\varfour$, but $\bOne\leqchcolr\bId$ holds. This is a consequence of our main Theorem~\ref{thm:main}.
        \item Recall that $\comb{Y}\in\Lambda$ is a fixed point combinator (fpc) if $\comb{Y} =_\beta \lam f.f(\comb{Y}f)$ and that all fpcs $\comb{Y},\comb{Y'}$ share the same B\"ohm tree. 
        By \cite[Theorem 8.6]{InteractionEquivalence}, $\monoToPlayer{\circ}{\comb{Y}} \equivchcol \monoToPlayer{\circ}{\comb{Y'}}$.
        \item The combinator $\comb{J}\in\Lambda$ from \cite{Wadsworth76} satisfying $\comb{J}x =_\beta \lam z.x(\comb{J}z)$ is an ``infinite $\eta$-expansion'' of $\lam x.x$. To check that $\monoToPlayer{\circ}{\comb{J}}\not\leqchcol \lam_\circ x.x$, take once again $C := \ctxhole\bullet\bId$. The fact that $\monoToPlayer{\circ}{\comb{J}}\leqchcolr \lam_\circ x.x$ follows from Theorem~\ref{thm:main} (by $\circ/\bullet$-symmetry).
\end{itemize}
\end{example}
In this paper, we focus on the interaction improvement $\leqchcolr$, which was introduced in \cite{InteractionEquivalence}, but not previously characterized either by relational semantics or in terms of a tree-like ordering.




\section{Checkers Relational Improvement Semantics}\label{sec:types}
\begin{figure}[t!]	
\resizebox{\textwidth}{!}{%
\arraycolsep=3pt
	\begin{tabular}{c\hcolspace c\hcolspace | \hcolspace c}
	\infer[\typingruleAx]{\var \hastype [\ltype] \ctypes{0} \var \hastype \ltype}{}
	&	
	\infer[\typingruleMany]{\uplus_{i\in I}\typectx_i \ctypes{\sum_{i\in I} k_i}  \tm \hastype [\ltype_i]_{i\in I}}{(\typectx_i \ctypes {k_i} \tm \hastype \ltype_i)_{i\in I}  & I~ \text{finite}}
	 &
	\infer[\typingruleAppNoInt]{\typectx \uplus \typectxtwo \ctypes {k_1+k_2} \ccapp{}\tm\tmtwo \hastype \ltype}{ \typectx \ctypes {k_1} \tm \hastype \mtype \typearrowp{\colr} \ltype & \typectxtwo \ctypes{k_2} \tmtwo \hastype \mtype  }
	
	\\[5pt]
	\infer[\typingruleAbs]{\typectx \ctypes k \cla{}\var\tm \hastype \mtype \typearrowp{\colr} \ltype}{\typectx, \var \hastype \mtype \ctypes k \tm \hastype\ltype}
	 &
	\infer[\typingruleApp]{\typectx \uplus \typectxtwo \ctypes {k} \capp{}\tm\tmtwo \hastype \ltype}{ \typectx \ctypes {k_1} \tm \hastype \mtype \typearrowp{\colr}\ltype & \typectxtwo \ctypes{k_2} \tmtwo \hastype \mtype  }
		 &
	\infer[\typingruleAppInt]{\typectx \uplus \typectxtwo \ctypes {k_1+k_2+1} \appp{\colr^\bot}\tm\tmtwo \hastype \ltype}{ \typectx \ctypes {k_1} \tm \hastype \mtype \typearrowp{\colr} \ltype & \typectxtwo \ctypes{k_2} \tmtwo \hastype \mtype  }
	\end{tabular}
}\medskip

\footnotesize
where in the rule $\typingruleApp$ we have $k=k_1 {+} k_2$ when $\colr=\colrtwo$, and $k = k_1 {+} k_2 {+} 1$ otherwise.
Note that the rule $\typingruleApp$ compactly represents the rules $\typingruleAppNoInt$ and $\typingruleAppInt$.

\caption{Checkers multi type system $\ctypes{}$.}
\label{fig:colored-types-for-cbn}
\end{figure}
This section is devoted to presenting a denotational semantics of the checkers calculus, which can be seen as an annotated version of the relational semantics of the \lam-calculus~\cite{BucciarelliEM07,BreuvartMR18}.
Since the work of \cite{DBLP:journals/mscs/Carvalho18}, it has been known that---just as filter models can be presented as intersection type systems---relational models can be presented via multitype systems.
Thanks to the absence of weakening and contraction, and the presence of multisets of types $[\ltype_1,\dots,\ltype_n]$, these type systems are resource-aware: intuitively, a \lam-term $\lam x.\tm$ having type $[\ltype_1,\dots,\ltype_n]\to\ltypetwo$ needs to consume $n$ copies of its argument to produce a result of type $\ltypetwo$. During evaluation, the argument is used once with type $\ltype_1$, once with type $\ltype_2$, etc. 
This makes it possible to infer intensional properties of a program, like the amount of steps needed to reach its head normal form, by examining its type derivations.

\subsection{Checkers Type System} We present a multi type system in which arrows are decorated with a color $\clr{}\in\set{\circ,\bullet}$, intuitively matching the color of the outer \lam-abstraction (of its head normal form),
as in $\cla{}\var\tm \hastype \mtype\typearrowp{\colr} \ltype$.
\begin{definition} 
        \begin{enumerate}
        \item \emph{Linear types} $\ltype,\ltypetwo$ and \emph{multi types} $\mtype,\mtypetwo$ are generated by:
        \begin{center}
        $
        \begin{array}{lrll @{\hspace{15pt}} l }
        \textsc{\em Linear types} & \ltype, \ltypetwo &\grameq& \vartype \mid \mtype \typearrowp{\colr} \ltype & \textrm{where $\colr\in\{\redclr,\blueclr\}$ and $\vartype$ is an atom}
        \\
        \textsc{\em Multi types} & \mtype, \mtypetwo &\grameq& \multitype{n}{\ltype} & n\geq 0\\
        \end{array}
        $
        \end{center}
                
We use $\gtype, \gtypetwo$ as meta-variables denoting either linear or multi types.
        \item Given multi types $\mtype,\mtypetwo$, we denote by $\mtype+\mtypetwo$ their multiset union, and by $\zero$ the empty multi type.
                In other words, $\zero$ denotes the neutral element of $+$.
        \item A \emph{type environment} $\typectx$ is a map from $\Var$ to multi types having a finite \emph{support}
        $
                \mathrm{supp}(\typectx) := \set{x\in\Var \mid \typectx(x)\neq\zero}
        $.
        We let $x_1:\mtype_1,\dots,x_n:\mtype_n$ denote the environment $\Gamma(y) = \mtype_i$ if $y = x_i$, otherwise $\Gamma(y) = \zero$.
\item The empty environment $\Gamma(x) = \zero$, for all $x\in\Var$, is denoted $\emptyset$ or omitted.
        \item \emph{Typing judgements} are quadruples denoted by $\typectx\ctypes{k} \tm : \gtype$, where $k\in\nat$, and can be derived by applying the rules
                given in Figure~\ref{fig:colored-types-for-cbn}.
        \end{enumerate}
      \end{definition}

      \begin{remark}
        The relational semantics of the \lam-calculus can be recovered from our system by erasing all color and index annotations, or embedded in our system by coloring all term-constructors and types black and setting all indices to zero. 
      \end{remark}
We shall prove that a checkers term is typable $\typectx\ctypes{k} \tm : \ltype$ exactly when it is head normalizable, and the index $k$ gives an upper bound on the number of head interaction steps from $\tm$ to its $\hchsym$-nf  (Theorem~\ref{thm:soundandcomplete}(1)).
This explains the rule $\typingruleApp$ where a checkers term $\tm \hastype \mtype \typearrowp{\colr}\ltype$ can always be $\colrtwo$-applied to a term $\tmtwo\hastype\mtype$, but the index $k$ associated with $\capp{}\tm\tmtwo \hastype \ltype$ must be incremented whenever $\colr \neq \colrtwo$.

\begin{remark} The type system $\ctypes k$ presented in \reffig{colored-types-for-cbn} is strongly related to the one used in \cite{InteractionEquivalence}.
In the latter, a second color $\colrtwo$ is annotated on the arrow type $\typearrowp{\colr\colrtwo}$, and in the rule $\typingruleApp$ such color needs to match the color of the application $\cdot^{\colrtwo}$. The second color is needed to have a notion of \emph{tight type}, but otherwise redundant. 
\end{remark}
\begin{samepage}
  We now present the main properties of our system.
  As in the relational case \cite{DBLP:journals/mscs/Carvalho18}, it enjoys \emph{quantitative} versions of subject reduction (SR) and expansion (SE). 

  \begin{definition}[Applicative size]
    Given a derivation $\tderiv$ of $\typectx\ctypes{k} \tm : \gtype$, in symbols $\tderiv\exder \typctx \ctypes k \tm\hastype \gtype$,
    its \emph{applicative size} $\insize{\tderiv}$ is the number of $\ruleAp$-rules in $\tderiv$.
  \end{definition}
  The resource awareness of the type system allows us to prove that the applicative size of $\tderiv\exder \typctx \ctypes k \tm\hastype \gtype$ decreases along each head step $\tohch$. Moreover, if the step is an interaction one, then the index $k$ decreases by exactly 1; it is stable otherwise.
\end{samepage}
\begin{toappendix}
        \begin{proposition}$\label{prop:ch-subject}$
        Let $\tm,\tm'\in\Lambdac$ be such that $\tm \tohch \tm'$.
        $\NoteProof{prop:ch-subject}$
        \begin{enumerate}
                \item 
                \emph{Quantitative subject reduction}:  
                        if $\tderiv\exder \typctx \ctypes k \tm\hastype \ltype$ then there is a derivation $\tderiv'\exder \typctx \ctypes {k'} \tm' \hastype \ltype$ such that $\insize{\tderiv'} = \insize{\tderiv} -1$. Moreover:
                \bsub
                \item if $\tm \tohint \tm'$ then $k'=k-1$;
                \item otherwise, if $\tm \tohnoint \tm'$ then $k'=k$.
                \esub
                \item 
                \emph{Subject expansion}: if $\tderiv'\exder \typctx \ctypes{k'} \tm' \hastype \ltype$ then there is a derivation $\tderiv\exder \typctx \ctypes{k} \tm\hastype \ltype$ for some $k$.
        \end{enumerate}
\end{proposition}
\end{toappendix}
The quantitative subject reduction entails the soundness of the checkers type system, and the typability of hnfs gives its completeness via subject expansion.

\begin{toappendix}
\begin{theorem}[Typability characterizes head normalization]$\label{thm:soundandcomplete}$
        Let $\tm\in\Lambdac$. 
        $\NoteProof{thm:soundandcomplete}$
        \begin{enumerate}
                \item
                \emph{Soundness}: if $\tderiv \exder \typctx \ctypes{k} \tm\hastype\ltype$ then there exists $k'\le k$ such that $\tm\bshcol{k'}$.
                \item
                \emph{Completeness}: if $\tm \bshcol{k}$        then $\typctx \ctypes{k} \tm\hastype\ltype$ is derivable for some $\typctx,\ltype$.
        \end{enumerate}
\end{theorem}
\end{toappendix}

\subsection{Checkers Type Interpretation and Preorders} 

As previously mentioned, multi type systems are often employed to present relational models of the \lam-calculus~\cite{DBLP:journals/mscs/Carvalho18,DBLP:journals/mscs/PaoliniPR17,BreuvartMR18} since the relational interpretation of a \lam-term is isomorphic to the set of its `typings' $(\typctx,\ltype)$. Analogously, we can define the colored interpretation of a checkers term $\tm$ as the set of its typings that in this setting are triples $(\typctx,\ltype,k)$ such that $\tderiv \exder \typctx \ctypes{k} \tm\hastype\ltype$.

\begin{definition}
The \emph{colored interpretation} of a checkers term $\tm\in\Lambdac$ is given by:
\[
        \semint \tm \defeq \{ \Tr{\typctx}{\ltype}{k} \mid 
        \exists 
        \, 
        \tderiv\exder \typctx \ctypes{k} \tm : \ltype \}.
\]
Two terms $\tm,\tmtwo$ are \emph{type equivalent} if they have the same colored interpretation.
\end{definition}

\begin{remark} By Proposition~\ref{prop:ch-subject}, the interpretation $\semint{\cdot}$ is invariant along silent head reductions. 
However, since the index $k$ is taken into account in the interpretation, we may have $\tm \tohint \tmp$ with $\semint{\tm}\neq\semint{\tmp}$. 
As an example, we have $\bId\circ\ x \tohint x$ and $\semint{\bId\circ\ x}\cap\semint{x}=\emptyset$, \eg $(x\hastype[\vartype],\vartype,0)\in\semint{x}-\semint{\bId\circ\ x}$.
\end{remark}
There are several possible ways of comparing these semantic interpretations. 
The most natural one is set theoretical inclusion $\semint\tm \,\subseteq \, \semint\tmtwo$, which leads to the checkers type preorder studied in \cite{InteractionEquivalence}. 
In this paper we want to capture the interaction improvement $\leqchcolr$, so we  need to refine the comparison.
A naive attempt would be to check that $\tm$ and $\tmtwo$ have the same types in the same environments, but allowing the index $k'$ associated with $\tmtwo$ to be smaller. Formally:
\begin{center}  
$\tm \leq \tmtwo$ if $\forall\typctx,~\ltype,~k,~ \typctx\ctypes k \tm \hastype \ltype \implies \exists k' \leq k,~ \typctx\ctypes{k'} \tmtwo \hastype \ltype$
\end{center}
Unfortunately, this definition does not work, since $\leqchcolr$ validates $\eta$-reduction on monochromatic checkers terms (cf.\ Example~\ref{ex:preorders}), while this comparison suffers from the same issue as the inclusion.
Consider the black $\eta$-expansion of $\var$:
\pagebreak

\begin{center}
        $\infer[\typingruleAbs]{
                \var\hastype [\emptytype \typearrowp{\redclr} \ltype] \ctypes{1} \bla\vartwo \bapp{\var}{\vartwo}\hastype \emptytype \typearrowp{\blueclr} \ltype 
        }{
                \infer[\typingruleAppInt]{
                        \var\hastype [\emptytype \typearrowp{\redclr} \ltype] \ctypes{1} \bapp{\var}{\vartwo}\hastype \ltype 
                }{ 
                        \infer[\typingruleAx]{
                                \var\hastype [\emptytype \typearrowp{\redclr} \ltype] \ctypes{0} \var \hastype \emptytype \typearrowp{\redclr} \ltype
                        }{}
                        ~~~~&~~~~
                        \infer[\typingruleMany]{
                                \ctypes{0} \vartwo \hastype \emptytype
                        }{}
                }
        }$
\end{center}
and note that $\bla\vartwo\bapp\var\vartwo \leqchcolr \var$ holds, whereas $\var \hastype [\emptytype \typearrowp{\redclr} \ltype] \ctypes{~k'} \var \hastype \emptytype \typearrowp{\blueclr} \ltype$ cannot hold, because of the mismatch between the colors. 
Indeed, the rule $\typingruleAx$ requires that both occurrences of $\emptytype \to \ltype$ share the same color.

\paragraph{Whitening Types.} The above example suggests that, if we wish to define a preorder $\leqpwc$ that coincides with $\leqchcolr$ on the interpretations of black terms, we need a comparison between typings that can reduce the cost of derivations and whiten certain arrows in the types.
Let us analyze the example further:
\[
        \begin{array}{lcl}
        (x:[\mtype\typearrowp{\colr}\ltype],\mtype\typearrowp{\colrtwo}\ltype,k)\in
        \semint{\bla\vartwo \bapp{\var}{\vartwo}} 
        \end{array}
\]
iff $\colrtwo = \blueclr$ (forced by the $\lam_\bullet y$) and either $\colr = \bullet$ and $k = 0$, or $\colr = \circ$ and $k = 1$. 
A triple of this kind\footnote{
Note that $(x\hastype [\vartype],\vartype,0)\in\semint{x}-\semint{\bla\vartwo \bapp{\var}{\vartwo}}$. 
This is expected as $\var \not\leqchcolr \bla\vartwo \bapp{\var}{\vartwo}$.
} belongs to the interpretation of $x$ whenever $\colr=\colrtwo$ and $k=0$. 
In particular, there is a whiter and cheaper typing $(x:[\mtype\typearrowp{\redclr}\ltype],\mtype\typearrowp{\redclr}\ltype,0)\in\semint{x}$.

This can be generalized. We shall prove (Lemma~\ref{lem:whitercheaper}) that whenever $\typctx\ctypes{k} \tm\hastype \ltype $ and $\tm$ is a black $\eta$-expansion of $\tmtwo$, we can derive $\typctx'\ctypes{k'} \tmtwo\hastype \ltype'$ where $(\typctx',\ltype')$ is a whiter version of $(\typctx, \ltype)$ and $k' \leq k$, yielding a derivation that is both whiter and cheaper.
Moreover, we show that whitening must be applied only to arrows of positive polarity.
We begin by defining a polarized whitening relation on types.

        \begin{figure}
        \[
        \infer{\vartype\whiterpol{\posorneg}{0}\vartype}{\posorneg\in\set{+,-}}\qquad
        \infer{\mtype'\typearrowp{\circ}\ltype'\whiterposp{k_1 + k_2 + 1} \mtype\typearrowp{\bullet}\ltype}{\mtype'\whiternegp{k_1} \mtype & \ltype'\whiterposp{k_2}\ltype }
        \qquad
        \infer{\mtype'\typearrowp{\colr}\ltype'\whiterpol{\posorneg}{k_1 + k_2 } \mtype\typearrowp{\colr}{}\ltype}{\mtype'\whiterpol{\posornegop}{k_1} \mtype & \ltype'\whiterpol{\posorneg}{k_2}\ltype}
        \]
        \[              
        \infer{[\ltype'_1,\dots,\ltype'_n] \whiterpol{\posorneg}{k_1+\cdots+k_n} [\ltype_1,\dots,\ltype_n]}{\ltype'_1\whiterpol{\posorneg}{k_1}\ltype_1&\cdots&\ltype'_n\whiterpol{\posorneg}{k_n}\ltype_n}
        \]      
        \caption{Polarized whitening of a type.}\label{fig:polwhite}
        \end{figure}

\begin{definition}[Polarized Whitening]
\bsub
\item Given a \emph{polarity} $\posorneg\in\set{+,-}$ we denote by $\posornegop$ the opposite polarity.
\item
        For all $k\in\nat$, define $\gtype' \whiterposp k \gtype$ and $\gtype' \whiternegp k \gtype$ by mutual induction in \reffig{polwhite}.
\item 
        When $\gtype' \whiterposp k \gtype$ (resp.\ $\gtype' \whiternegp k \gtype$) holds we say that $\gtype'$ is \emph{$k$-whiter} than $\gtype$ on positively (resp.\ negatively) occurring arrows. 
\item   
        The relations above extend to environments $\Gamma$ and pairs $\Pair{\Gamma}{\ltype}$ as follows:
        \[
        \infer{\emptyset\whiterpol{\posorneg}{0} \emptyset}{}
        \qquad
        \infer{\typctx',x:\mtype' \whiterpol{\posorneg}{k_1+k_2} \typctx,x:\mtype}{\typctx' \whiterpol{\posorneg}{k_1} \typctx & \mtype'\whiterpol{\posorneg}{k_2} \mtype}
        \qquad
        \infer{\Pair{\typctx'}{\ltype'} \whiterpol{\posorneg}{k_1+k_2} \Pair{\typctx}{\ltype}}
        {\typctx' \whiterpol{\posornegop}{k_1}\typctx & \ltype' \whiterpol{\posorneg}{k_2} \ltype}
        \]  
        \esub
\end{definition}
Intuitively, $\gtype'\whiterpol{\posorneg}{k}\gtype$ holds if $\gtype'$ is obtained from $\gtype$ by whitening $k$ arrows occurring with polarity $\posorneg$.
In particular, the underlying uncolored type must be the same.
Note the absence, in \reffig{polwhite}, of a rule that allows one to infer $\mtype' \typearrowp{\circ} \ltype' \whiternegp{k} \mtype \typearrowp{\bullet} \ltype$.
This omission is intentional, as it prevents the whitening of arrows that occur negatively in the type that would break the polarization of the change of color. 

Everything is now in place to introduce the preorder $\leqpwc$ on checkers terms.

\begin{definition}[Polarized Whiter-Cheaper Improvement]\label{def:PWCpreorder}
        For all checkers terms $\tm$, $\tmtwo$ we define $\tm\leqpwc \tmtwo$ if and only if
        \[
        \forall \Tr{\typctx}{\ltype}{k}\in\semint{\tm},\,
        \exists \Tr{\typctx'}{\ltype'}{k'}\in\semint{\tmtwo}\,, \exists d,\,
        \Pair{\typctx'}{\ltype'} \whiterposp{d} \Pair{\typctx}{\ltype} \textrm{ such that }
        k \geq k'+d 
        \]
\end{definition}
Summing up, this means that for all derivations of $\typctx\ctypes{k}\tm\hastype\ltype$, there exists one of $\typctx'\ctypes{k'}\tmtwo\hastype\ltype'$ which is whiter, and cheaper by at least the amount of whitenings.

We now investigate some properties of the whitening relations.

\begin{toappendix}
\begin{lemma}\label{lem:whitercheaperrel}\NoteProof{lem:whitercheaperrel}
        For all environments $\typctx,\typctx'$ and linear types $\ltype,\ltype'$, we have:
        \begin{enumerate}[(i)]
                
                \item \emph{$0$-whitening is equality:} for $\posorneg\in\{\possym,\negsym\}$, we have that:
                \begin{itemize}
                        \item $\ltype\whiterpp{\posorneg}{0}\ltype$ if and only if $\ltype = \ltypetwo$;
                        \item $\typctx\whiterpp{\posorneg}{0}\typctx'$ if and only if $\typctx=\typctx'$;
                        \item Therefore, $\Pair{\typctx}{\ltype}\whiterpp{\posorneg}{0}\Pair{\typctx'}{\ltype'}$ if and only if
                        $\typctx=\typctx'$ and $\ltype = \ltypetwo$.
                \end{itemize}
                \item \emph{Inversion:}  for $\posorneg\in\{\possym,\negsym\},\colr\in\{\redclr,\blueclr\}$ and $k\in\nat$, we have that:
                \begin{center}
                        $\Pair{\typctx',x:\mtype'}{\ltype'} \whiterpp\posorneg{k} \Pair{\typctx,x:\mtype}{\ltype}$ if and only if $\Pair{\typctx'}{{\mtype'}\typearrowp\colr{\ltype'}} \whiterpp\posorneg{k} \Pair{\typctx}{\mtype\typearrowp\colr\ltype}$.
                \end{center}
                \item \emph{Transitivity of polarized whitening:}  for $\posorneg\in\{\possym,\negsym\}$, we have that:
                \begin{center}
                        if $\Pair{\typctx}{\ltype} \whiterpp\posorneg{k_1}\Pair{\typctx'}{\ltype'} \whiterpp\posorneg{k_2}\Pair{\typctx''}{\ltype''}$
                        then $\Pair{\typctx}{\ltype} \whiterpp\posorneg{k_1+k_2} \Pair{\typctx''}{\ltype''}$.
                \end{center}
        \end{enumerate}
\end{lemma}
\end{toappendix}


\section{Checkers Semantics for the \lam-calculus}
\label{sect:bohm-trees}

In the introduction, we argued that the checkers calculus is interesting in itself, but in this paper we mainly use it to infer properties of the standard $\lambda$-calculus.

As discussed in Section~\ref{sec:checkers}, ordinary $\lambda$-terms can be embedded into the checkers calculus via the coloring map $\monoToBlue{(\cdot)} : \Lambda \to \Lambdac$ from Definition~\ref{def:paintinglambdas}, and then compared using the various preorders we defined on checkers terms.

\begin{definition}[Preorders on \lam-terms]\label{def:interactionpreorder}
        For all \lam-terms $\tm,\tmtwo\in\Lambda$, define:
        \bsub
        \item \emph{Interaction preorder $(\leqctx)$}.
        $\tm\leqctx \tmtwo$ iff $\monoToBlue{\tm}\leqchcol \monoToBlue{\tmtwo}$;
        \item \emph{Interaction improvement ($\leqimp$)}.
        $\tm\leqimp \tmtwo$ iff $\monoToBlue{\tm}\leqchcolr \monoToBlue{\tmtwo}$;f
        \item \emph{Polarized whiter-cheaper preorder $(\leqpwcplain)$.} $\tm\lecol \tmtwo$ iff $\monoToBlue{\tm}\lecolpos \monoToBlue{\tmtwo}$.
        \esub
\end{definition}
To help the reader avoid any confusion, we reserve the subscript `$\intsym$' for relations between checkers terms, while it is omitted in relations between $\lambda$-terms.
However, be aware that the quantification in $\leqctx$ and $\leqimp$ still ranges over all checkers contexts, and the interaction steps are being counted.

The preorder $\leqctx$ has been characterized as B\"ohm tree inclusion~\cite{InteractionEquivalence}, and semantically corresponds to the preorder induced by Plotkin’s model $\mathcal{P}\omega$~\cite{DBLP:journals/siamcomp/Plotkin76}.
We have seen that, although there is no analogue of $\eta$-reduction in the checkers calculus (Remark~\ref{rem:eta}), the inequation $\bOne \leqchcolr \bId$ holds (Example~\ref{ex:preorders}), which entails $\lambda y.\tm y \leqimp \tm$ for $\lambda$-terms (with $y \notin \fv{\tm}$).
In other words, $\leqimp$ validates $\eta$-reduction in the sense that $\tm \to_\eta \tmtwo$ implies $\tm \leqimp \tmtwo$.

We shall see that the situation is even subtler: $\leqimp$ also validates possibly infinite $\eta$-reductions. 
This notion has been formalized by Lassen~\cite{lassen1999bisimulation} as follows. 
We write $\tm\bshs h$ to mean that $\tm$ is head normalizable and $h$ is its hnf.

\begin{definition}[\bohm tree preorder up to $\eta$-reductions]
        The \emph{\bohm preorder up to $\eta$-reductions} $\etaredbtleq$ is defined coinductively on \lam-terms $\tm,\tmtwo$ as the largest relation $\tm \etaredbtleq \tmtwo$ closed under the following clauses:
        \begin{itemize}
                \item[(bot)] $\tm\bshdiv$  \ie $\tm$ has no head normal form.
                \item[($\relation{H}\eta_\symfont{red}$)] $t\bshs \lam x_1\dots x_{n+p}.y\,t_1\cdots t_{k+p}$ and  $u\bshs  \lam x_1\dots x_n.y\,u_1\cdots u_k $, 
                such that $p\ge 0$,\linebreak $y$ is equally bound or free in both terms, 
                $(\tm_i\etaredbtleq u_i)_{i=1}^{k}$, $(\tm_i\etaredbtleq x_i)_{i =k+1}^{k+p}$ and $(x_i)_{i =k+1}^{k+p}$ are chosen by $\alpha$-conversion not to occur free in $yu_1\cdots u_k$.
        \end{itemize}
        The \emph{\bohm preorder up to possibly infinite $\eta$-conversion} $\sqsubseteq_{\laxbsym}$ is defined similarly, by adding a clause catching $\eta$-expansions in $\tmtwo$, thus symmetric to ($\relation{H}\eta_\symfont{red}$).     
\end{definition}
The definition of $\tm \etaredbtleq \tmtwo$ requires a certain level of complexity to take into account possibly infinite $\eta$-expansions occurring in the B\"ohm tree of $\tm$.
Breuvart et al.~\cite{BreuvartMR18} characterized $\etaredbtleq$ as the preorder induced by the relational graph model $\mathcal{E}$~\cite{DBLP:journals/entcs/HylandNPR06}, Ronchi Della Rocca as the one induced by the filter model $\mathcal{D}_{\textrm{BCD}}$~\cite{DBLP:journals/iandc/Rocca82}.

\begin{example}\label{ex:playingwithJ} 
The $\comb{J} := \comb{Y}(\lam zxy.x(zy))$ from Example~\ref{ex:preorders} produces no $\eta$-redexes:
\[
        \comb{J} =_\beta \lam xy_0.x(\comb{J}y_0) =_\beta \lam xy_0.x(\lam y_1.x\comb{J}y_1) =_\beta  \lam xy_0.x(\lam y_1.x(\lam y_2.y_1(\comb{J}y_2)) =_\beta \cdots
\]
but it is an infinite $\eta$-expansion of $\comb{I}:=\lam x.x$, and in fact $\comb{J}\etaredbtleq\comb{I}$. This example clarifies that the freshness condition in ($\relation{H}\eta_\symfont{red}$) does not need to hold at any finite step, but holds at the limit. 
By contextuality, we get $\lam y.x(\lam z.y(\comb{J}z)) \etaredbtleq \lam y.xy$, and it is interesting to compare their interpretations once painted black. E.g., 
\[\newcommand{\typingruleManyShort}{\mathsf{m}}
\infer[\typingruleAbs]{
                                \var\hastype [[\emptytype \typearrowp{\blueclr} \ltypetwo]  \typearrowp{\colrtwo} \ltype] 
                                \ctypes{(\Kr{\clr{}}{\bullet}+\Kr{\clrtwo{}}{\bullet})} 
                                \bla\vartwo \bapp{\var}{(\bla\varthree\bapp\vartwo(\monoToBlue{\comb{J}\varthree}))}
                                \hastype [\emptytype \typearrowp{\colr} \ltypetwo]  \typearrowp{\blueclr} \ltype 
                        }{
                                \infer[\typingruleAppInt]{
                                        \var\hastype [[\emptytype \typearrowp{\blueclr} \ltypetwo]  \typearrowp{\colrtwo} \ltype], \vartwo \hastype [\emptytype \typearrowp{\colr} \ltypetwo] 
                                         \ctypes{(\Kr{\clr{}}{\bullet}+\Kr{\clrtwo{}}{\bullet})} 
                                         \bapp{\var}{\bla\varthree\bapp\vartwo(\monoToBlue{\comb{J}\varthree})}\hastype \ltype 
                                }{ 
                                        \infer[\typingruleAx]{
                                                \var\hastype [[\emptytype \typearrowp{\blueclr} \ltypetwo]  \typearrowp{\colrtwo} \ltype] \ctypes{0} \var \hastype [\emptytype \typearrowp{\blueclr} \ltypetwo]  \typearrowp{\colrtwo} \ltype
                                        }{}
                                        &
                                        \infer[\typingruleManyShort]{
                                                \vartwo \hastype [\emptytype \typearrowp{\colr} \ltypetwo] 
                                                \ctypes{\Kr{\clr{}}{\bullet}}
                                                \bla\varthree\bapp\vartwo\varthree \hastype [\emptytype \typearrowp{\blueclr} \ltypetwo] 
                                        }{\infer[\typingruleAbs ]{ 
                                        \vartwo \hastype [\emptytype \typearrowp{\colr} \ltypetwo] 
                                        \ctypes{\Kr{\clr{}}{\bullet}} 
                                        \bla\varthree\bapp\vartwo(\monoToBlue{\comb{J}\varthree}) \hastype \emptytype \typearrowp{\blueclr} \ltypetwo }{
                                \infer[\typingruleApp]{ 
                                        \vartwo \hastype [\emptytype \typearrowp{\colr} \ltypetwo] 
                                        \ctypes{\Kr{\clr{}}{\bullet}} 
                                        \bapp\vartwo{(\monoToBlue{\comb{J}\varthree})} \hastype \ltypetwo 
                                        }{ \vartwo \hastype [\emptytype \typearrowp{\colr} \ltypetwo] 
                                        \ctypes{0} 
                                        \vartwo \hastype \emptytype \typearrowp{\colr} \ltypetwo 
                                        & 
                                        \ctypes{0} 
                                        \monoToBlue{\comb{J}\varthree}\hastype\emptytype}{}
                        }}
                                }
                        }
\]
where $\mathsf{m}$ is short for $\typingruleMany$, and $\Kr{\cdot}{\cdot}$ is (the dual of) Kronecker's delta, i.e.\ $\Kr{\clr{}}{\clrtwo{}} := 1$ if $\clr{}\neq\clrtwo{}$, and $\Kr{\clr{}}{\clrtwo{}} := 0$ otherwise. By taking $\clr{}=\circ$, one obtains a typing that is not suitable for $\lam y.xy$ since $\var\hastype [[\emptytype \typearrowp{\blueclr} \ltypetwo]  \typearrowp{\colrtwo} \ltype] \not\ctypes{k} \bla\vartwo \bapp{\var}{\bla\varthree\bapp\vartwo\varthree}\hastype [\emptytype \typearrowp{\redclr} \ltypetwo]  \typearrowp{\blueclr} \ltype$. 

There exists however a positively 1-whiter typing:\\[3pt]
\[
        \infer[\typingruleAbs]{
                                \var\hastype [[\emptytype \typearrowp{\circ} \ltypetwo]  \typearrowp{\colrtwo} \ltype] \ctypes{\Kr{\colrtwo}{\bullet}} \bla\vartwo \bapp{\var}{\vartwo}\hastype [\emptytype \typearrowp{\circ} \ltypetwo]  \typearrowp{\blueclr} \ltype 
                        }{
                                \infer[\typingruleAppInt]{
                                        \var\hastype [[\emptytype \typearrowp{\circ} \ltypetwo]  \typearrowp{\colrtwo} \ltype], \vartwo \hastype [\emptytype \typearrowp{\circ} \ltypetwo] \ctypes{\Kr{\colrtwo}{\bullet}} \bapp{\var}{\vartwo}\hastype \ltype 
                                }{ 
                                        \infer[\typingruleAx]{
                                                \var\hastype [[\emptytype \typearrowp{\circ} \ltypetwo]  \typearrowp{\colrtwo} \ltype] \ctypes{0} \var \hastype [\emptytype \typearrowp{\circ} \ltypetwo]  \typearrowp{\colrtwo} \ltype
                                        }{}
                                        ~~~~&~~~~
                                        \infer[\typingruleMany]{
                                                 \vartwo \hastype [\emptytype \typearrowp{\circ} \ltypetwo] \ctypes{0} \vartwo \hastype [\emptytype \typearrowp{\circ} \ltypetwo] 
                                        }{\vartwo \hastype [\emptytype \typearrowp{\circ} \ltypetwo] \ctypes{0} \vartwo \hastype \emptytype \typearrowp{\circ} \ltypetwo }{}
                                }
                        }
\]
that is 1-cheaper since $\Kr{\colrtwo}{\bullet}< \Kr{\colr}{\bullet} + \Kr{\colrtwo}{\bullet} = \Kr{\colrtwo}{\bullet}+1$. As a result of our main Theorem~\ref{thm:main}, we shall see that $\lam y.x(\lam z.y(\comb{J}z)) \leqpwcplain \lam y.xy$ holds.
\end{example}

\subsection{Two inclusions for the \bohm preorder up to $\eta$-reductions}
The remainder of the paper is devoted to showing that the preorders $\leqimp$, $\leqpwcplain$, and $\etaredbtleq$ coincide. In this subsection, we relate the \bohm preorder up to $\eta$-reductions to the polarized whiter-cheaper improvement ($\etaredbtleq \, \subseteq \,\leqpwcplain$) and to interaction improvement ($\intrleq\,\subseteq\,\etaredbtleq$).
To prove $\etaredbtleq\ \subseteq\ \leqpwcplain$, we need to study first the case of a possibly infinite $\eta$-expansion of a single variable.

\begin{toappendix}
\begin{lemma}\label{l:eta-id}\NoteProof{l:eta-id}
        Let $\tm\in\Lambda$ and $x\in\Var$ be such that $\tm \slbtleq x$. 
        \begin{enumerate}
        \item If $\typctx\ctypes{k}\monoToBlue{\tm} : \ltype$ then there exist $\ltype',\ltype''$ such that $\typctx = x\hastype[\ltype']$, $\ltype''\whiterneg_{k'} \ltype'$ and  $\ltype''\whiterpos_{k''} \ltype$ with $k = k'+k''$;
        \item If $\typctx\ctypes{k}\monoToBlue{\tm} : \mtype$ then there exist $\mtype',\mtype''$ such that $\typctx = x\hastype[\mtype']$, $\mtype''\whiterneg_{k'} \mtype'$ and  $\mtype''\whiterpos_{k''} \mtype$ with $k = k'+k''$.
        \end{enumerate}
\end{lemma}
\end{toappendix}
We now generalize the result to any pair of \lam-terms related by $\etaredbtleq$.

\begin{toappendix}
\begin{lemma}\label{lem:whitercheaper}\NoteProof{lem:whitercheaper}
Assume $\tm\slbtleq \tmtwo$ and $\typctx\ctypes{k} 
        \monoToBlue{\tm} : \ltype$.
        Then $\typctx'\ctypes{k'} \monoToBlue{\tmtwo} : \ltype'$ with 
        $\Pair{\typctx'}{\ltype'} \whiterposp p \Pair{\typctx}{\ltype}$, for 
        $0 \leq p = k-k'$.
\end{lemma}
\end{toappendix}

\begin{corollary}\label{cor:slbtleqentailsleqpwcplain} For all $\tm,\tmtwo\in\Lambda$, $\tm\slbtleq \tmtwo$ entails $\tm\leqpwcplain \tmtwo$.
\end{corollary}

\paragraph{Completeness of the \bohm preorder up to $\eta$-reductions.}
To prove the inclusion $\intrleq\ \subseteq\ \etaredbtleq$ it is sufficient to slightly adapt the proof that the interaction preorder entails the \bohm tree preorder~\cite{InteractionEquivalence}.
The proof exploits the famous \bohm out technique, used by Hyland to construct a \lam-calculus context $C$ separating $\tm$ and $\tmtwo$ whenever $\tm\not\etabtle\tmtwo$, i.e., $C\ctxholep{\tm}\bshs$ and $C\ctxholep{\tmtwo}\bshdiv$\cite{Hyland76}.
Since we count interaction steps, we can also separate $\eta$-convertible \lam-terms like $\bOne$ and $\bId$, using the white context $\ctxhole{}\circ(\lam_\circ x.x)$.
We extend the Definition~\ref{def:paintinglambdas} of white painting $\monoToRed{(\cdot)}$ to \lam-calculus contexts $C$ by adding the case $\monoToRed{\ctxhole{}}=\ctxhole{}$.

\begin{toappendix}
        \begin{lemma}[Interaction \bohm-out]
                \NoteProof{l:separating-eta-red} 
        Let $t,u\in\Lambda$ be such that $t\etabtle u$ and $t \not\etaredbtleq u$.
        Then, there exists a \lam-calculus context $\ctx$ such that $
        \monoToRed{\ctx}\ctxholep{\monoToBlue{t}}\bshcol{i}$ and
        $\monoToRed{\ctx}\ctxholep{\monoToBlue{u}}\bshcol{i'}$ with $i' > i$.
        \label{l:separating-eta-red} 
        \end{lemma}
\end{toappendix}

\begin{theorem}[Completeness] \label{th:intleq-included-in-bohm}
        Let $t,u\in\Lambda$. If $t \intrleq u$ then $t\etaredbtleq u$.
\end{theorem}
\begin{proof} 
Assume $t\not\etaredbtleq u$, towards a contradiction. There are two cases:
	\begin{itemize}
		\item If $t\not\etabtle u$, then there is a context $C$ such that $C\ctxholep{t}\!\Downarrow_\head$, while $C\ctxholep{u}\!\not\Downarrow_\head$~\cite{Hyland76}.
		Since head reductions can be simulated by $\tohch$ by \reflemma{correspondence}.(ii), it follows that $\monoToRed{\ctx}\ctxholep{\monoToBlue{t}}\bshcols$, while $\monoToRed{\ctx}\ctxholep{\monoToBlue{u}}\bshcoldiv$. This shows $t\not\intrleq u$.
		\item If $t\etabtle u$ then $t\not\intrleq u$ follows directly from \bohm out (Lemma~\ref{l:separating-eta-red}).\qed 
	\end{itemize}
	
\end{proof}



\section{Compositionality of Polarized Whiter-Cheaper}
\label{sect:compatibility}

Our goal is to use the polarized whiter-cheaper preorder to characterize the interaction improvement ordering, which is a contextual preorder. We therefore need to show that our preorder is preserved when placing related terms into contexts, \ie that it is \emph{compositional}. This section, which is the main technical work of the paper, is devoted to establishing this fact.

Let us first describe the difficulty we encounter when proving compositionality. We are trying to show the following:

\begin{center}
        If $\tm \lecolpos \tmtwo$ and $\tmthree \lecolpos \tmfour$ then $\appp\colr\tm \tmthree \lecolpos \appp\colr\tmtwo\tmfour$ for any $\colr$.
\end{center}
Let us start by considering an arbitrary derivation of $\appp\colr\tm\tmthree$:
\[
\infer[\typingruleApp]{\typctx\uplus\typctxtwo\ctypes{k_1 + k_2 + \Kr{\colr}{\colrtwo}} \appp\colr\tm\tmthree \hastype \ltype}
{
        \typctx_1\ctypes{k_1}\tm\hastype \mtype \typearrowp\colrtwo \ltype 
        & 
        \typctx_2\ctypes{k_2}\tmthree\hastype \mtype 
}
\]
By hypothesis ($\tm \lecolpos \tmtwo$ and $\tmthree \lecolpos \tmfour$), we know that there are whiter and cheaper derivations for the premises of the application rule, but we have no way of knowing if they will still be able to be combined into an application rule:
\[
\infer[\typingruleApp (\text{if } \mtype'=\mtype'')]{\typctx_1'\uplus\typctx_2'\ctypes{k'_1 + k'_2 + \Kr{\colr}{\colrtwo'}} \appp\colr\tmtwo\tmfour \hastype \ltype'}
{
        \typctx_1'\ctypes{k'_1}\tmtwo\hastype \mtype' \typearrowp{\colrtwo'} \ltype' 
        & 
        \typctx_2'\ctypes{k'_2}\tmfour\hastype \mtype'' 
}
\]
where $\Kr{\colr}{\colrtwo'}$ is (the dual of) Kronecker's delta (as in Example~\ref{ex:playingwithJ}).
We shall show that we can indeed build a multi type $\mtype'''$ on which $\tmtwo$ and $\tmfour$ agree.
We know that $\mtype,\mtype'$ and $\mtype''$ are related by the polarized whitening relation in the sense that:
\begin{center}
        $\mtype' \whiternegp {l} \mtype$ and $\mtype '' \whiterposp {h} \mtype$ for some $h,l$ such that $k'_1 + l \leq k_1$ and $k'_2 + l \leq k_2$.
\end{center}
\begin{example}
        \begin{enumerate}
                \item It does happen that the argument type $\mtype$ is unchanged, as is the case of
        $\bla\vartwo\bapp\var\vartwo \leqpwc \var$ applied to an argument $\tmthree \leqpwc \tmthree$.
        
        \begin{center}
\scalebox{0.85}{                $
                        \infer[\typingruleApp]{\typctx,\var\hastype[\mtype \typearrowp{\redclr}\ltype]\ctypes{1 + k + \Kr{\blueclr}{\colrtwo}} \appp\colrtwo{(\bla\vartwo\bapp\var\vartwo)}\tmthree \hastype \ltype}
                {
                        \var\hastype[\mtype \typearrowp{\redclr}\ltype]\ctypes{1}\bla\vartwo\bapp\var\vartwo\hastype \mtype \typearrowp{\blueclr}\ltype
                        & 
                        \typctx\ctypes{k}\tmthree\hastype \mtype 
                }
                \hspace{10pt}
                \infer[\typingruleApp]{\typctx,\var\hastype[\mtype \typearrowp{\redclr}\ltype]\ctypes{k + \Kr{\redclr}{\colrtwo}} \appp\colrtwo{\var}\tmthree \hastype \ltype}
        {
                \var\hastype[\mtype \typearrowp{\redclr}\ltype]\ctypes{0}\var\hastype \mtype \typearrowp{\redclr}\ltype
                & 
                \typctx\ctypes{k}\tmthree\hastype \mtype 
        }
        $}
\end{center}

        What if we were to use another term than $\tmthree$ as argument, say $\tmfour$ such that $\tmthree\leqpwc \tmfour$? We have that there exists $\Pair{\typctx'}{\mtype'} \whiterposp d \Pair{\typctx}{\mtype}$ such that $\typctx'\ctypes{k'}\tmfour \hastype \mtype'$ with $k \geq k' +d$. Then it is easy to change the typing for $\var$ so that it fits $\mtype'$.
        
        \begin{center}
        $
        \infer[\typingruleApp]{\typctx',\var\hastype[\mtype' \typearrowp{\redclr}\ltype]\ctypes{k' + \Kr{\redclr}{\colrtwo}} \appp\colrtwo{\var}\tmthree \hastype \ltype}
        {
                \var\hastype[\mtype' \typearrowp{\redclr}\ltype]\ctypes{0}\var\hastype \mtype' \typearrowp{\redclr}\ltype
                & 
                \typctx'\ctypes{k'}\tmfour\hastype \mtype '
        }
        $
        \end{center}

        Note that $\Pair{\typctx',\var\hastype[\mtype \typearrowp{\redclr}\ltype]}{\ltype} \whiterposp d \Pair{\typctx,\var\hastype[\mtype' \typearrowp{\redclr}\ltype]}{\ltype}$.
                The change from $\var\hastype[\mtype \typearrowp{\redclr}\ltype]\ctypes{0}\var\hastype \mtype \typearrowp{\redclr}\ltype$ to $\var\hastype[\mtype' \typearrowp{\redclr}\ltype]\ctypes{0}\var\hastype \mtype' \typearrowp{\redclr}\ltype$ possibly whitens in negative and positive positions, but the negative changes do not appear in the conclusion of the derivation as they are consumed in the application rule.
        
        \item Consider now ${\bla\var\bapp\var\bla\varthree\bapp{\vartwo}\varthree} \leqpwc {\bla\var\bapp\var\vartwo}$, and let us apply a term $\tmthree$ on both sides. We have a derivation for $\appp\colrtwo{({\bla\var\bapp\var\bla\varthree\bapp{\vartwo}\varthree})}\tmthree$:
\[
\resizebox{\textwidth}{!}{%
        \infer[\typingruleApp]{\typctx,\vartwo\hastype[\mtype\typearrowp{\redclr}\ltype]\ctypes{1 + k + \Kr{\blueclr}{\colrtwo} + \Kr{\colr}{\blueclr}} \appp\colrtwo{(\bla\var\bapp\var\bla\varthree\bapp{\vartwo}\varthree)}\tmthree \hastype \ltype'}
        {
\infer[\typingruleAbs]{ \vartwo\hastype[\mtype\typearrowp{\redclr}\ltype]\ctypes{1+ \Kr{\colr}{\blueclr}}\bla\var\bapp\var\bla\varthree\bapp{\vartwo}\varthree\hastype [[\mtype\typearrowp{\blueclr}\ltype] \typearrowp{\colr}\ltype']\typearrowp{\blueclr}\ltype'}{    
\infer[\typingruleApp]{
\vartwo\hastype[\mtype\typearrowp{\redclr}\ltype],\var\hastype[[\mtype\typearrowp{\blueclr}\ltype] \typearrowp{\colr}\ltype']\ctypes{1+ \Kr{\colr}{\blueclr}}\bapp\var\bla\varthree\bapp{\vartwo}\varthree\hastype \ltype'}
{
        \var\hastype[[\mtype\typearrowp{\blueclr}\ltype] \typearrowp{\colr}\ltype']\ctypes{0}\var\hastype [\mtype\typearrowp{\blueclr}\ltype] \typearrowp{\colr}\ltype'
&
 \vartwo\hastype[\mtype\typearrowp{\redclr}\ltype]\ctypes{1}\bla\varthree\bapp{\vartwo}\varthree\hastype [\mtype\typearrowp{\blueclr}\ltype]
}}
                & 
                \typctx\ctypes{k}\tmthree\hastype [[\mtype\typearrowp{\blueclr}\ltype] \typearrowp{\colr}\ltype']
        }
}
\]      
There exists a positive whiter-cheaper derivation for $\bla\var\bapp\var\vartwo$ but it is hard to use for the application to $\tmthree$ afterwards. See the incomplete derivation:
        \[
\resizebox{\textwidth}{!}{%
                        \infer*[\typingruleApp?]{\typctx,\vartwo\hastype[\mtype\typearrowp{\redclr}\ltype]\ctypes{k + \Kr{\blueclr}{\colrtwo} + \Kr{\colr}{\blueclr}} \appp\colrtwo{(\bla\var\bapp\var\vartwo)}\tmthree \hastype \ltype'}
                        {
                                \infer[\typingruleAbs]{ \vartwo\hastype[\mtype\typearrowp{\redclr}\ltype]\ctypes{\Kr{\colr}{\blueclr}}\bla\var\bapp\var\vartwo\hastype [[\mtype\typearrowp{\redclr}\ltype] \typearrowp{\colr}\ltype']\typearrowp{\blueclr}\ltype'}{     
                                        \infer[\typingruleApp]{
                                                \vartwo\hastype[\mtype\typearrowp{\redclr}\ltype],\var\hastype[[\mtype\typearrowp{\redclr}\ltype] \typearrowp{\colr}\ltype']\ctypes{ \Kr{\colr}{\blueclr}}\bapp\var\vartwo\hastype \ltype'}
                                        {
                                                \var\hastype[[\mtype\typearrowp{\redclr}\ltype] \typearrowp{\colr}\ltype']\ctypes{0}\var\hastype [\mtype\typearrowp{\redclr}\ltype] \typearrowp{\colr}\ltype'
                                                &
                                                \vartwo\hastype[\mtype\typearrowp{\redclr}\ltype]\ctypes{0}\vartwo\hastype [\mtype\typearrowp{\redclr}\ltype]
                                }}
                                & 
                                \typctx\ctypes{k}\tmthree\hastype [[\mtype\typearrowp{\blueclr}\ltype] \typearrowp{\colr}\ltype']
                }}
        \]
        We therefore need to change the black arrow in the type of $\tmthree$, which appears in a negative position. Such a change is possible, but may lead to more positive changes in $\Pair{\typctx}{[[\mtype\typearrowp{\blueclr}\ltype] \typearrowp{\colr}\ltype']}$. This reasoning is exactly the object of the repainting mechanism described in \reflemma{repaint}.
\end{enumerate}
\end{example}

\paragraph{Repainting Negative Occurrences.} 
Negative color occurrences in a typing judgment $\typctx \ctypes{k} \tm \hastype \ltype$ can be seen as \emph{unspecified} colors, not directly determined by those appearing in the syntax of the term.
However, they cannot be recolored arbitrarily.
We formally specify how one may \emph{repaint} a negatively occurring arrow white, by constructing a new typing derivation that either also includes one positively occurring arrow repainted white, or whose index $k$ is shifted by~1.
 
\begin{toappendix}
\begin{lemma}[Repainting]
        $\label{l:repaint}$\NoteProof{l:repaint}
        Suppose $\typctx\hasstype[k] \tm : \ltype$ and that $\Pair{\typctx'}{\ltype'} \whiternegp 1 \Pair{\typctx}{\ltype}$. Then there exists a derivation of $\typctx'' \hasstype[k'] \tm:\ltype''$ such that one of the following holds:
        \begin{enumerate}[(i)]
                \item  $\Pair{\typctx''}{\ltype''}\whiterposp 1 \Pair{\typctx'}{\ltype'}$ and $k=k'$; or
                \item  $\Pair{\typctx''}{\ltype''} = \Pair{\typctx'}{\ltype'}$ and $\abs{k'-k} = 1$. 
        \end{enumerate}
        Equivalently, Point (i) and (ii) may be rephrased as: there exists $0\leq i \leq 1$ such that  $\Pair{\typctx''}{\ltype''}\whiterposp i \Pair{\typctx'}{\ltype'}$ and $\abs{k-k'} \leq 1-i$.
\end{lemma}
\end{toappendix}
This first repainting lemma only applies to changes of a singular color in a singular arrow type, but we might need to repaint several. With the help of a commutation property for the $\whiternegp 1$ and $\whiterposp 1$ relations (\reflemma{commutation-neg-and-pos} below), we show that one can apply repeatedly the repainting lemma, obtaining Proposition~\ref{prop:multirepaint}.

        \begin{lemma}
                \label{l:commutation-neg-and-pos}
                Let $\typctx,\madewhiterpos{\typctx}, \madewhiterneg{\typctx}$ and $\ltype,\madewhiterpos{\ltype},\madewhiterneg{\ltype}$ such that 
                \[
                \Pair{\madewhiterneg{\typctx}}{\madewhiterneg{\ltype}}\whiternegp 1 \Pair{\typctx}{\ltype}\textrm{ and }\Pair{\madewhiterpos{\typctx}}{\madewhiterpos{\ltype}}\whiterposp 1 \Pair{\typctx}{\ltype}.
                \]
                Then there are $\madewhiternegandpos\typctx,\madewhiternegandpos\ltype$ such that 
                \[ 
                        \Pair{\madewhiternegandpos\typctx}{\madewhiternegandpos\ltype}
                        \whiterposp 1
                        \Pair{\madewhiterneg{\typctx}}{\madewhiterneg{\ltype}}\textrm{
and }\Pair{\madewhiternegandpos\typctx}{\madewhiternegandpos\ltype} 
                        \whiternegp 1
                        \Pair{\madewhiterpos{\typctx}}{\madewhiterpos{\ltype}}.\]
\end{lemma}
\begin{proof}
        The typing \(\Pair{\madewhiterneg{\typctx}}{\madewhiterneg{\ltype}}\) arises by repainting one negatively-occurring arrow in \(\Pair{\typctx}{\ltype}\) from black to white; \(\Pair{\madewhiterpos{\typctx}}{\madewhiterpos{\ltype}}\) arises by repainting one positively occurring arrow in \(\Pair{\typctx}{\ltype}\). Construct \(\Pair{\madewhiternegandpos\typctx}{\madewhiternegandpos\ltype}\) by making both these repaintings.\qed
\end{proof}



\begin{toappendix}
        \begin{proposition}[Sequence of Repainting]
        $\label{prop:multirepaint}$\NoteProof{prop:multirepaint}
        Suppose $\typctx\hasstype[k] \tm : \ltype$ and that $\Pair{\typctx'}{\ltype'} \whiternegp {k_1} \Pair{\typctx}{\ltype}$ for some $k_1 \geq 0$. Then there exists a derivation of $\typctx'' \hasstype[k'] \tm:\ltype''$ and $k_2$ with $0 \leq k_2 \leq k_1$ such that  
        \begin{itemize}
                \item  $\Pair{\typctx''}{\ltype''}\whiterposp {k_2} \Pair{\typctx'}{\ltype'}$; and
                \item  $\abs{k- k'} \leq k_1 - k_2$. 
        \end{itemize}
        \end{proposition}
\end{toappendix}

\paragraph{Back to the Compositionality Proof.} Now that we have introduced an appropriate notion of repainting, we return to the proof of compositionality, where the following proposition provides the main argument for the application case.

\begin{toappendix}
\begin{proposition}\label{prop:apprepaint}\NoteProof{prop:apprepaint}
        Given $\typctx \hasstype[k] \tm: \mtype \typearrowp{\colr}\ltype $ and $\typctxtwo \hasstype[l] \tmtwo: \mtypetwo$ with either $\mtype \whiternegp {d} \mtypetwo$ or $\mtypetwo \whiterposp {d} \mtype$, there exists a typing $\typctx'+\typctxtwo' \hasstype[m] \appp\colrtwo\tm\tmtwo: \ltype'$ and $d'$ with $ 0 \leq d' \leq d$ such that $\Pair{\typctx'+\typctxtwo'}{\ltype'} \whiterposp {d'} \Pair{\typctx+\typctxtwo}{\ltype}$ and $m \leq k+l+ \Kr{\colr}{\colrtwo} + d -d'$. 
\end{proposition}
\end{toappendix}

\begin{toappendix}
\begin{proposition}[Compositionality of Polarized Whiter-Cheaper Improvement]\NoteFullProof{prop:leqpwc-contextual}
        If $\tm \lecolpos \tmtwo$ then $\ctxp\tm \lecolpos \ctxp\tmtwo$ for any checkers context $\ctx$.$\label{prop:leqpwc-contextual}$
\end{proposition}
\end{toappendix}

\begin{proof}
        Note that it is sufficient (by transitivity of $\lecolpos$) to prove for all $\tm,\tmtwo,\tmthree$ such that  $\tm \lecolpos \tmtwo$, we have that for all $\colr\in\set{\circ,\bullet}$:
        \begin{center}
                $\lap{\colr}\var\tm \lecolpos \lap{\colr}\var\tmtwo$;\hspace{20pt} 
                 $\appp\colr\tm\tmthree \lecolpos \appp\colr\tmtwo\tmthree$;\hspace{20pt}  and
                 $\appp\colr\tmthree\tm\lecolpos  \appp\colr\tmthree\tmtwo$.
        \end{center}
        
        We prove here the first application case; the other cases are easier.
        
        Suppose $\Tr{\typctx}{\ltype}{k} \in \semint{\tm \cdot^{\colrtwo} \tmthree}$. Then we must have $\Tr{\typctx_1}{\mtype \typearrowp{\colr}\ltype}{k_1} \in \semint{\tm}$ and $\Tr{\typctx_2}{\mtype}{k_2} \in \semint{\tmthree}$ with $k = k_1 + k_2 + \Kr{\colr}{\colrtwo}$ and $\typctx = \typctx_1 + \typctx_2$. Since $\tm\lecolpos \tmtwo$ we can find $\Tr{\typctx'_1}{\mtype' \typearrowp{\clrp{}} \ltype'}{k'_1} \in \semint{\tmtwo}$ with $\Pair{\typctx'_1}{\mtype' \typearrowp{\clrp{}} \ltype'} \whiterposp {d} \Pair{\typctx_1}{\mtype \typearrowp{\colr}\ltype}$ and $k_1 \geq k'_1 + d$. This implies that there are $d_1$, $d_2$ such that $d = d_1 + d_2 + \Kr{\colr}{\clrp{}}$ and $\Pair{\typctx'_1}{\ltype'}\whiterposp {d_1} \Pair{\typctx_1}{\ltype}$ and $\mtype' \whiternegp {d_2} \mtype$. By Proposition~\ref{prop:apprepaint} there exists $\Tr{\typctx''_1+\typctx'_2}{\ltype''}{m} \in \semint{\tmtwo \cdot^{\colrtwo} \tmthree}$ with $\Pair{\typctx''_1+\typctx'_2}{\ltype''} \whiterposp {d'} \Pair{\typctx'_1+\typctx_2}{\ltype'}$ and $m \leq k'_1 + k_2 + \Kr{\clrp{}}{\colrtwo} +d_2-d'$. So we have\vspace{-3pt}
                \[
                \Pair{\typctx''_1+\typctx'_2}{\ltype''} \whiterposp {d'} \Pair{\typctx'_1+\typctx_2}{\ltype'} \whiterposp {d_1} \Pair{\typctx_1+\typctx_2}{\ltype}
                \]
                and hence $\Pair{\typctx''_1+\typctx'_2}{\ltype''} \whiterposp {d'+d_1} \Pair{\typctx_1+\typctx_2}{\ltype} = \Pair{\typctx}{\ltype}$. It only remains to show the easy inequation $k \geq m+d'+d_1$.\qed
\end{proof}

By compositionality, soundness and completeness of the checkers multi type system, we are able to relate $\leqpwc$ and $\leqchcolr$.

\begin{theorem}[PWC Improvement implies Interaction Improvement]
        For all $\tm,\tmtwo\in\Lambda$, $\monoToBlue{\tm}\lecolpos \monoToBlue{\tmtwo}$ entails $\tm\leqchcolr \tmtwo$.\label{thm:pwcimpliesleqchcolr}
\end{theorem}
\begin{proof}
        Suppose $\monoToBlue{\tm}\lecolpos \monoToBlue{\tmtwo}$. 
        We have to show that for any checkers-context $\ctx$, if $\ctxp{\monoToBlue{\tm}} \bshcol k$ for some $k$ then $\ctxp{\monoToBlue{\tmtwo}} \bshcol{k'}$ for some $k' \leq k$. So suppose $\ctxp{\monoToBlue{\tm}} \bshcol k$. By completeness of the checkers type system (Theorem~\ref{thm:soundandcomplete}.2), we have some $\Tr{\typctx}{\ltype}{k} \in \semint{\ctxp{\monoToBlue{\tm}}}$. Proposition~\ref{prop:leqpwc-contextual} tells us that $\ctxp{\monoToBlue{\tm}}\lecolpos \ctxp{\monoToBlue{\tmtwo}}$ so by definition of $\lecolpos$ there is some $\Tr{\typctx'}{\ltype'}{k'} \in \semint{\ctxp{\monoToBlue{\tmtwo}}}$ with $k' \leq k$. By soundness of the type system (Theorem~\ref{thm:soundandcomplete}.1), there is $k'' \leq k'$ such that $\ctxp{\monoToBlue{\tmtwo}}\bshcol{k''}$ as required. \qed
\end{proof}

\newcommand{\leqrelplain}{\sqsubseteq^{\mathrm{rel}}}
\paragraph{Wrapping Up.} We can at last state our final theorem, completing the characterization of interaction improvement.
Indeed, interaction improvement also characterizes the preorder induced by relational semantics $\leqrelplain$, i.e., the inequational theory of the model $\mathcal{E}$ in Breuvart et al. \cite{BreuvartMR18}. We focus on this model because it can be presented via a multi type system very similar to the one of \reffig{colored-types-for-cbn}.
In fact, checkers multi types are obtained by coloring the well-known multi types for head evaluation \cite{BKV17,DBLP:journals/mscs/PaoliniPR17}. If one removes the colors from the types in the interpretation $\semint{\monoToBlue{\tm}}$ of a term then one recovers exactly the interpretation $\sem\tm^{\textrm{rel}}$ in relational semantics (as presented by multi types). However, it is not true that $\sem\tm^{\textrm{rel}}\,\subseteq\,\sem\tmtwo^{\textrm{rel}}$ implies $\semint{\monoToBlue{\tm}} \,\subseteq\,\semint{\monoToBlue{\tmtwo}}$ (because of the various ways of coloring types). 
This is exactly the reason why we introduced the $\leqpwc$ preorder to compare checkers type interpretations instead of the vanilla set inclusion.

\begin{theorem}\label{thm:main}
        For $\tm,\tmtwo\in\Lambda$, the following are equivalent:
        \begin{enumerate}
                \item \emph{\bohm tree preorder up to $\eta$-reductions:} $\tm\etaredbtleq\tmtwo$;
                \item \emph{Polarized Whiter-Cheaper Type Improvement:} $\tm\leqpwcplain\tmtwo$;
                \item \emph{Interaction Improvement:} $\tm\intrleq\tmtwo$;
                \item \emph{Plain Type Preorder:} $\tm\leqrelplain\tmtwo$.
        \end{enumerate}
\end{theorem}

\begin{proof}$(1 \Rightarrow 2)$. By Corollary \ref{cor:slbtleqentailsleqpwcplain}.
        \begin{tikzpicture}[overlay,yshift=3pt]
                \node at (4,-.4){
                        \begin{tikzcd}[row sep=small]
                {\etaredbtleq} \arrow[rr, "\text{\scriptsize Cor. \ref{cor:slbtleqentailsleqpwcplain}}" description, bend left, shift left] \arrow[dd, "\text{\cite{BreuvartMR18}}" description, leftrightarrow] &  & {\leqpwcplain} \arrow[rr, "\text{\scriptsize Thm. \ref{thm:pwcimpliesleqchcolr}}" description, bend left, shift left] &  & {\intrleq} \arrow[llll, "\text{\scriptsize Thm. \ref{th:intleq-included-in-bohm}}" description, bend left, shift left] 
                \\\\
                {\leqrelplain}                                                                                 &  &                                                       &  &                                                          
        \end{tikzcd}
                };
        \end{tikzpicture}

$(2 \Rightarrow 3)$. By Theorem~\ref{thm:pwcimpliesleqchcolr}.

$(3 \Rightarrow 1)$. By Theorem~\ref{th:intleq-included-in-bohm}.

$(1 \Leftrightarrow 4)$. By \cite[Theorem 5.6]{BreuvartMR18}.
\hfill\qed
\end{proof}

\newcommand{\whiteintrleq}{\sqsubseteq^{\redclr\mathrm{imp}}}
\paragraph{White head contexts are enough.} 
Our result of completeness states that interaction improvement is included in the \bohm tree up to $\eta$-reductions preorder. 
Due to the particular shape of the separating context built in \reflemma{separating-eta-red}, our result has stronger consequences.
The \emph{white applicative interaction improvement} can be defined analogously to the standard interaction improvement, but restricted to head contexts consisting only of white applications and lambda abstractions. In this sense, these contexts can be viewed as uniformly ``whitened'' head contexts from the plain $\lambda$-calculus. Formally, we write $\tm \whiteintrleq \tmtwo$ if:
\begin{center}
        $\forall$ head contexts $\ctx$ such that $
        \monoToRed{\ctx}\ctxholep{\monoToBlue{t}}\bshcol{i}$, we have  
        $\monoToRed{\ctx}\ctxholep{\monoToBlue{u}}\bshcol{i'}$ for some $i' \leq i$.
\end{center}

It is evident that the class of applicative white contexts is a subset of the unrestricted general checker contexts $\chcontexts$. 
Therefore, the inclusion $\intrleq\,\subseteq\,\whiteintrleq$ holds trivially. More precisely:
\begin{equation}\label{eq:reduce-to-white-applicative-ctx}
t \intrleq u \;\xRightarrow{\text{triv.}}\; t \whiteintrleq u \;\xRightarrow{\text{L.\ref{l:separating-eta-red}}}\; t \etaredbtleq u
\end{equation}
Since we have already proved that $\etaredbtleq\, \subseteq\, \intrleq$, it follows immediately from (\ref{eq:reduce-to-white-applicative-ctx}) that it suffices to compare terms within white contexts to determine the interaction improvement between two programs.
The applicative restriction was not emphasized in \cite{InteractionEquivalence}, even though the separation argument there already used an applicative context. Here, the restriction is particularly relevant because other proof methods, such as Howe’s method \cite{Howe1996method,DBLP:books/cu/12/Pitts12}, do not easily establish that applicative contexts alone are sufficient. Notably, prior to this work, it could not even be proven that $\tm \to_\eta \tmtwo$ implies $\tm \intrleq \tmtwo$; this was only conjectured in \cite{InteractionEquivalence}.

\ignore{
\begin{proposition}\label{prop:apprepaint}
        Given $\typctx \hasstype[k] M: \mu \typearrowpp{\colr}{}\ltype $ and $\Delta \hasstype[l] N: \nu$ with either $\mu \whiternegp {d} \nu$ or $\nu \whiterposp {d} \mu$, there exists a typing $\typctx'+\Delta' \hasstype[m] M \cdot^{\colrtwo} N: \ltype'$ and $d'$ with $ 0 \leq d' \leq d$ such that $\Pair{\typctx'+\Delta'}{\ltype'} \whiterposp {d'} \Pair{\typctx+\Delta}{\ltype}$ and $m \leq k+l+ \Kr{\colr}{\colrtwo} + d -d'$. 
\end{proposition}
\begin{proof}
        We proceed by induction on the number of black-annotated arrows in the types $\mu$ and $\nu$. 
        
        For the base case, if there are no black-annotated arrows in $\mu$ or $\nu$, then $\mu \whiternegp {d} \nu$ or $\nu \whiterposp {d} \mu$ entails that $d= 0$ and $\mu = \nu$. Then we can immediately use the application rule to conclude that $\typctx+\Delta \hasstype[k+l+\Kr{\colr}{\colrtwo}] M \cdot^{\colrtwo} N : \ltype$, and taking $d'=0$ we conclude. 
        
        For the inductive step, we treat the case where $\nu \whiterposp {d} \mu$; the other case is similar and simpler. 
        
        We have $\nu \typearrowpp{\colr}{}\ltype \whiternegp {d} \mu \typearrowpp{\colr}{}\ltype$.   Applying Corollary~\ref{cor:multirepaint} to $M$ we obtain a typing $\typctx' \hasstype[k'] M:\mu' \typearrowpp{\clrp{}}{}\ltype'$ where $\Pair{\typctx'}{\mu' \typearrowpp{\clrp{}}{}\ltype'} \whiterposp {d_1} \Pair{\typctx}{\nu \typearrowpp{\colr}{}\ltype}$ and $k' \leq k + d-d_1$.
        
        The $d_1$ recolourings that transform $\Pair{\typctx}{\nu \typearrowpp{\colr}{}\ltype}$ into $\Pair{\typctx'}{\mu' \typearrowpp{\clrp{}}{}\ltype'}$ split among the components of the type so that 
        there are $d_2, d_3$ such that $\Pair{\typctx'}{\ltype'} \whiterposp {d_2} \Pair{\typctx}{\ltype}$, $\mu' \whiternegp {d_3} \nu$ and $d_1 = d_2 + d_3 + \Kr{\colr}{\clrp{}}$. Since $\mu'$ has strictly fewer black-annotated arrows than $\mu$,
        we can apply the inductive hypothesis to $\typctx' \hasstype[k'] M:\mu' \typearrowpp{\clrp{}}{}\ltype'$ and $\Delta \hasstype[l] N: \nu$. This gives 
        us a typing $\typctx'' + \Delta' \hasstype[m] M \cdot^{\colrtwo} N: \ltype''$ where $\Pair{\typctx''+\Delta'}{\ltype''} \whiterposp {d'} \Pair{\typctx'+\Delta}{\ltype'}$ and
        $m \leq k' + l +  \Kr{\clrp{}}{\colrtwo} + d_3 - d'$.
        Observe that $\Pair{\typctx''+\Delta'}{\ltype''} \whiterposp {d'+d_2} \Pair{\typctx+\Delta}{\ltype}$. To complete the proof we calculate
        \begin{eqnarray*}
                m   & \leq & k'+l+ \Kr{\clrp{}}{\colrtwo} + d_3 - d'  \\
                & \leq & k + d - d_1 + l+ \Kr{\clrp{}}{\colrtwo} + d_3 - d'  \\
                & =  & k + l + \Kr{\clrp{}}{\colrtwo} + d - d_2 - d_3 - \Kr{\colr}{\clrp{}} + d_3 - d'  \\
                & = & k + l+ \Kr{\clrp{}}{\colrtwo} -  \Kr{\colr}{\clrp{}} + d - (d'+d_2)\\
                & \leq & k + l+ \Kr{\colr}{\colrtwo} + d - (d'+d_2)
        \end{eqnarray*}
        noting that $ \Kr{\clrp{}}{\colrtwo} - \Kr{\colr}{\clrp{}} \leq \Kr{\colr}{\colrtwo}$. 
\end{proof}

\begin{proposition}\label{prop:lecolcontextual}
        The ordering $\lecolpos$ on checkers terms is contextual. That is, if $M \lecolpos M'$ and $C[]$ is a checkers context then $C[M] \lecolpos C[M']$. 
\end{proposition}
\begin{proof}
        It suffices to prove that $M \lecolpos M'$ implies $\lambda^{\colr}x.M \lecolpos \lambda^{\colr}x.M'$ (for any $\colr$) and for all $N$ and $\colr$, $M \cdot^{\colr} N \lecolpos M' \cdot^{\colr} N$ and $N \cdot^{\colr} M \lecolpos N \cdot^{\colr} M'$.
        
        For the abstraction case, note that $\Tr{\typctx}{\mu \typearrowpp{\colr}{}\ltype}{k} \in \semint{\lambda^{\colr} x. M}$ if and only if
        $\Tr{\typctx, x:\mu}{\ltype}{k} \in \semint{M}$, and $\Pair{\typctx}{\mu \typearrowpp{\colr}{}\ltype} \whiterposp {d} \Pair{\typctx'}{\mu' \typearrowpp{\colr}{}\ltype'}$ if and only if $\Pair{\typctx,x:\mu}{\ltype} \whiterposp d \Pair{\typctx', x:\mu'}{\ltype'}.$ The result then follows directly.
        
        For the first application case, suppose $\Tr{\typctx}{\ltype}{k} \in \semint{M \cdot^{\colrtwo} N}$. Then we must have $\Tr{\typctx_1}{\mu \typearrowpp{\colr}{}\ltype}{k_1} \in \semint{M}$ and $\Tr{\typctx_2}{\mu}{k_2} \in \semint{N}$ with $k = k_1 + k_2 + \Kr{\colr}{\colrtwo}$ and $\typctx = \typctx_1 + \typctx_2$. Since $M\lecolpos M'$ we can find $\Tr{\typctx'_1}{\mu' \typearrowpp{\clrp{}}{} \ltype'}{k'_1} \in \semint{M'}$ with $\Pair{\typctx'_1}{\mu' \typearrowpp{\clrp{}}{} \ltype'} \whiterposp {d} \Pair{\typctx_1}{\mu \typearrowpp{\colr}{}\ltype}$ and $k_1 \geq k'_1 + d$. This implies that there are $d_1$, $d_2$ such that $d = d_1 + d_2 + \Kr{\colr}{\clrp{}}$ and $\Pair{\typctx'_1}{\ltype'}\whiterposp {d_1} \Pair{\typctx_1}{\ltype}$ and $\mu' \whiternegp {d_2} \mu$. By Proposition~\ref{prop:apprepaint} there exists $\Tr{\typctx''_1+\typctx'_2}{\ltype''}{m} \in \semint{M' \cdot^{\colrtwo} N}$ with $\Pair{\typctx''_1+\typctx'_2}{\ltype''} \whiterposp {d'} \Pair{\typctx'_1+\typctx_2}{\ltype'}$ and $m \leq k'_1 + k_2 + \Kr{\clrp{}}{\colrtwo} +d_2-d'$. Then we have
        \[
        \Pair{\typctx''_1+\typctx'_2}{\ltype''} \whiterposp {d'} \Pair{\typctx'_1+\typctx_2}{\ltype'} \whiterposp {d_1} \Pair{\typctx_1+\typctx_2}{\ltype}
        \]
        and hence $\Pair{\typctx''_1+\typctx'_2}{\ltype''} \whiterposp {d'+d_1} \Pair{\typctx_1+\typctx_2}{\ltype} = \Pair{\typctx}{\ltype}$. It remains to show that $k \geq m+d'+d_1$:
        \begin{eqnarray*}
                m   & \leq & k'_1 + k_2 + \Kr{\clrp{}}{\colrtwo} + d_2 - d'  \\
                & \leq & k_1 - d + k_2 + \Kr{\clrp{}}{\colrtwo} + d_2 - d'  \\
                & =  & k_1 + k_2 + \Kr{\clrp{}}{\colrtwo} +d_2 - d' - d_1 - d_2 - \Kr{\colr}{\clrp{}} \\
                & = &  k_1 + k_2 + \Kr{\clrp{}}{\colrtwo} - \Kr{\colr}{\clrp{}} - d' - d_1 \\
                & \leq & k_1 + k_2 + \Kr{\colr}{\colrtwo}- d' - d_1 \\
                & = & k - d' - d_1
        \end{eqnarray*}
        as required.
        
        The other application case is handled symmetrically. 
\end{proof}

\begin{proposition}
        For all $M,N\in\Lambda$, $\monoToBlue{M}\lecolpos \monoToBlue{N}$ entails $M\leqchcolr N$.
\end{proposition}
\begin{proof}
        Suppose $\monoToBlue{M}\lecolpos \monoToBlue{N}$. 
        We have to show that for any checkers-context $C[]$, if $C[\monoToBlue{M}] \Downarrow _k$ for some $k$ then $C[\monoToBlue{N}] \Downarrow_{k'}$ for some $k' \leq k$. So suppose $C[\monoToBlue{M}] \Downarrow _k$. By the first part of Lemma~\ref{lem:soundandcomplete}, we have some $\Tr{\typctx}{\ltype}{k} \in \semint{C[\monoToBlue{M}]}$. Proposition~\ref{prop:lecolcontextual} tells us that $C[\monoToBlue{M}]\lecolpos C[\monoToBlue{N}]$ so by definition of $\lecolpos$ there is some $\Tr{\typctx'}{\ltype'}{k'} \in \semint{C[\monoToBlue{N}]}$ with $k' \leq k$. By the second part of Lemma~\ref{lem:soundandcomplete}, there is $k'' \leq k'$ such that $C[\monoToBlue{N}]\Downarrow_{k''}$ as required. 
\end{proof}



}


\section{Conclusions}
The main contribution of this paper is to make precise in what sense the relational semantics conveys quantitative information about programs.
We establish this by characterizing the relational preorder in terms of interaction improvement---a quantitative refinement of the contextual preorder that measures the number of interactions between a term and its context.
Our key technical tool is the checkers calculus of~\cite{InteractionEquivalence}, which, in the spirit of game semantics, treats programs and environments as first-class entities.
\paragraph{Future Work.} Several research directions stem from this work:
\begin{itemize}
        \item \emph{Towards more applied calculi.} We have carried out our analysis in the world of the untyped call-by-name $\lambda$-calculus. 
        We would like to extend our results on the one hand to typed calculi, such as \textsf{PCF}, and on the other one to calculi with sharing, such as call-by-value and call-by-need.
        \item \emph{Revisiting improvement theory.} Sands and collaborators employ improvement theory to establish several quantitative results on functional program transformations (see, \eg,~\cite{DBLP:journals/tcs/Sands96,SandsImprovementTheory,SandsToplas}).
We aim to explore whether our framework can also be applied in this setting.
To do so, it may be necessary to relax the constraint on the difference in interaction steps in the definition of $\intrleq$, for instance by allowing a linear overhead.
        
\item \emph{Categorical analysis; relationship with other models}.
We aim to investigate how our annotated semantics can be understood at a more abstract, categorical level.
While Cartesian closed categories are part of the picture, we expect that additional structure will be required to capture the “colors” of the constructs and their operational interpretation.
A more abstract formulation of the model may allow us to clarify its relationship with game-theoretic approaches to program semantics~\cite{DBLP:journals/iandc/HylandO00,DBLP:journals/iandc/AbramskyJM00,LevyLICSGames,DBLP:conf/fossacs/AlcoleiCL19,Clairambault24}, which originally inspired this work, and to explore the broader applicability of our techniques.        
\end{itemize}


\paragraph{Acknowledgements.} We thank Beniamino Accattoli for interesting discussions and the anonymous reviewers for their careful reading and useful suggestions.

\newpage
\bibliographystyle{splncs04}
\bibliography{include/biblio}

\techreport{\newpage
\appendix
\setboolean{appendix}{true}

\section{Proofs of Section~\ref{sec:types}}

\begin{lemma}[Splitting multisets with respect to derivations]
	\label{l:types-splitting-multisets}
	Let $\tm$ be a term, $\tderiv \derives  \typectx \ctypes{k} \tm : \mtype$ a derivation, and $\mtype = \mtypetwo \mplus \mtypethree$  a splitting. Then there exist two derivations $\tderiv_{\mtypetwo} \derives  \typectx_{\mtypetwo} \ctypes{k_1} \tm : \mtypetwo$, and $\tderiv_{\mtypethree} \derives  \typectx_{\mtypethree} \ctypes{k_2} \tm : \mtypethree$ such that $\typctx = \typctx_{\mtypetwo} \mplus \typctx_{\mtypethree}$, $k = k_1 + k_2$, and $\insize{\tderiv}=\insize{\tderiv_{\mtypetwo}}+\insize{\tderiv_{\mtypethree}}$.
\end{lemma}

\begin{proof}
	\applabel{l:types-splitting-multisets}
	The last rule of $\tderiv \derives  \typectx_{\mtype} \ctypes{k} \tm : \mtype$ can only be $\typingruleMany$, thus it is enough to re-group its hypotheses according to $\mtypetwo$ and $\mtypethree$. Since $\typingruleMany$ rules do not count in the measure of type derivations, it is immediate that $\insize{\tderiv}=\insize{\tderiv_{\mtypetwo}}+\insize{\tderiv_{\mtypethree}}$.
\end{proof}

\begin{lemma}[Substitution]
	\label{l:types-substitution}
	If $\tderiv \derives \typctx, \var\hastype \mtype \ctypes k \tm \hastype \gtype$ and $\tderivtwo \derives \typctxtwo \ctypes {k'} \tmtwo \hastype \mtype$ then there exists a derivation $\tderivp \derives \typctx\uplus\typctxtwo \ctypes {k+k'} \tm\isub\var\tmtwo \hastype \gtype$ such that $\insize{\tderiv'}=\insize{\tderiv}+\insize{\tderivtwo}$.
\end{lemma}

\begin{proof}
	By induction on $\tderiv$. Cases of $\gtype$:
	\begin{enumerate}
		\item Linear types, \ie $\gtype \defeq \ltype$. Cases of the last derivation rule:
		
		\begin{enumerate}
			\item \emph{Unsubstituted variable}, \ie: $$\tderiv = \infer[\typingruleAx]{\var\hastype \zero,\vartwo \hastype [\ltype] \ctypes {0} \vartwo \hastype \ltype}{}$$
			Then as $\mtype = \zero$, $\tderivtwo$ must be of the following shape: 
			
			$$\infer[\typingruleMany]{\typctxtwo \ctypes {0} \tmtwo \hastype \zero}{}$$ 
			
			Hence $\tderivp \defeq \tderiv$ works, as $\vartwo\isub\var\tmtwo = \vartwo$ and $k'=0$. Clearly, $\insize{\tderiv}=\insize{\tderiv'}=0$, as required.
			\item \emph{Substituted variable}, \ie: $$\tderiv = \infer[\typingruleAx]{\var \hastype [\ltype] \ctypes 0 \var \hastype \ltype}{}$$
			Then, as $\mtype = [\ltype]$, $\tderivtwo$ must be of the following shape: 
			
			$$\infer[\typingruleMany]{\typctxtwo \ctypes {k'} \tmtwo \hastype [\ltype]}{\tderivtwo' \exder \typctxtwo \ctypes {k'} \tmtwo \hastype \ltype}$$
			
			Hence $\tderivp \defeq \tderivtwo'$ works, as $\var\isub\var\tmtwo = \tmtwo$ and $k=0$. Clearly, $\insize{\tderiv'}=\insize{\tderivtwo}=\insize{\tderiv} + \insize{\tderivtwo}$.
			
			\item \emph{Abstraction}, \ie: $$\tderiv = \infer[\typingruleAbs]{\typctx, \var \hastype \mtype \ctypes {k} \cla{}\vartwo\tmp \hastype \mtypetwo \typearrowp{\clr{}} \ltype}{\rho\derives\typctx, \var \hastype \mtype, \vartwo \hastype \mtypetwo \ctypes {k} \tmp \hastype \ltype}$$
			
			By induction, there exists a derivation  $\rho'$ of final judgment $\typctx,\vartwo\hastype\mtypetwo \ctypes {k+k'} \tmp\isub\var\tmtwo \hastype \ltype$ such that $\insize{\rho'}=\insize{\rho} +\insize{\tderivtwo}$, from which we can conclude obtaining $\tderiv'$ by applying the derivation rule $\typingruleAbs$. We observe that $\insize{\tderiv'}=\insize{\rho'}=\insize{\rho} +\insize{\tderivtwo}= \insize{\tderiv} +\insize{\tderivtwo}$.
			
			\item \emph{Application}, \ie: 
			$$	\infer[\typingruleApp]{\typectx_1 \uplus \typectx_2, \var \hastype \mtype  \ctypes {k} \capp{}{\tm_1}{\tm_2} \hastype \ltype}{ \rho_1\derives\typectx_1, \var\hastype\mtype_1 \ctypes{l_1}{\tm_1} \hastype \mtype \typearrowp{\clr{}} \ltype & \rho_2 \derives\typectx_2, \var\hastype\mtype_2  \ctypes {l_2} {\tm_2} \hastype \mtype  }$$ where $k =  l_1 + l_2 + \Kr{\clr{}}{\clrtwo{}}$.
			
			We have that $\mtype = \mtype_1 \mplus \mtype_2$. By \reflemma{types-splitting-multisets}, the derivation $\tderivtwo$ splits into two derivations of final judgments $\tderivtwo_1 \derives \typctxtwo_1 \ctypes {k'_1} \tmtwo \hastype \mtype_1$ and $\tderivtwo_2 \derives \typctxtwo_2 \ctypes {k'_2} \tmtwo \hastype \mtype_2$ such that $k' = k'_1 + k'_2$, and $\insize{\tderivtwo}=\insize{\tderivtwo'} + \insize{\tderivtwo''}$.
			
			By \ih, there are two derivations of final judgment  $\rho_1\derives\typectx_1 \ctypes  {l_1+k'_1} {\tm_1\isub\var\tmtwo} \hastype \mtype \typearrowp{\clr{}} \ltype$ and  $\rho_2\derives \typectx_2 \ctypes {l_2+ k'_2} {\tm_2\isub\var\tmtwo} \hastype \mtype $, such that $\insize{\rho_i'}=\insize{\rho_i} +\insize{\tderivtwo_i}$. 
			We conclude obtaining $\tderiv'$ by applying to these two derivations the rule $\typingruleApp$, as $k+k' = l_1 + k'_1+ l_2 + k'_2+\Kr{\clr{}}{\clrtwo{}}$. We observe that $\insize{\tderiv'}=1+\insize{\rho'_1} + \insize{\rho'_2} =1+\insize{\rho_1} +\insize{\tderivtwo_1} + \insize{\rho_2} +\insize{\tderivtwo_2}= \insize{\tderiv} +\insize{\tderivtwo}$.
			
		\end{enumerate}
		
		\item Multi types, \ie $\gtype \defeq [\ltype_i]_{i\in I}$ with $\iset$ finite. The last rule of the derivation must be $\typingruleMany$:		
		$$
		\infer[\typingruleMany]{\uplus_{i\in I}\typectx_i, \var\hastype\mtype \ctypes{\sum_{i\in I} l_i}  \tm \hastype [\ltype_i]_{i\in I}}
		{(\tderiv_i \derives \typectx_i, \var\hastype\mtype_i \ctypes {l_i} \tm \hastype \ltype_i)_{i\in I}}$$		
		with $\mtype = \uplus_{i\in I} \mtype_i$ and $k = \sum_{i\in I}l_i$ and $\uplus_{i\in I}\typectx_i = \typctx$.
		
		By the multiset splitting lemma (\reflemma{types-splitting-multisets}), the derivation $\tderivtwo$ for $\tmtwo$ in the hypotheses splits in several derivations of final judgments $\tderivtwo_i \derives \typctxtwo_i \ctypes {k'_i} \tmtwo \hastype \mtype_i$ such that $k' = \sum_{i\in I}k'_i$ and $\uplus_{i\in I}\typctxtwo_i = \typctxtwo$, and $\sum\insize{\tderiv_i}=\insize{\tderivtwo}$.
		
		Then by induction hypothesis, there exist several derivations $\tderivp_i \derives \typectx_i \cup \typctxtwo_i \ctypes  {l_i+k'_i} {\tm\isub\var\tmtwo} \hastype \ltype_i$, such that $\insize{\tderiv_{i}'} = \insize{\tderiv_{i}} + \insize{\tderivtwo_i}$. The derivation $\tderivp$ of the statement is obtained by applying rule $\typingruleMany$ to the family of derivations $\{\tderivp_i\}_{i\in I}$, as follows:
		$$
		\infer[\typingruleMany]{\uplus_{i\in I}\typectx_i\cup\typctxtwo_i \ctypes{\sum_{i\in I} (l_i+k'_i)}  \tm\isub\var\tmtwo \hastype [\ltype_i]_{i\in I}}
		{(\tderivp_i \exder \typectx_i \cup \typctxtwo_i \ctypes  {l_i+k'_i} {\tm\isub\var\tmtwo} \hastype \ltype_i )_{i\in\iset}}$$		
		Note indeed that $k+k' = \sum_{i\in I} (l_i + k'_i)$ and $\uplus_{i\in I}\typectx_i\uplus\typctxtwo_i = \typctx \uplus \typctxtwo$. Moreover, $\insize{\tderiv'}= \sum_i \insize{\tderiv_{i}'} = \sum_i\insize{\tderiv_{i}} + \insize{\tderivtwo_i}  = \insize{\tderiv} + \insize{\tderivtwo}$.
	\end{enumerate}
\end{proof}

\begin{lemma}[Merging multisets w.r.t.\ derivations]
	\label{l:silly-merging-multisets}
	Let $\tm\in\Lambdac$. Consider two derivations:
	\begin{itemize}
		\item $\tderiv_{\mtypetwo} \derives  \typectx_{\mtypetwo} \ctypes{k_1} \tm : \mtypetwo$, and
		\item $\tderiv_{\mtypethree} \derives  \typectx_{\mtypethree} \ctypes{k_2} \tm : \mtypethree$.
	\end{itemize}
	Then there exists a derivation $\tderiv_{\mtypetwo} \derives  \typectx_{\mtypetwo}\uplus\typctx_{\mtypethree} \ctypes{k_1+k_2} \tm : \mtypetwo\mplus \mtypethree$.
\end{lemma}

\begin{proof}
	The last rules of $\tderiv_{\mtypetwo} \derives  \typectx_{\mtypetwo} \ctypes{k_1} \tm : \mtypetwo$ and $\tderiv_{\mtypethree} \derives  \typectx_{\mtypethree} \ctypes{k_2} \tm : \mtypethree$ can only be the rule $\typingruleMany$, thus it is enough to re-group their hypotheses.
\end{proof}
\pagebreak
\begin{lemma}[Anti-substitution]
	\label{l:types-anti-substitution}
	Let $\tderiv \derives  \typectx\ \uplus \typectxtwo \ctypes{k+k'} \tm\isub\var\tmtwo : \gtype$. Then there exist
	\begin{itemize}
		\item a multi type $\mtype$,
		\item a derivation $\tderiv_\tm\derives{\typectx,\var:\mtype}\ctypes{k}{\tm}:{\gtype}$, and 
		\item a derivation $\tderiv_\tmtwo\derives{\typectxtwo}\ctypes{k'} {\tmtwo}:{\mtype}$.
	\end{itemize}
\end{lemma}

\begin{proof}
	By lexicographic induction on $(\gtype,\tm)$. Cases of $\gtype$:
	\begin{enumerate}
		\item Linear types, \ie $\gtype \defeq \ltype$. Cases of the last derivation rule:
		
		\begin{enumerate}
			\item \emph{Unsubstituted variable}, \ie: $\vartwo\isub\var\tmtwo = \vartwo$ and
			
			$$\tderiv = \infer[\typingruleAx]{\vartwo \hastype [\ltype] \ctypes {0} \vartwo \hastype \ltype}{}$$
			Then, taking $\mtype \defeq \zero$, $\tderiv_\tmtwo$ must be of the following shape: 
			
			$$\infer[\typingruleMany]{\typctxtwo \ctypes {0} \tmtwo \hastype \zero}{}$$ 
			
			Hence $\tderiv_\tm \defeq \tderiv$ works.
			\item \emph{Substituted variable}, \ie: $\var\isub\var\tmtwo = \tmtwo$ and $\tderiv \derives \typctxtwo \ctypes {k'} \tmtwo \hastype \ltype$. We take $\mtype\defeq[\ltype]$. Then:
			\[
			\infer{\tderiv_\tmtwo \derives \typctxtwo \ctypes {k'} \tmtwo \hastype [\ltype]}{\tderiv \derives \typctxtwo \ctypes {k'} \tmtwo \hastype \ltype}
			\]
			
			and
			
			\[
			\infer[\typingruleAx]{\tderiv_\tm=\var \hastype [\ltype] \ctypes {0} \var \hastype \ltype}{}\]
			
			\item \emph{Abstraction}, \ie: $$\tderiv = \infer[\typingruleAbs]{\typctx\uplus\typectxtwo \ctypes {k+k'} \cla{}\vartwo\tmp\isub\var\tmtwo \hastype \mtypetwo \typearrowp{\clr{}} \ltype}{\rho\derives\typctx\uplus\typectxtwo, \vartwo \hastype \mtypetwo \ctypes {k+k'} \tmp\isub\var\tmtwo \hastype \ltype}$$
			
			By induction, there exists a derivation  $\rho'$ of final judgment $\typctx,\vartwo\hastype\mtypetwo,\var:\mtype \ctypes {k} \tmp \hastype \ltype$ and a type derivation $\tderiv_\tmtwo$ of final judgment $\typctxtwo\ctypes {k'} \tmtwo \hastype \mtype$. Then

			$$\tderiv_\tm = \infer[\typingruleAbs]{\typctx, \var \hastype \mtype \ctypes {k} \cla{}\vartwo\tmp \hastype \mtypetwo \typearrowp{\clr{}} \ltype}{\rho'\derives\typctx, \var \hastype \mtype, \vartwo \hastype \mtypetwo \ctypes {k} \tmp \hastype \ltype}$$
			
			\item \emph{Application}, \ie: 
			
			$$	\infer[\typingruleApp]{\typectx_1 \uplus \typectx_2\uplus\typectxtwo_1\uplus\typectxtwo_2  \ctypes  {k+k'+ \Kr{\clr{}}{\clrtwo{}}} (\capp{}{\tm_1}{\tm_2})\isub\var\tmtwo \hastype \ltype}{ \rho_1\derives\typectx_1\uplus\typectxtwo_1 \ctypes  {k_1+k_1'} {\tm_1}\isub\var\tmtwo \hastype \mtypetwo \typearrowp{\clr{}} \ltype & \rho_2 \derives\typectx_2\uplus\typectxtwo_2  \ctypes {k_2+k_2'} {\tm_2}\isub\var\tmtwo \hastype \mtypetwo  }$$ where $k =  k_1+k_2$ and $k' =  k_1'+k_2'$.
			
			By \ih, we have a type derivation $\rho_1'$ such that $\rho_1'\derives\typectx_1, \var\hastype\mtype_1 \ctypes  {k_1} {\tm_1} \hastype \mtypetwo \typearrowp{\clr{}} \ltype$, a type derivation $\tderiv_{\tmtwo_1}$ such that $\tderiv_{\tmtwo_1} \derives \typctxtwo_1 \ctypes {k'_1} \tmtwo \hastype \mtype_1$, a type derivation $\rho_2'$ such that $\rho_2' \derives\typectx_2, \var\hastype\mtype_2  \ctypes {k_2} {\tm_2} \hastype \mtypetwo$, and a type derivation $\tderiv_{\tmtwo_2}$ such that $\tderiv_{\tmtwo_2} \derives \typctxtwo_2 \ctypes {k'_2} \tmtwo \hastype \mtype_2$. $\tderiv_\tmtwo$ is obtained by merging $\tderiv_{\tmtwo_1}$ and $\tderiv_{\tmtwo_2}$ (\reflemma{silly-merging-multisets}). $\tderiv_\tm$ is as follows:
			$$	
			\infer[\typingruleApp]{\typectx_1 \uplus \typectx_2, \var \hastype \mtype_1\uplus\mtype_2  \ctypes  {k+ \Kr{\clr{}}{\clrtwo{}}} \capp{}{\tm_1}{\tm_2} \hastype \ltype}{ \rho_1'\derives\typectx_1, \var\hastype\mtype_1 \ctypes  {k_2} {\tm_1} \hastype \mtypetwo \typearrowp{\clr{}} \ltype & \rho_2' \derives\typectx_2, \var\hastype\mtype_2  \ctypes {k_2} {\tm_2} \hastype \mtypetwo  }
			$$
			
		\end{enumerate}
		
		\item Multi types, \ie $\gtype \defeq [\ltype_i]_{i\in I}$ with $\iset$ finite. The last rule of the derivation must be $\typingruleMany$:		
		$$
		\infer[\typingruleMany]{\uplus_{i\in I}\typectx_i\cup\typctxtwo_i \ctypes{\sum_{i\in I} (k_i+k'_i)}  \tm\isub\var\tmtwo \hastype [\ltype_i]_{i\in I}}
		{(\tderivp_i \exder \typectx_i \cup \typctxtwo_i \ctypes  {k_i+k'_i} {\tm\isub\var\tmtwo} \hastype \ltype_i )_{i\in\iset}}
		$$	
		
		By applying the \ih, we obtain for each $i\in I$: $\tderiv_i \derives \typectx_i, \var\hastype\mtype_i \ctypes {k_i} \tm \hastype \ltype_i$ and
		$\tderivtwo_i \derives \typctxtwo_i \ctypes {k'_i} \tmtwo \hastype \mtype_i$. 
		$\tderiv_\tmtwo$ is obtaining by merging the $\tderivtwo_i$ (\reflemma{silly-merging-multisets}). $\tderiv_\tm$ is as follows:	
		$$
		\infer[\typingruleMany]{\uplus_{i\in I}\typectx_i, \var\hastype\mtype \ctypes{\sum_{i\in I} k_i}  \tm \hastype [\ltype_i]_{i\in I}}
		{(\tderiv_i \exder \typectx_i, \var\hastype\mtype_i \ctypes {k_i} \tm \hastype \ltype_i)_{i\in I}}$$		
		with $\mtype = \uplus_{i\in I} \mtype_i$.
	\end{enumerate}
\end{proof}

In the following proof, we use the following grammar generating head contexts:
\[\begin{array}{r@{\hspace{.15cm}}r@{\hspace{.1cm}}l@{\hspace{.1cm}}ll}
	\textsc{Weak Head Contexts} & \hauxctx & \grameq & \ctxhole \mid \hauxctx\tm
	\\
	\textsc{Head Contexts} & \hctx & \grameq & \la\var\hctx \mid \hauxctx
	
\end{array}\]
\gettoappendix{prop:ch-subject}
\begin{proof}
	\applabel{prop:ch-subject}
	\begin{enumerate}
		\item \begin{enumerate}
			\item \emph{Root step}, \ie $\tm = \capp{}{(\cla{}\var\tmthree)}\tmtwo\tohch \tmthree\isub\var\tmtwo=\tmp$.
			
			Then the last rule of the derivation $\tderiv$ is $\typingruleApp$ and is followed by the $\typingruleAbs$ rule on the left:
			
			$$	\infer[\typingruleApp]{\typectx \uplus \typectxtwo \ctypes {k} \capp{}{(\cla{}\var\tmthree)}\tmtwo \hastype \ltype}{ \infer{\typectx \ctypes {k_1} \cla{}\var\tmthree \hastype \mtype \typearrowp{\clr{}} \ltype}{\typectx, \var \hastype \mtype\ctypes {k_1} \tmthree \hastype \ltype} & \typectxtwo \ctypes{k_2} \tmtwo \hastype \mtype  }$$ where $k =  k_1 + k_2 + \Kr{\clr{}}{\clrtwo{}}$.
			
			By the Substitution \reflemma{types-substitution}, we have that there exists a derivation $\tderiv' \derives \typctx \uplus \typctxtwo \ctypes{k_1+k_2} \tmthree \isub\var\tmfour \hastype \ltype$ such that $\insize{\tderiv'}=\insize{\tderiv_s}+\insize{\tderiv_u}=\insize{\tderiv}-1$. We conclude by observing that if $\tm \tohnoint\tmp$ then $k = k_1 + k_2$, and otherwise if $\tm \tohint \tmp$ then $k_1+k_2=k-1$.
			
			\item \emph{Contextual closure}. We have two subcases:
			\begin{enumerate}
				\item Weak contexts, \ie $\tm= \capp{}{\hauxctxp{s}}\tmtwo\toh \capp{}{\hauxctxp{s'}}\tmtwo=\tm'$. Then the last rule of $\tderiv$ is $\typingruleApp$:
				\[
				\infer{\typectx \uplus \typectxtwo \ctypes {k} \capp{}{\hauxctxp{s}}\tmtwo \hastype \ltype}{\tderivtwo\derives
					\typectx \ctypes {k_1} \hauxctxp{s} \hastype \mtype \typearrowp{\clr{}} \ltype & 
					\typectxtwo \ctypes{k_2} \tmtwo \hastype \mtype } 
				\]
				where $k =  k_1 + k_2+ \Kr{\clr{}}{\clrtwo{}}$. By \ih, there exists a derivation $\tderivtwo'\derives\typectx\ctypes{k_1'} \hauxctxp{s'} \hastype \mtype \typearrowp{\clr{}} \ltype$, where $\insize{\tderivtwo'}=\insize{\tderivtwo}-1$ and 
				\[
				k_1' = \begin{cases}
					k_1, & if ~\hauxctxp{s} \tohnoint\hauxctxp{s'},
					\\
					k_1-1, & if~ \hauxctxp{s} \tohint \hauxctxp{s'}.
				\end{cases}
				\] Then, we can build the type derivation $\tderiv'$ as:
				\[
				\infer{\typectx \uplus \typectxtwo \ctypes {k'} \capp{}{\hauxctxp{s'}}\tmtwo \hastype \ltype}{\tderivtwo'\derives
					\typectx \ctypes {k_1'} \hauxctxp{s'} \hastype \mtype \typearrowp{\clr{}} \ltype & 
					\typectxtwo \ctypes{k_2} \tmtwo \hastype \mtype } 
				\]
				where clearly $\size{\tderiv'}=\size{\tderiv}-1$ and $k' = \begin{cases}
					k, & \textrm{if }~\capp{}{\hauxctxp{s}}\tmtwo \tohnoint\capp{}{\hauxctxp{s'}}\tmtwo,
					\\
					k-1, & \textrm{if }~ \capp{}{\hauxctxp{s}}\tmtwo \tohint \capp{}{\hauxctxp{s'}}\tmtwo.
				\end{cases}$
				\item Head contexts, \ie $\tm= \cla{}{\var}{H\ctxholep{s}}\toh \cla{}{\var}{H\ctxholep{s'}}=\tm'$. Analogous.
			\end{enumerate}
		\end{enumerate}
		\item \begin{enumerate}
			\item \emph{Root step}, \ie $\tm = (\cla{}\var\tmthree)\cdot^\clrtwo{}\tmtwo\rtobnoint \tmthree\isub\var\tmtwo=\tmp$.
			
			By the Anti-Substitution \reflemma{types-anti-substitution}, there exist
			$\tderivtwo \derives \typctx',\var:\mtype \ctypes { k_1}  \tmthree \hastype \ltype$ and $\rho \derives \typctx'' \ctypes {k_2}  \tmtwo \hastype \mtype$ such that $\typctx=\typctx'\uplus\typctx''$ and $k=k_1+k_2$.
			
			Then we can build the type derivation $\tderiv\derives \typctx \ctypes { k'}  \tm \hastype \ltype$ as follows:
			
			$$	\infer[\typingruleApp]{\typectx' \uplus \typectx'' \ctypes {k'} (\cla{}\var\tmthree)\cdot^\clrtwo{}\tmtwo \hastype \ltype}{ \infer{\typectx' \ctypes {k_1} \cla{}\var\tmthree \hastype \mtype \typearrowp{\clr{}} \ltype}{\sigma\derives\typectx', \var \hastype \mtype\ctypes {k_1} \tmthree \hastype \ltype} & \rho\derives\typectx'' \ctypes{k_2} \tmtwo \hastype \mtype  }$$
			where $k'=k+\Kr{\clr{}}{\clrtwo{}}$.
			\item \emph{Contextual closure}. We have two subcases:
			\begin{enumerate}
				\item Weak contexts, \ie $\tm= \capp{}{\hauxctxp{s}}\tmtwo\toh \capp{}{\hauxctxp{s'}}\tmtwo=\tm'$. Then the last rule of $\tderiv'$ is $\typingruleApp$:
				\[
				\infer{\typectx \uplus \typectxtwo \ctypes {k} \capp{}{\hauxctxp{s'}}\tmtwo \hastype \ltype}{\tderivtwo'\derives
					\typectx \ctypes {k_1} \hauxctxp{s'} \hastype \mtype \typearrowp{\clr{}} \ltype & 
					\typectxtwo \ctypes{k_2} \tmtwo \hastype \mtype } 
				\]
				where $k = k_1 + k_2 + \Kr{\clr{}}{\clrtwo{}}$. By \ih, there exists a derivation $\tderivtwo\derives\typectx\ctypes{k_1'} \hauxctxp{s} \hastype \mtype \typearrowp{\clr{}} \ltype$
				Then, we can build $\tderiv$ as follows:
				\[
				\infer{\typectx \uplus \typectxtwo \ctypes {k'} \capp{}{\hauxctxp{s}}\tmtwo \hastype \ltype}{\tderivtwo\derives
					\typectx \ctypes {k_1'} \hauxctxp{s} \hastype \mtype \typearrowp{\clr{}} \ltype & 
					\typectxtwo \ctypes{k_2} \tmtwo \hastype \mtype } 
				\]
				\item Head contexts, \ie $\tm= \cla{}{\var}{H\ctxholep{s}}\toh \cla{}{\var}{H\ctxholep{s'}}=\tm'$. Analogous.
			\end{enumerate}
		\end{enumerate}
	\end{enumerate}
\end{proof}

\begin{proposition}[Head normal forms are typable with zero weight]\label{prop:nf-are-typable-tight}
	Let $\htm\in\Lambdac$ be a $\hchsym$-normal form.
	Then, there exists a derivation $\typctx \ctypes{0} \htm \hastype \ltype$.
\end{proposition}

\begin{proof}\hfill
%
%
%
%
%
%
Let $\htm =\manyclam{n}{\var}\manycapp{m}{\var}{\tm}$. Set $\ltype'\defeq\mset{\emptytype \typearrowp{\clrtwo{1}} \cdots \emptytype \typearrowp{\clrtwo{m}} \vartype}$. 
Firstly, we build the following derivation:
	\[
	\infer=[m \typingruleApp]{\var\hastype\mset{\ltype'} \ctypes{0} \manycapp{m}{\var}{\tm} \hastype \vartype}
			{\infer{\vdots}{
			\infer[\typingruleApp]{\var\hastype\mset{\ltype'} \ctypes{0} \capp{1}\var {\tm_1} \hastype \emptytype \typearrowp{\clrtwo{2}} \cdots \emptytype \typearrowp{\clrtwo{m}} \vartype
			}{
				\var\hastype[\ltype'] \ctypes{0} \var \hastype \ltype' & \infer[\typingruleMany]{\emptytypectx \ctypes{0} \tm_{1} \hastype \emptytype}{}}
			}
			& \infer[\typingruleMany]{\emptytypectx \ctypes{0} \tm_{m} \hastype \emptytype}{}
			}
	\]
Then, there are two cases:
\begin{itemize}
\item $\var = \var_i$ for some $i\in\set{1,\ldots,n}$. For the sake of simplicity, let us say that $i=n$, the other cases are analogous. Then, we obtain the following derivation:
	\[
\infer=[(n-1) \typingruleAbs]{\ctypes{0} \manyclam{n}{\var}\manycapp{m}{\var}{\tm} \hastype \emptytype \typearrowp{\clr{1}} \cdots \emptytype \typearrowp{\clr{n-1}}\mset{\ltype'}\typearrowp{\clr{n}} \vartype
}{
	\infer[\typingruleAbs]{ \ctypes{0} \cla{n}{\var_n}\manycapp{m}{\var}{\tm} \hastype \mset{\ltype'}\typearrowp{\clr{n}} \vartype
	}{
			\var\hastype\mset{\ltype'} \ctypes{0} \manycapp{m}{\var}{\tm} \hastype \vartype
	}
}
	\]

\item $\var \neq \var_i$ for all $i\in\set{1,\ldots,n}$.  Then, we obtain the following derivation:
	\[
\infer=[n\typingruleAbs]{\var\hastype\mset{\ltype'} \ctypes{0} \manyclam{n}{\var}\manycapp{m}{\var}{\tm} \hastype \emptytype \typearrowp{\clr{1}} \cdots \emptytype \typearrowp{\clr{n}} \vartype
}{
			\var\hastype\mset{\ltype'} \ctypes{0} \manycapp{m}{\var}{\tm} \hastype \vartype
}
	\]
\end{itemize}

\end{proof}

\gettoappendix{thm:soundandcomplete}
\begin{proof}
	\applabel{thm:soundandcomplete}
	
	\begin{enumerate}
		\item By induction on $\size{\tderiv}$ and case analysis on whether $\tm$ is head normal. If $\tm$ is head
		normal then $\tm$ is head normalizable in $0\leq k$ interaction steps. If $\tm\tohch\tmtwo$ then there two cases:
		\begin{itemize}
			\item $\tm\tohnoint\tmtwo$. Then by quantitative subject reduction
			(Prop.~\ref{prop:ch-subject}(1)), there is $\tderivtwo\derives\typectx\ctypes k\tmtwo \hastype\ltype$ such that $\insize{\tderiv} = \insize{\tderivtwo}+1$. By i.h., $\tmtwo$ is head normalizable in less than $k$ interaction steps. Then, the same holds for $\tm$.
			\item $\tm\tohint\tmtwo$. Then by quantitative subject reduction
			(Prop.~\ref{prop:ch-subject}(1)), there is $\tderivtwo\derives\typectx\ctypes {k-1}\tmtwo \hastype\ltype$ such that $\insize{\tderiv} = \insize{\tderivtwo}+1$. By i.h., $\tmtwo$ is head normalizable in less than $k-1$ interaction steps. Then, $\tm$ is head normalizable in less than $k$.
		\end{itemize}
		\item We have that $\tm \tohch^m \htm$, where $\hnf$ is a head normal form. By induction on $m$. 
		Cases:
		\begin{enumerate}
			\item If $m=0$, then $\tm = \htm$. Then we conclude by Proposition~\ref{prop:nf-are-typable-tight}. 
			
			\item If $m>0$, then $\tm \tohch \tmtwo \tohch^{m-1} \htm$. By \ih, there 
			exists 
			$\tderivtwo\derives\typectx \ctypes{k''} {\tmtwo}:{\ltype}$ and $\tmtwo \bshcol{k''}$. By subject expansion (Prop.~\ref{prop:ch-subject}(2)), there exists 
			$\tderiv\derives\typectx \ctypes{k'} {\tm}:{\ltype}$. By subject reduction (Prop.~\ref{prop:ch-subject}(1)), if $\tm \tohnoint \tmtwo$ then $k'=k''=k$, otherwise if $\tm \tohint \tmtwo$, then $k'= 1 + k''=k$.
		\end{enumerate}
	\end{enumerate}
\end{proof}

\gettoappendix{lem:whitercheaperrel}\applabel{lem:whitercheaperrel}
\begin{proof}
	\begin{enumerate}[(i)]
	\item $(\Rightarrow)$ By an easy induction on a derivation of $\ltype\whiterpp{\posorneg}{0}\ltype$ (resp.\ $\typctx\whiterpp{\posorneg}{0}\typctx'$ or $\Pair{\typctx}{\ltype}\whiterpp{\posorneg}{0}\Pair{\typctx'}{\ltype'}$), using the fact that all coefficients must be 0. 
	
	$(\Leftarrow)$ Straightforward.
	\item $(\Rightarrow)$ We consider the case of positive polarity, the other case being symmetric. 
	By unpacking the definitions, we get $\typctx'\whiternegp{k_1}\typctx$, $\mtype'\whiternegp{k_2}\mtype$ and $\ltype'\whiterposp{k_3} \ltype$ with $k = k_1+k_2+k_3$. Conclude.
	
	$(\Leftarrow)$ Straightforward.	
	\item By (Inversion), it is sufficient to prove that $\whiterpp\posorneg{k}$ is transitive on linear types and on multi types. 
	We prove the following statements simultaneously ($\forall k_1,k_2\in\nat$):
	\begin{equation}\label{eq:transontypes}
	\ltype_1\whiterpp\posorneg{k_1}\ltype_2 \textrm{ and } \ltype_2\whiterpp\posorneg{k_2}\ltype_3 \qquad\!\!\Rightarrow\quad \ltype_1\whiterpp\posorneg{k_1+k_2} \ltype_3
	\end{equation}
	\begin{equation}\label{eq:transonmtypes}
	\mtype_1\whiterpp\posorneg{k_1}\mtype_2 \textrm{ and } \mtype_2\whiterpp\posorneg{k_2}\mtype_3 \quad\Rightarrow\quad \mtype_1\whiterpp\posorneg{k_1+k_2} \mtype_3
	\end{equation}	
	We proceed by induction on the derivations of $\ltype\whiterpp\posorneg{k_1}\ltype'$ and $\mtype\whiterpp\posorneg{k_1}\mtype'$, and call IH1 and IH2 the respective induction hypotheses. We split into cases depending on the last applied rule.

	\eqref{eq:transontypes}
	\ul{Base case}: $\alpha\whiterpp\posorneg0\alpha$, i.e.\ $\ltype_1 = \ltype_2= \alpha$ and $k_1 = 0$. By hypothesis, $\alpha\whiterpp\posorneg{k_2+0}\ltype_3$.
	
	\ul{Induction.} Case $\mtype'\typearrowp{\circ}\ltype'\whiterposp{k + l + 1} \mtype\typearrowp{\bullet}\ltype$, i.e.\ $\ltype_1 = \mtype'\typearrowp{\circ}\ltype'$, $\ltype_2 = \mtype\typearrowp{\bullet}\ltype$ with ${\tt a} = +$ and $k_1 = k + l + 1$. This holds because $\mtype'\whiternegp k \mtype$ and $\ltype'\whiterposp l\ltype$. Now, $\mtype\typearrowp{\bullet}\ltype\whiterposp{k_2}\ltype_3$ entails $\ltype_3 = \mtype''\typearrowp{\bullet}\ltype''$, for some  $\mtype'',\ltype''$ such that $\mtype\whiternegp{k'}\mtype''$ and $\ltype\whiterposp{l'}\ltype''$ with $k_2 = k'+l'$. 
	From $\mtype'\whiternegp k \mtype$ and $\mtype\whiternegp{k'}\mtype''$ we get $\mtype'\whiternegp{k+k'} \mtype''$ by IH2. 
	From $\ltype'\whiterposp l \ltype$ and $\ltype\whiterposp{l'}\ltype''$ we get $\ltype'\whiterposp{l+l'} \ltype''$ by IH1.
	Hence, we obtain $\mtype'\typearrowp{\circ}\ltype'\whiterposp{k+k'+l+l' + 1(= k_1+k_2+1)}\mtype''\typearrowp{\bullet}\ltype''$.

	Case $\mtype'\typearrowp{\colr}\ltype'\whiterposp{k + l } \mtype\typearrowp{\colr}{}\ltype$, i.e.\ $\ltype_1 = \mtype'\typearrowp{\colr}\ltype'$, $\ltype_2 = \mtype\typearrowp{\colr}{}\ltype$ with ${\tt a} = +$ and $k_1 = k + l$. This holds because $\mtype'\whiternegp k \mtype$ and $\ltype'\whiterposp l\ltype$.
	From $\mtype\typearrowp{\colr}{}\ltype\whiterposp{k_2}\ltype_3$ we obtain $\ltype_3 = \mtype''\typearrowp{\colrtwo}{}\ltype''$. 
	There are three subcases:
	\begin{itemize}
	\item $\colr = \colrtwo$. Then $\mtype\whiternegp{k'}\mtype''$ and $\ltype\whiterposp{l'}\ltype''$ with $k_2 = k'+l'.$
	From $\mtype'\whiternegp k \mtype$ and $\mtype\whiternegp{k'}\mtype''$ we get $\mtype'\whiternegp{k+k'} \mtype''$ by IH2. 
	From $\ltype'\whiterposp l \ltype$ and $\ltype\whiterposp{l'}\ltype''$ we get $\ltype'\whiterposp{l+l'} \ltype''$ by IH1.
	Hence $\mtype'\typearrowp{\colr}\ltype'\whiterposp{k+k'+l+l'(= k_1+k_2)}\mtype''\typearrowp{\colr}\ltype''$.
	\item $\colr = \bullet$ and $\colrtwo = \circ$. Vacuous, as no rule is applicable.
	\item $\colr = \circ$ and $\colrtwo = \bullet$. Then $\mtype\whiternegp{k'}\mtype''$ and $\ltype\whiterposp{l'}\ltype''$ with $k_2 = k' + l' +1.$
	From $\mtype'\whiternegp k \mtype$ and $\mtype\whiternegp{k'}\mtype''$ we get $\mtype'\whiternegp{k+k'} \mtype''$ by IH2. 
	From $\ltype'\whiterposp l \ltype$ and $\ltype\whiterposp{l'}\ltype''$ we get $\ltype'\whiterposp{l+l'} \ltype''$ by IH1.
	Hence, we obtain $\mtype'\typearrowp{\circ}\ltype'\whiterposp{k+k'+l+l' + 1(= k_1+k_2 + 1)}\mtype''\typearrowp{\bullet}\ltype''$.

	\end{itemize}
	Case $\mtype'\typearrowp{\colr}\ltype'\whiternegp{k+l} \mtype\typearrowp{\colr}\ltype$. Analogous to the previous case.

	\eqref{eq:transonmtypes} The only case is $[\ltype'_1,\dots,\ltype'_n] \whiterpp\posorneg{k_1+\cdots+k_n} [\ltype_1,\dots,\ltype_n]$, and it follows straightforwardly from the IH1. \qed
	\end{enumerate}	
\end{proof}

\section{Proofs of \refsect{bohm-trees} (\bohm trees)}

The dual of \emph{Kronecker's delta} $\Kr{\cdot}{\cdot}$ is defined by ($\forall \colr,\colrtwo\in\set{\circ\bullet}$):
\[
\Kr{\colr}{\colrtwo}:=\begin{cases}
1,&\textrm{ if }\colr=\colrtwo,\\
0,&\textrm{otherwise}.\\
\end{cases}
\]
and will be used often in the following proofs to express the dependency between the indices of the derivations and the colors appearing on the types.

\gettoappendix{l:eta-id} 
\begin{proof}\applabel{l:eta-id} We proceed by mutual induction and call HI1 and HI2 the corresponding induction hypotheses. 
	More precisely, we prove (1) by induction on a derivation of $\Gamma\ctypes{k}  \monoToBlue{\tm} : \ltype$ and (2) on a derivation of $\Gamma\ctypes{k} \monoToBlue{\tm} : \mtype$.
	Wlog we can consider $\tm$ to be in head normal form, because of subject reduction and--crucially--the fact that SR does not increase the size of the typing derivation. 
	(Note also that silent head SR does not change the index $k$ of the typing derivation.)
	\begin{enumerate}[(i)]
		\item 
		\begin{itemize}
			\item Base. $\tm = x$. Then $\Gamma = x\hastype[\ltype]$, 
			$\ltype \whiterneg_0 \ltype$,  and $\ltype \whiterpos_0 \ltype$. Take 
			$k' = k'' = 0$.
			\item
			Inductive case. $\tm = \lam z_1\ldots z_n.x\tmtwo_1\cdots 
			\tmtwo_n$ with $\tmtwo_i \slbtleq z_i \neq x$ and $z_i$ do not occur free in  the $\tmtwo_j$'s, for all $i \neq j$. 
			By definition of $\monoToBlue{(\cdot)}$, we have 
			\[
			\monoToBlue{\tm} = 
			\manyblam{n}{z}\bapp{\bapp{\bapp{x}{\monoToBlue{\tmtwo_1}}} 
			{\cdots}}{\monoToBlue{\tmtwo_n}}
			\] 
			whence $\ltype$ must have shape
			$\ltype = 
			\mtype'_1\typearrowp{\blueclr}\cdots\typearrowp{\blueclr}\mtype'_n\typearrowp{\blueclr}
			 \ltype_0$. Thus:
			\[
			\infer{\Gamma\ctypes{\sum_{i=1}^n (\Kr{\clr{i}}{\blueclr} + 
			k_i)} 
			\manyblam{n}{z}\bapp{\bapp{\bapp{x}{\monoToBlue{\tmtwo_1}}}{\cdots}}{\monoToBlue{\tmtwo_n}}
			 : 
			\mtype'_1\typearrowp{\blueclr}\cdots\mtype'_n\typearrowp{\blueclr}
			 \ltype_0}{
				\infer{\sum_{j=0}^n\Gamma_j
				\ctypes{\sum_{i=1}^n (\Kr{\clr{i}}{\blueclr} + k_i)}  
				\bapp{\bapp{\bapp{x}{\monoToBlue{\tmtwo_1}}}{\cdots}}{\monoToBlue{\tmtwo_n}}
				 : \ltype_0}{
					\Gamma_0 \ctypes{0} x : 
					\mtype_1\typearrowp{\clr{1}}\cdots\mtype_n\typearrowp{\clr{n}}
					 \ltype_0
					&
					(\Gamma_i \ctypes{k_i} \monoToBlue{\tmtwo_i} : 
					\mtype_i)_{i\leq n}
				}
			}
			\]
			with 
			$\sum_{j=0}^n\Gamma_j = \Gamma, z_1 :\mtype'_1,\dots,z_n:\mtype'_n $.

			From the axiom $\Gamma_0\ctypes{0} x : 
			\mtype_1\typearrowp{\clr{1}}\cdots\mtype_n\typearrowp{\clr{n}} 
			\ltype_0$, we get that $\Gamma_0 = x\hastype[\ltype']$ for 
			$\ltype' = 
			\mtype_1\typearrowp{\clr{1}}\cdots\mtype_n\typearrowp{\clr{n}} 
			\ltype_0$.
			By applying IH2 to $\Gamma_i \ctypes{k_i} 
			\monoToBlue{\tmtwo_i} : \mtype_i$, we obtain $\Gamma_i = z_i : \mtype'_i$, and the existence of $\mtypetwo_i$ such that 
			$\mtypetwo_i\whiterneg_{k'_i} \mtype'_i$ and 
			$\mtypetwo_i\whiterpos_{k''_i} \mtype_i$ with $k_i = k'_i+k''_i$. 
			This entails that $\Gamma = \Gamma_0 = x\hastype[\ltype']$.
			We conclude by taking 
			$\ltype'' = 
			\mtypetwo_1\typearrowp{\clr{1}}\cdots\mtypetwo_n\typearrowp{\clr{n}}
			 \ltype_0$ since
			\[
			\infer{\mtypetwo_1\typearrowp{\clr{1}}\cdots\mtypetwo_n\typearrowp{\clr{n}}
			 \ltype_0\whiterneg_{\sum_i 
			k''_i}\mtype_1\typearrowp{\clr{1}}\cdots\mtype_n\typearrowp{\clr{n}}
			 \ltype_0}{
				\mtypetwo_i\whiterpos_{k''_i} \mtype_i
				&
				\ltype_0\whiterneg_0\ltype_0
			}       
			\]
			and
			\[
			\infer{\mtypetwo_1\typearrowp{\clr{1}}\cdots\mtypetwo_n\typearrowp{\clr{n}}
			 \ltype_0\whiterpos_{\sum_i (\Kr{\clr{i}}{\blueclr} + 
			k'_i)}\mtype'_1\typearrowp{\blueclr}\cdots\mtype'_n\typearrowp{\blueclr}
			 \ltype_0}{
				\mtypetwo_i\whiterneg_{k'_i} \mtype'_i
				&
				\ltype_0\whiterpos_0\ltype_0
			}       
			\]
			
			Since $k'_i+k''_i = k_i$, we conclude $\sum_i (\Kr{\clr{i}}{\blueclr} + k_{i}) = \sum_i (\Kr{\clr{i}}{\blueclr} + k'_{i}) + \sum_ik''_{i}$.
			
		\end{itemize}
		
		\item 
		\begin{itemize}
			\item Base case. Namely, $\mtype = [\,]$ then $\Gamma = x\hastype [\,]$ and $k=0$. 
			Taking $k'=k''=0$, we get $[\,]\whiterpos_0 [\,]$ and 
			$[\,]\whiterneg_0 [\,]$.
			
			\item Induction case. Straightforward, from the IH1.\qed
		\end{itemize}
	\end{enumerate}
\end{proof}

\gettoappendix{lem:whitercheaper}

\begin{proof}\applabel{lem:whitercheaper} We proceed by induction on a derivation of $\typctx\ctypes{k} 
\monoToBlue{\tm} : \ltype$. Since $\tm$ is typable, it has a head normal 
form.
	By SR/SE, we can assume $\tm$ in head normal form without loosing 
	generality. 
	
	\begin{itemize}
		\item Base case $\tm = \lam x_1\dots x_n.y$. Since $\tm$ 
	contains no possibly infinite $\eta$-expansions, $\tm\slbtleq \tmtwo$ 
	entails 
	$\tm=_\beta \tmtwo$ and we are done by SR/SE on silent steps. 
	
	\item Induction case. Since $\tm\slbtleq \tmtwo$ we must have 
	\[ \tm = 
	\lam 
	x_1\dots x_{n}z_1\dots z_r.y\tm_1\cdots \tm_{m}\tmthree_1\cdots 
	\tmthree_r\textrm{ and }\tmtwo =_\beta 
	\lam x_1\dots x_{n}.y\tmtwo_1\cdots \tmtwo_{m}
	\] with $\tm_j\slbtleq 
	\tmtwo_j$, 
	$\tmthree_i\slbtleq z_i$ and $z_i$ fresh for $y\vec \tmtwo$.
	Because of the structure of $\tm$, $\ltype$ needs 
	to have the shape 
	$\ltype = \mtype_1\typearrowp{\bullet}\cdots\mtype_n\typearrowp{\bullet} 
	\mtype'_1\typearrowp{\bullet}\cdots\mtype'_r\typearrowp{\bullet}\ltype_0$.
	Moreover, there exists a decomposition $\typctx,\vec x : \vec \mtype = y : 
	[\ltype_1] + \sum_j\Delta_j$ for $\ltype_1 = 
	\mtypetwo_1\typearrowp{\clr{1}}\cdots\mtypetwo_m\typearrowp{\clr{m}} 
	\mtype''_1\typearrowp{\clrtwo{1}}\cdots\mtype''_r\typearrowp{\clrtwo{r}}\ltype_0$.
	 We have:
	\[
	\resizebox{\textwidth}{!}{%
	\infer{\typctx\ctypes{\sum_j(k_j+\Kr{\clr{j}}{\bullet}) + \sum_i(e_i 
	+\Kr{\clrtwo{i}}{\bullet})}\lam_\bullet \vec x\vec 
	z.y\bullet\monoToBlue{\tm}_1\bullet\cdot\bullet 
	\monoToBlue{\tm}_{m}\bullet\monoToBlue{\tmthree}_1\bullet\cdot\bullet 
	\monoToBlue{\tmthree}_r : 
	\mtype_1\typearrowp{\bullet}\cdots\mtype_n\typearrowp{\bullet} 
	\mtype'_1\typearrowp{\bullet}\cdots\mtype'_r\typearrowp{\bullet}\ltype_0}{
		\infer{
			\typctx, \vec x : \vec \mtype \ctypes{\sum_j 
			(k_j+\Kr{\clr{j}}{\bullet}) + \sum_i(e_i 
			+\Kr{\clrtwo{i}}{\bullet})} 
			\lam^\bullet \vec 
			z.y\bullet\monoToBlue{\tm}_1\bullet\cdot\bullet 
			\monoToBlue{\tm}_{m}\bullet\monoToBlue{\tmthree}_1\bullet\cdot\bullet
			\monoToBlue{\tmthree}_r: 
			\mtype'_1\typearrowp{\bullet}\cdots\mtype'_r\typearrowp{\bullet}\ltype_0
		}{
			\infer{\typctx, \vec x : \vec \mtype, \vec z : 
			\vec\mtype'\ctypes{\sum_j(k_j+\Kr{\clr{j}}{\bullet}) + 
			\sum_i(e_i + \Kr{\clrtwo{i}}{\bullet})}
				y\bullet\monoToBlue{\tm}_1\bullet\cdot\bullet 
				\monoToBlue{\tm}_{m}\bullet\monoToBlue{\tmthree}_1\bullet\cdot\bullet
				\monoToBlue{\tmthree}_r: \ltype_0 }{
				\infer{y : [\ltype_1] + \sum_j\Delta_j 
				\ctypes{\sum_j(k_j+\Kr{\clr{j}}{\bullet})}
					y\bullet\monoToBlue{\tm}_1\bullet\cdot\bullet 
					\monoToBlue{\tm}_{m} :  
					\mtype''_1\typearrowp{\clrtwo{1}}\cdots\mtype''_r 
					\typearrowp{\clrtwo{r}}\ltype_0}{
					y:[\ltype_1]\ctypes{0} y : 
					\mtypetwo_1\typearrowp{\clr{1}}\cdots\mtypetwo_m 
					\typearrowp{\clr{m}} 
					\mtype''_1\typearrowp{\clrtwo{1}}\cdots\mtype''_r\typearrowp{\clrtwo{r}}\ltype_0
					&
					\Delta_j\ctypes{k_j} \monoToBlue \tm_j : \mtypetwo_j
				}
				&
				z_i : \mtype'_i\ctypes{e_i}\monoToBlue{\tmthree}_i : \mtype''_i
			}
		}
	}
	}
	\]
	From $\Delta_j\ctypes{k_j} \monoToBlue \tm_j : \mtypetwo_j$, it follows by 
	induction hypothesis that there are derivations of $\Delta'_j\ctypes{k'_j} 
	\monoToBlue \tmtwo_j : \mtypetwo'_j$ with $\Pair{\Delta'_j}{\mtypetwo'_j} 
	\whiterposp {d_j} \Pair{\Delta_j}{\mtypetwo_j}$ where $0\leq d_j = 
	k_j - k_j'$ for each $1\leq j \leq m$. 
	Since $\tmthree_i\slbtleq z_i$, and $z_i : 
	\mtype'_i\ctypes{e_i}\monoToBlue{\tmthree}_i : \mtype''_i$, by applying 
	Point 2 of \reflemma{eta-id} we get $\Pair{z_i : 
	\mtype'''_i}{\mtype'''_i}\whiterposp{e_i}\Pair{z_i : \mtype'_i}{\mtype_i''} 
	$, for some 
	$\mtype'''_i$. In particular, we have:
	\begin{itemize}
		\item $\mtype'''_i \whiternegp{e'_i} \mtype'_i$ and 
		$\mtype'''_i\whiterposp{e''_i} 
		\mtype''_i$ for some $e'_i + e''_i = e_i$.
		\item $\Delta'_j\whiternegp{d^1_j}\Delta_j$ and 
		$\mtypetwo_j'\whiterposp{d^2_j} 
		\mtypetwo_j$ for some $d^1_j + d^2_j = d_j$.
	\end{itemize}
	
	So, we can construct the following derivation for $\tmtwo$, where 
	$\typctx',\vec x : \vec \mtype^* = y : [\ltype_2] + \sum_j\Delta'_j$ for
	$\ltype_2 = 
	\mtypetwo'_1\typearrowp{\bullet}\cdots\mtypetwo'_m\typearrowp{\bullet} 
	\mtype'''_1\typearrowp{\clrtwo{1}}\cdots\mtype'''_r\typearrowp{\clrtwo{r}}\ltype_0$.
	\[
	\resizebox{\textwidth}{!}{%
	\infer{\typctx'\ctypes{(\sum_j k'_j + \Kr{\clr{j}}{\bullet})}
		\lam^\bullet x_1\dots x_n.y\bullet 
		\monoToBlue{\tmtwo_1}\bullet\cdot\bullet\monoToBlue \tmtwo_m : 
		\mtype^*_1\typearrowp{\bullet}\cdots\mtype^*_n\typearrowp{\bullet} 
		\mtype'''_1\typearrowp{\clrtwo{1}}\cdots\mtype'''_r\typearrowp{\clrtwo{r}}\ltype_0}{
		\infer{\typctx',\vec x : \vec \mtype^*\ctypes{(\sum_j k'_j + 
		\Kr{\clr{j}}{\bullet})} y\bullet 
		\monoToBlue{\tmtwo_1}\bullet\cdot\bullet\monoToBlue \tmtwo_m : 
		\mtype'''_1\typearrowp{\clrtwo{1}}\cdots\mtype'''_r\typearrowp{\clrtwo{r}}\ltype_0}{
			y : [\ltype_2]\ctypes{0} y : 
			\mtypetwo'_1\typearrowp{\clr{1}}\cdots\mtypetwo'_m\typearrowp{\clr{m}}
			 \mtype'''_1\typearrowp{\clrtwo{1}}\cdots\mtype'''_r 
			 \typearrowp{\clrtwo{r}}\ltype_0
			&
			\Delta'_j\ctypes{k'_j} \monoToBlue \tmtwo_j : \mtypetwo'_j
		}
	}
	}
	\]
	We must show that there exists $0\leq p = 
	\sum_j(k_j -k'_j)
	+ \sum_i(e_i +\Kr{\clrtwo{i}}{\bullet})$ such that
	\[
	\begin{array}{lcl}
		&&\Pair{\typctx'}{\mtype^*_1\typearrowp{\bullet}\cdots\mtype^*_n\typearrowp{\bullet}
		\mtype'''_1\typearrowp{\clrtwo{1}}\cdots\mtype'''_r\typearrowp{\clrtwo{r}}\ltype_0}\\
		&\whiterposp 
		p&\Pair{\typctx}{\mtype_1\typearrowp{\bullet}\cdots\mtype_n\typearrowp{\bullet}
		\mtype'_1\typearrowp{\bullet}\cdots\mtype'_r\typearrowp{\bullet}\ltype_0}\\
	\end{array}
	\]
	We need to perform some calculations. We have: 
	\[
	\sum_j \Delta'_j \whiternegp{(\sum_j d^1_j)}\sum_j\Delta_j\qquad\textrm{ and 
	}\qquad 
	\ltype_2 \whiternegp{(\sum_j d^2_j + \sum_i e''_i)} \ltype_1
	\]
	Since $\typctx',\vec x : \vec\mtype^* = y : [\ltype_2] + 
	\sum_j\Delta'_j$ and 
	$\typctx,\vec x : \vec \mtype = y : [\ltype_1] + \sum_j\Delta_j$ and $\sum_j 
	d_j = \sum_j(d^1_j+d^2_j)$, we get:
	\begin{equation}\label{eq:context}
		\typctx',\vec x : \vec\mtype^* \whiternegp{\sum_j d_j+ \sum_i e''_i} 
		\typctx,\vec x : \vec \mtype
	\end{equation}
	From $\mtype'''_i \whiternegp{e'_i} \mtype'_i$, we get:
	\begin{equation}\label{eq:type}
		\mtype'''_1\typearrowp{\clrtwo{1}}\cdots\mtype'''_r\typearrowp{\clrtwo{r}}\ltype_0
		 \whiterposp{\sum_i(\Kr{\clrtwo{i}}{\bullet}
		 + 
		 e'_i)}\mtype'_1\typearrowp{\bullet}\cdots\mtype'_r\typearrowp{\bullet}\ltype_0
	\end{equation}
	Therefore, we are going to take
	\begin{align}
	p &:= (\sum_j d_j) + \sum_i e''_i + \sum_i(\Kr{\clrtwo{i}}{\bullet} + e'_i) = 
	\sum_j d_j+ \sum_i(e_i+\Kr{\clrtwo{i}}{\bullet})\\
	& = \sum_j (k_j -k'_j)+ \sum_i(e_i+\Kr{\clrtwo{i}}{\bullet})
	\end{align}
	
	Then, we have:
	\[
	\resizebox{\textwidth}{!}{%
	\infer{\Pair{\typctx'}{\mtype^*_1\typearrowp{\bullet}\cdots\mtype^*_n\typearrowp{\bullet}
			\mtype'''_1\typearrowp{\clrtwo{1}}\cdots\mtype'''_r\typearrowp{\clrtwo{r}}\ltype_0}
		\whiterposp 
		p\Pair{\typctx}{\mtype_1\typearrowp{\bullet}\cdots\mtype_n\typearrowp{\bullet}
			\mtype'_1\typearrowp{\bullet}\cdots\mtype'_r\typearrowp{\bullet}\ltype_0}}{
		\infer{\Pair{\typctx',\vec{x}:\mtype^*}{
				\mtype'''_1\typearrowp{\clrtwo{1}}\cdots\mtype'''_r\typearrowp{\clrtwo{r}}\ltype_0}
			\whiterposp 
			p\Pair{\typctx,\vec{x}:\vec{\mtype}}{
				\mtype'_1\typearrowp{\bullet}\cdots\mtype'_r\typearrowp{\bullet}\ltype_0}}{
		\infer{\typctx',\vec{x}:\mtype^* \whiternegp {\sum_j d_j+ \sum_i 
		e''_i} 
		\typctx,\vec{x}:\vec{\mtype}  }{\text{Eq. \ref{eq:context}}}	&& 
		\infer{\mtype'''_1\typearrowp{\clrtwo{1}}\cdots\mtype'''_r\typearrowp{\clrtwo{r}}\ltype_0
	\whiterposp {\sum_i(\Kr{\clrtwo{i}}{\bullet}
		+ 
		e'_i)}
	\mtype'_1\typearrowp{\bullet}\cdots\mtype'_r\typearrowp{\bullet}\ltype_0}
	 {\text{Eq. \ref{eq:type}}}}
	}
	}
	\]

%
%
%
	
	\end{itemize}
	This concludes the proof. \qed
\end{proof}

Checkers head reduction inherits various expected properties of head reduction, such as determinism. In particular, we shall use the following immediate substitutivity property.
\begin{lemma}[Substitutivity {\cite[Lemma 3.5]{InteractionEquivalence}}]
	\label{l:color-substitutivity}
	Let $\tm,\tmp,\tmtwo\in\Lambdac$ and $\relation{R}\in\{\bnointsym,\bintsym,\bchsym,\hnointsym,\hintsym,\hchsym\}$. If $\tm\Rew{\relation{R}}\tmp$ then $\tm\isub\var\tmtwo\Rew{\relation{R}}\tmp\isub\var\tmtwo$. 
\end{lemma}

%

\paragraph{Terminology and Notations for the Proof}
As customary in mathematical analysis, we say that a relation $\relation{P}(-)$ holds for all $K\in\nat$ \emph{large enough} whenever there exists a $K'\in\nat$ such that $\relation{P}(K)$ holds for all $K\ge K'$. We also use the notation $tu^{\sim n}$ for $(\cdots((tu)u)\cdots)u$ ($n$ times). 
Also, we say that two head normal forms $h,h'$ are \emph{spine equivalent}, written $h \speq h'$, if there are $n,k\ge 0$ such that:
\begin{equation}\label{eq:t-le-u}
	h= \las{\var}{n} \vartwo \,\apps{\tm}{k} \quad\textrm{ and }\quad
	h'= \lambda \var_1\ldots \var_{n}. \vartwo \,\apps{\tmtwo}{k}.
\end{equation}
On closed \lam-terms, the B\"ohm out technique amounts to applying the \emph{tupler} $\Tupler{n}$ and the \emph{$i$-th selector $\Proj{n}{i}$} defined as follows:
\begin{center}$
\begin{array}{rrllllllll}
\textsc{$n$-tuples} & \Tuple{t_1,\dots,t_n} & \defeq & \lam x.xt_1\cdots t_n, & \mbox{ with $x$ fresh};
\\
\textsc{Tuplers}& \Tupler n &\defeq & \lam x_1\ldots x_n.\Tuple{x_1,\dots,x_n};
\\
\textsc{Selectors} & \Proj{n}{i} & \defeq & \lam x_1\ldots x_n.x_i,&\textrm{with } 1\le i \le n.
	\end{array}$
\end{center}
So, the tupler $\Tupler n$ takes $n$ arguments $\tm_1,\ldots,\tm_n$ and returns the tuple $\Tuple{\tm_1,\dots,\tm_n}$, while the selector $\Proj{n}{i}$ takes $n$ arguments $\tm_1,\dots,\tm_n$ and returns the $i$-th argument $\tm_i$.
Note that $\Proj{1}1 = \comb{I}$. Then, $\Tupler nt_1\cdots t_nu \toh^* ut_1\cdots t_n$ and $\Proj{n}{i}t_1\cdots t_n\toh^* t_i$, whence we have the following combined \emph{extraction property}:
\begin{equation}
\begin{array}{ccc}
\Tupler nt_1\cdots t_n\Proj{n}{i} &\toh^*& t_i. 
\end{array}
\label{eq:extraction-property}
\end{equation}

\gettoappendix{l:separating-eta-red}

\begin{proof} 
\applabel{l:separating-eta-red} 
We prove a stronger statement, \ie that there exist closed terms $\vec s\in\Lambda$ such that, for all $\vec y$ containing $\FV{t}\cup\FV{u}$  and for all $K\in\nat$ large enough, the following holds:
\begin{center}$
\rapp{
	\monoToBlue{t}\isub{\vec y}{\monoToRed{\Tupler{K}}}
}{\monoToRed{\vec{s}}}\bshcol{i}
\quad\text{ and }\quad
\rapp{
	\monoToBlue{u}\isub{\vec y}{\monoToRed{\Tupler{K}}}
}{\monoToRed{\vec{s}}}\bshcol{i'}\textrm{ with }i' > i.
$\end{center}
Given variables $\vec x$ and a \lam-term $t$ we write $\sigma_{\vec x}$ for $\isub{\vec x}{\monoToRed{\Tupler{K}}}$, and $t^{\sigma_{\vec x}}$ for $ t\isub{\vec x}{\monoToRed{\Tupler{K}}}$.

Note that $t \not\etaredbtleq u$ is only possible if $t\bsh{}$. Moreover, $t\bsh{}\!\!h$ and $t\etabtle u$ entail $u\bsh{}\!\!h'$, for some $h'$.
We proceed by induction on the length of a minimal path $\delta\in\nat^*$ such that $t\restr_\delta\ \not\speq u\restr_\delta$.

\underline{Base case} $\delta = \emptyseq$, \ie $h\ \not\speq h'$. Then $t\etabtle u$ is only possible if the amount of spine abstractions and applications in $h,h'$ can be matched via $\eta$-expansions appearing in $h'$, that is:
\begin{center}
	\begin{tabular}{c c}
	\multicolumn{2}{c}{$
	t \to^*_\head h = \lam x_1\dots x_n.y\,t_1\cdots t_k
	$}
	\\
	and &
	$
	u \to^*_\head h'= \lam x_1\dots x_nz_1\dots z_m.y\,u_1\cdots u_{k+m}
$
\end{tabular}
\end{center}
for $n,k\ge0$ and $m>0$. 
There are two subcases to consider, depending on whether $y$ is free.
\begin{enumerate}
\item \emph{$y$ is free}, \ie $y\in\vec y$. Take any $K \ge k+m$, and  empty $\vec s$. For $t$, we have:
{ \begin{center}$
	\begin{array}{c|ll}
	\intsym\textsc{-Steps}&\textsc{Terms and $\nointsym$-steps}
	\\[2pt]
	\hline
	&
	\monoToBlue{t}\isub{\vec y}{\monoToRed{\Tupler{K}}}
	\ \ \tohnoint^* \ \ 
	\monoToBlue{h}\isub{\vec y}{\monoToRed{\Tupler{K}}},\phantom{X^{X^{X^{}}}}\hfill\textrm{ by L.~\ref{l:correspondence}(i) \&  \reflemmaeq{color-substitutivity}}
		\\[2pt]
	=	
	&
	\bapp{\bapp{\bapp{
		\manyblam{n}{x}\monoToRed{\Tupler{K}}
	}{{\monoToBlue{t_1}}^{\sigma_{\vec y}}}}{\cdots}}{{\monoToBlue{t_k}}^{\sigma_{\vec y}}}
	\\[2pt]	
	\tohint^k	
	&
	\manyblam{n}{x}\manyrlam[k+1]{K}{w}\Tuple{{\monoToBlue{t_1}}^{\sigma_{\vec y}},\dots,{\monoToBlue{t_k}}^{\sigma_{\vec y}},w_{k+1},\dots,w_K}_{\redclr}\\[2pt]
	\end{array}
$\end{center}}
where $\Tuple{-,\dots,-}_{\redclr} $ denotes the tuple with white applications $ \lam z.z-\circ\cdots\circ -$. For $u$, we have:
{\begin{center}$
	\begin{array}{c|ll}
	\intsym\textsc{-Steps}&\textsc{Terms and $\nointsym$-steps}
	\\[2pt]
	\hline
	&

		\monoToBlue{u}\isub{\vec y}{\monoToRed{\Tupler{K}}}
	\ \ \tohnoint^* \ \ 
		\monoToBlue{h'}\isub{\vec y}{\monoToRed{\Tupler{K}}},\phantom{X^{X^{X^{}}}}\hfill\textrm{ by L.~\ref{l:correspondence}(i) \& \reflemmaeq{color-substitutivity}},
		\\[2pt]
=	&
	\bapp{\bapp{\bapp{
\monoToRed{\Tupler{K}}
	}{{\monoToBlue{u_1}}^{\sigma_{\vec y}}}}{\cdots}}{{\monoToBlue{u_{k+m}}}^{\sigma_{\vec y}}}
	\\[2pt]
\tohint^{k+m}	&
		\manyblam{n}{x}\manyrlam[k+1]{K}{w}\Tuple{{\monoToBlue{u_1}}^{\sigma_{\vec y}},\dots,{\monoToBlue{u_{k+m}}}^{\sigma_{\vec y}},w_{k+m+1},\dots,w_K}_{\redclr}
	\end{array}
$\end{center}}
Summing up, ${\monoToBlue{t}}^{\sigma_{\vec y}}\bshcol{k}$ and  ${\monoToBlue{u}}^{\sigma_{\vec y}}\bshcol{k+m}$. The statement holds since $m>0$.

\item \emph{$y$ is bound}, \ie $y = x_j\in\vec x$. Take any $K\ge k+m$, and let the arguments $\vec s$ be $n$ copies of $\Tupler{K}$ (noted $\Tupler{K}^{\sim n}$ for short). On the one hand:
{\begin{center}$
	\begin{array}{c|lll}
	\intsym\textsc{-Steps}&\textsc{Terms and $\nointsym$-steps}
	\\[2pt]
	\hline
	&\rapp{\monoToBlue{t}\isub{\vec y}{\monoToRed{\Tupler{K}}}}{{\monoToRed{\Tupler{K}}}^{\sim n}}
	\ \ \tohnoint^* \ \ 
\rapp{
		\monoToBlue{h}\isub{\vec y}{\monoToRed{\Tupler{K}}}
	}{{\monoToRed{\Tupler{K}}}^{\sim n}},
	&\textrm{by L.~\ref{l:correspondence}(i) \& \reflemmaeq{color-substitutivity}},
	\\[2pt]
=	&
	\rapp{
		\big(\manyblam{n}{x}\bapp{\bapp{\bapp{x_j}{{\monoToBlue{t_1}}^{\sigma_{\vec y}}}}{\cdots}}{{\monoToBlue{t_k}}^{\sigma_{\vec y}}}\big)
	}{{\monoToRed{\Tupler{K}}}^{\sim n}}\\[2pt]
\tohint^n	&
	\bapp{\bapp{\bapp{
		\monoToRed{\Tupler{K}}
	}{{\monoToBlue{t_1}}^{\sigma_{\vec x\vec y}}}}{\cdots}}{{\monoToBlue{t_k}}^{\sigma_{\vec x\vec y}}}\\[2pt]	
\tohint^k	&
		\manyrlam[k+1]{K}{w}\Tuple{{\monoToBlue{t_1}}^{\sigma_{\vec x\vec y}},\dots,{\monoToBlue{t_k}}^{\sigma_{\vec x\vec y}},w_{k+1},\dots,w_K}_{\redclr}
	\end{array}
$\end{center}}
On the other hand:
{\begin{center}$
	\begin{array}{c|lll}
	\intsym\textsc{-Steps}&\textsc{Terms and $\nointsym$-steps}
	\\[2pt]
	\hline
	&\rapp{\monoToBlue{u}\isub{\vec y}{\monoToRed{\Tupler{K}}}}{{\monoToRed{\Tupler{K}}}^{\sim n}}
	\ \ \tohnoint^* \ \ 
	\rapp{
		\monoToBlue{h'}\isub{\vec y}{\monoToRed{\Tupler{K}}}
	}{{\monoToRed{\Tupler{K}}}^{\sim n}},
		\ \ \ \ \textrm{by L.~\ref{l:correspondence}(i) \& \reflemmaeq{color-substitutivity}},
	\\[2pt]
=	&
	\rapp{
		\big({\lambda_{\blueclr\cdots \blueclr}}x_1\ldots x_{n}\vec z.\,\bapp{\bapp{\bapp{x_j}{{\monoToBlue{u_1}}^{\sigma_{\vec y}}}}{\cdots}}{{\monoToBlue{u_{k+m}}}^{\sigma_{\vec y}}}\big)
	}{{\monoToRed{\Tupler{K}}}^{\sim n}}\\[2pt]
\tohch^n	&
	\bapp{\bapp{\bapp{
				\monoToRed{\Tupler{K}}
	}{{\monoToBlue{u_1}}^{\sigma_{\vec x\vec y}}}}{\cdots}}{{\monoToBlue{u_{k+m}}}^{\sigma_{\vec x\vec y}}}\\[2pt]	
\tohch^{k+m}&
		\manyrlam[k+m+1]{K}{w}\Tuple{{\monoToBlue{u_1}}^{\sigma_{\vec x\vec y}},\dots,{\monoToBlue{u_{k+m}}}^{\sigma_{\vec x\vec y}},w_{k+m+1},\dots,w_K}_{\redclr}.
	\end{array}
$\end{center}}
Summing up, ${\monoToBlue{t}}^{\sigma_{\vec y}}\bshcol{n+k}$ and  ${\monoToBlue{u}}^{\sigma_{\vec y}}\bshcol{n+k+m}$. The statement holds as $m>0$.
\end{enumerate}

 \underline{Inductive case} $\delta = j\cdot\gamma$. In this case, we must have:
\begin{center}$
	t \to^*_\head h = \lam x_1\dots x_n.y\,t_1\cdots t_k
	\qquad
	\textrm{ and }
	\qquad	
	u \to^*_\head h'= \lam x_1\dots x_n.y\,u_1\cdots u_k
$\end{center}
with $t_j \not\etaredbtleq u_j$ and $(t_l \etabtle u_l)_{l \le k}$. By \ih, there exists $K'$ and $\vec{s'}$ such that for all $K\ge K'$:
\begin{center}$
	\rapp{{\monoToBlue{t_j}}^{\sigma_{\vec x\vec y}}}{\monoToRed{\,\vec{s'}}}
	\bshcol{i}
	\quad\mbox{ and }\quad
	\rapp{{\monoToBlue{u_j}}^{\sigma_{\vec x\vec y}}}{\monoToRed{\,\vec{s'}}}\bshcol{i'}\textrm{ with }i' > i.
$\end{center}
We consider any $K\ge \max\{K',k\}$.
We assume wlog.\ that $y$ is free, the other case being analogous.
{ \begin{center}$
	\begin{array}{c|lll}
	\textsc{Steps}&\textsc{Terms}
	\\[2pt]
	\hline
&
\rapp{
		\rapp{
			\rapp{
				\monoToBlue{t}\isub{\vec y}{\monoToRed{\Tupler{K}}}
			}{{\monoToRed{\Tupler{K}}}^{\sim n+K-k}}
		}{\monoToRed{\Proj{K}{j}}}
	}{\monoToRed{\,\vec{s'}}}, 
	\hfill \textrm{by L.~\ref{l:correspondence}(i) \& \reflemmaeq{color-substitutivity}},
	\\[2pt]
	\tohnoint^* &
	\rapp{
		\rapp{
			\rapp{
				\monoToBlue{h}\isub{\vec y}{\monoToRed{\Tupler{K}}}
			}{{\monoToRed{\Tupler{K}}}^{\sim n+K-k}}
		}{\monoToRed{\Proj{K}{j}}}
	}{\monoToRed{\,\vec{s'}}}
	&
	\\[2pt]
	=
	&
	\rapp{
		\rapp{
			\rapp{
				\big(\manyblam{n}{x}\bapp{\bapp{\bapp{\monoToRed{\Tupler{K}}}{{\monoToBlue{t_1}}^{\sigma_{\vec y}}}}{\cdots}}{{\monoToBlue{t_k}}^{\sigma_{\vec y}}}\big)
			}{{\monoToRed{\Tupler{K}}}^{\sim n+K-k}}
		}{\monoToRed{\Proj{K}{j}}}
	}{\monoToRed{\,\vec{s'}}}\\[2pt]
	\tohint^n&
	\rapp{
		\rapp{
			\rapp{
				\bapp{\bapp{\bapp{\monoToRed{\Tupler{K}}}{{\monoToBlue{t_1}}^{\sigma_{\vec x\vec y}}}}{\cdots}}{{\monoToBlue{t_k}}^{\sigma_{\vec x\vec y}}}
			}{{\monoToRed{\Tupler{K}}}^{\sim K-k}}
		}{\monoToRed{\Proj{K}{j}}}
	}{\monoToRed{\,\vec{s'}}}\\[2pt]				
	\tohint^k\tohnoint^*&
	\rapp{{\monoToBlue{t_j}}^{\sigma_{\vec x\vec y}}}{\monoToRed{\,\vec{s'}}} 
	\hfill \textrm{by \refeq{extraction-property}}.
	\end{array}
$\end{center}}
An identical sequence of steps extracts $\monoToBlue{u_j}$ from the other term, that is, we have:
\begin{center}
$\begin{array}{ccc}
\rapp{\rapp{\rapp{
				\monoToBlue{u}\isub{\vec y}{\monoToRed{\Tupler{K}}}
			}{{\monoToRed{\Tupler{K}}}^{\sim n+K-k}}
		}{\monoToRed{\Proj{K}{j}}}
	}{\monoToRed{\,\vec{s'}}}

&\tohnoint^*\tohint^{n+k}\tohnoint^*&
	\rapp{{\monoToBlue{u_j}}^{\sigma_{\vec x\vec y}}}{\monoToRed{\,\vec{s'}}}
	\end{array}$
	\end{center}
Note the same number of $\intsym$-steps. By defining $\monoToRed{\vec s}$ as the arguments ${\monoToRed{\Tupler{K}}}^{\sim n+K-k}$, $\monoToRed{\Proj{K}{j}}$, $\monoToRed{\,\vec{s'}}$, and by composing with what is obtained by the \ih, we obtain:
\begin{center}$
\rapp{
	\monoToBlue{t}\isub{\vec y}{\monoToRed{\Tupler{K}}}
}{\monoToRed{\vec{s}}}\bshcol{n+k+i}
\quad\text{and}\quad
\rapp{
	\monoToBlue{u}\isub{\vec y}{\monoToRed{\Tupler{K}}}
}{\monoToRed{\vec{s}}}\bshcol{n+k+i'},
$\end{center}
which is an instance of the statement because $i < i'$ by \ih\qed
\end{proof}


\section{Proofs of \refsect{compatibility} (Compositionality of Polarized Whiter-Cheaper)}

\gettoappendix{l:repaint}


\begin{proof}\applabel{l:repaint}
	By induction on the structure of the term $\tm$. 
	
	\begin{itemize}
		\item Variable, that is $\tm=\var$.
	In that case we have a derivation:
	\[
	\infer{x : [\ltype]\hasstype[0] x : \ltype}{}
	\]
	if $\Pair{\typctx'}{\ltype'} \whiternegp 1 \Pair{x:[\ltype]}{\ltype}$, we have that either:
	\begin{itemize}
		\item \emph{The change is on the right:}
		$\typctx' = x:[\ltype]$ with $\ltype' \whiternegp 1 \ltype$.
		
		Then we can take the following derivation
		\[
		\infer{x : [\ltype']\hasstype[0] x : \ltype'}{}
		\]
		
		for which $\Pair{x:[\ltype']}{\ltype'} \whiterposp 1 \Pair{x:[\ltype]}{\ltype'}$ and the ($k$-)index of the derivation remains $0$.
		
		\item or \emph{the change is on the left:} $\typctx' = x:[\ltype'']$ with $\ltype'' \whiterposp 1 \ltype$ and $\ltype' = \ltype$.
		
		Then we can take the following derivation
		\[
		\infer{x : [\ltype'']\hasstype[0] x : \ltype''}{}. 
		\]
		
		for which $\Pair{x:[\ltype'']}{\ltype''} \whiterposp 1 \Pair{x:[\ltype'']}{\ltype}$ and the ($k$-)index of the derivation remains $0$.
	\end{itemize} 
	
	\item \emph{Abstraction}, that is $\tm = \lap\colr\var\tmtwo$ such that:
	\[
	\infer{\typctx\hasstype \lap\colr\var\tmtwo : \mtype \typearrowp{\colr}\ltype}{\typctx,x : \mtype \hasstype \tmtwo : \ltype}
	\] 
	suppose $\Pair{\typctx'}{\mtype' \typearrowp{\colr}\ltype'} \whiternegp 1 \Pair{\typctx}{\mtype \typearrowp{\colr}\ltype}$. (Note that the annotation on the arrow in the result type must be unchanged because we consider only $\whiternegp 1$.)
	
	We have that also $\Pair{\typctx',x:\mtype'}{\ltype'}\whiternegp 1 \Pair{\typctx, x:\mtype}{\ltype}$ so we can apply the inductive hypothesis and find a derivation of $\typctx'', x:\mtype'' \hasstype[k'] \tmtwo:\ltype''$ with either $k=k'$ and $\Pair{\typctx'', x:\mtype''}{\ltype''}\whiterposp 1
	\Pair{\typctx', x:\mtype'}{\ltype'}$, or $\typctx'' = \typctx'$, $\mtype'' = \mtype'$, $\ltype'' = \ltype'$ and $\abs{k'-k} =1$. In either case we can derive
	\[
	\typctx''\hasstype[k'] \lap\colr\var\tmtwo : \mtype'' \typearrowp{\colr}\ltype''
	\]
	satisfying the required condition.
	
	
	\item \emph{Application,} that is $\tm=\appp\colrtwo\tmtwo\tmthree$ such that:
	\[
	\infer{\typctx_1 + \typctx_2 \hasstype[k_1+k_2+\Kr{\colr}{\colrtwo}]\appp\colrtwo\tmtwo\tmthree :\ltype }
	{\typctx_1 \hasstype[k_1] \tmtwo : \mtype\typearrowp{\colr}\ltype
		&
		\typctx_2 \hasstype[k_2] \tmthree : \mtype
	}
	\]
	Let $\Pair{\typctx'}{\ltype'} \whiternegp 1 \Pair{\typctx_1+\typctx_2}{\ltype}$. Note that we can write $\typctx'$ (not necessarily uniquely) as $\typctx'_1 + \typctx'_2$ such that one of the following holds:
	\bsub
	\item $\typctx'_1 = \typctx_1$, $\ltype' = \ltype$ and $\typctx'_2 \whiterposp 1 \typctx_2$
	\item $\typctx'_1 = \typctx_1$, $\typctx'_2 = \typctx_2$ and $\ltype' \whiternegp 1 \ltype$; or
	\item $\typctx'_1 \whiterposp 1 \typctx_1$, $\typctx'_2 = \typctx_2$ and $\ltype' =  \ltype$. 
	\esub
	
	In each case, we repeatedly apply the induction hypothesis to $\tmtwo$ and $\tmthree$ until we reach a typing that works for both:
	
	\begin{center}
		there exists $\typctx_1'',\mtype'',\ltype'',\typctx_2''$ such that $\typctx_1'' \hasstype[k_1''] \tmtwo : \mtype''\typearrowp{\colr''}\ltype''
		$, $
		\typctx_2'' \hasstype[k_2''] \tmthree : \mtype''$ and there exists $0\leq i \leq 1$ such that $\Pair{\typctx_1''+ \typctx_2''}{\ltype''}\whiterposp i\Pair{\typctx'}{\ltype'}$ and $\abs{k_1''+k_2''+\Kr{\colr''}{\colrtwo} - k_1+k_2+\Kr{\colr}{\colrtwo}  } \leq 1 - i$
	\end{center}
	
	We fully develop how to obtain such an $\mtype''$ for the third case, but we can construct similarly $\mtype''$ for the other two cases.
	\begin{enumerate}
\item[(iii)]	
		We consider the third case for now. 
		We have that:		
		\[
		\typctx_1 \hasstype[k_1] \tmtwo : \mtype\typearrowp{\colr}\ltype \text{ and } \typctx'_1 \whiterposp 1 \typctx_1
		\]
		so we can apply the \ih to get $\typctxtwo_1, \mtype_1, \ltype_1,l_1, 0\leq i_1 \leq 1, \clr{1}$ such that 
		\[
		\typctxtwo_1 \hasstype[l_1] \tmtwo : \mtype_1\typearrowp{\clr{1}}\ltype_1 \text{, } \Pair{\typctxtwo_1}{\mtype_1\typearrowp{\clr{1}}\ltype_1} \whiterposp{i_1} \Pair{\typctx'_1}{ \mtype\typearrowp{\colr}\ltype} \text{ and } \abs{k_1 -l_1}\leq 1- i_1
		\]
		As there is only one color change, we can study two sub-cases:
		\begin{itemize}
			\item The color change is not in $\mtype$, that is $\mtype_1 = \mtype$.
			
			In this case, we can immediately return to the main derivation:
			
				\[
			\infer{\typctxtwo_1 + \typctx_2 \hasstype[l_1+k_2+\Kr{\clr{1}}{\colrtwo}]\appp\colrtwo\tmtwo\tmthree :\ltype }
			{\typctxtwo_1 \hasstype[l_1] \tmtwo : \mtype\typearrowp{\clr{1}}\ltype_1
				&
				\typctx_2 \hasstype[k_2] \tmthree : \mtype
			}
			\]
			
			Note that $\abs{\Kr{\clr{}}{\colrtwo} - \Kr{\clr{1}}{\colrtwo}} = \Kr{\clr{1}}{\colr}$ (as $\clr{1}$ is either the same color as $\colr$ or whiter), which allows us to conclude.
			
			\item The color change is in $\mtype$, hence $ \mtype_1 \whiternegp 1 \mtype$ but $\typctxtwo_1=\typctx'_1,\ltype= \ltype_1$ and $\clr{1} = \colr$. 
			
			In this case, we can reapply the \ih this time to 
			$\typctx_2 \hasstype[k_2] \tmthree : \mtype$ with $\mtype_1 \whiternegp 1 \mtype$ (in fact applying it to the one linear type in $\mtype$ whose color is changed). If we fall back to the case where the color change is not in $\mtype_1$ we are done. If we do not, we keep applying the \ih back and forth between $\tm$ and $\tmtwo$.
			
			We end up with a sequence of multi types $\mtype_0, \ldots, \mtype_j$
			such that $\mtype_0 = \mtype$ and 
			\[
			\ldots \whiternegp 1 \mtype_{2i} \whiterposp 1 \mtype_{2i-1} \whiternegp 1 \ldots \whiternegp 1 \mtype_2 \whiterposp 1 \mtype_1 \whiternegp 1 \mtype
			\]
			and for each $i$ we have
			\[
			\typctx'_1 \hasstype[k_1] \tm :\mtype_{2i+1} \typearrowp{\colr} \ltype
			\qquad
			\typctx_2 \hasstype[k_2] \tmtwo :\mtype_{2i}.
			\]
			Note that there cannot be an infinite such sequence because the type $\mtype$ is finite so can only be recoloured a finite number of times. We show that if $j$ is the maximal index of the sequence, $\mtype_j$ is the right candidate for our argument type. We proceed by case analysis on whether $j$ is odd or even:
			
			\begin{itemize}
				\item 
				Suppose $j$ is odd. Then we have
				\[
				\typctx'_1 \hasstype[k_1] \tm :\mtype_{j} \typearrowp{\colr} \ltype
				\qquad
				\typctx_2 \hasstype[k_2] \tmtwo :\mtype_{j-1}
				\]
				and $\mtype_j \whiternegp 1 \mtype_{j-1}$. By the inductive hypothesis there exists a derivation of $\typctx'_2 \hasstype[k'_2] \tmtwo : \mtype'_j$ with either
				\begin{itemize}
					\item $\Pair{\typctx'_2}{\mtype'_j} \whiterposp 1 \Pair{\typctx_2}{\mtype_j}$ and $k'_2 = k_2$; or
					\item $\Pair{\typctx'_2}{\mtype'_j} = \Pair{\typctx_2}{\mtype_j}$ and $\abs{k'_2 - k_2} =1$. 
				\end{itemize}
				In the first case, by maximality of the sequence of $\mtype_i$, it must be the case that $\mtype'_j = \mtype_j$ and $\typctx'_2 \whiternegp 1 \typctx_2$. Then we can derive
				\[
				\infer{\typctx'_1 + \typctx'_2 \hasstype[k_1+k_2+\Kr{\colr}{\colrtwo}] \tm\cdot^{\colrtwo} \tmtwo :\ltype }
				{\typctx'_1 \hasstype[k_1] \tm : \mtype_j\typearrowp{\colr}\ltype
					&
					\typctx''_2 \hasstype[k_2] \tmtwo : \mtype_j
				}
				\]
				Recalling that $\typctx' = \typctx'_1 + \typctx_2$ we observe that $\Pair{\typctx'_1 + \typctx'_2}{\ltype} \whiterposp 1 \Pair{\typctx'}{\ltype}$ as required.
				
				In the second case, we derive
				\[
				\infer{\typctx'_1 + \typctx_2 \hasstype[k_1+k'_2+\Kr{\colr}{\colrtwo}] \tm\cdot^{\colrtwo} \tmtwo :\ltype }
				{\typctx'_1 \hasstype[k_1] \tm : \mtype_j\typearrowp{\colr}\ltype
					&
					\typctx_2 \hasstype[k'_2] \tmtwo : \mtype_j
				}
				\]
				and since $\typctx' = \typctx'_1 + \typctx_2$ the argument is complete.
				
				\item 
				Now suppose $j$ is even. Then we have
				\[
				\typctx'_1 \hasstype[k_1] \tm :\mtype_{j-1} \typearrowp{\colr} \ltype
				\qquad
				\typctx_2 \hasstype[k_2] \tmtwo :\mtype_j
				\]
				and $\mtype_j \whiterposp 1 \mtype_{j-1}$, so that $\mtype_j\typearrowp{\colr}\ltype \whiternegp 1 \mtype_{j-1}\typearrowp{\colr}\ltype$. By the inductive hypothesis we obtain $\typctx''_1 \hasstype[k'_1] \tm :\mtype'_j \typearrowp{\colrp} \ltype'$ satisfying one of:
				\begin{itemize}
					\item $\Pair{\typctx''_1}{\mtype'_j \typearrowp{\colrp}\ltype'} \whiterposp 1
					\Pair{\typctx'_1}{\mtype_j \typearrowp{\colr} \ltype}$ and $k'_1=k_1$; or
					\item $\Pair{\typctx''_1}{\mtype'_j \typearrowp{\colrp} \ltype'} = 
					\Pair{\typctx'_1}{\mtype_j \typearrowp{\colr} \ltype}$ and $\abs{k'_1 - k_1}=1$. 
				\end{itemize}
				In the first case, note that the two typings differ in exactly one colouring, and that this cannot be between $\mtype'_j$ and $\mtype_j$ by maximality of the sequence of $\mtype_i$s. Therefore $\mtype'_j = \mtype_j$. If $\colrp = \colr$ then we have
				$\Pair{\typctx''_1}{\mtype_j \typearrowp{\colr} \ltype'} \whiterposp 1    \Pair{\typctx'_1}{\mtype_j \typearrowp{\colr} \ltype}$ and we can derive
				\[
				\infer{\typctx''_1 + \typctx_2 \hasstype[k_1+k_2+\Kr{\colr}{\colrtwo}] \tm\cdot^{\colrtwo} \tmtwo :\ltype }
				{\typctx''_1 \hasstype[k_1] \tm : \mtype_j\typearrowp{\colr}{}\ltype'
					&
					\typctx_2 \hasstype[k_2] \tmtwo : \mtype_j
				}
				\]
				with $\Pair{\typctx''_1+\typctx_2}{\ltype'} \whiterposp 1 \Pair{\typctx'_1+\typctx_2}{\ltype}$ as required. Alternatively if $\colrp\not =\colr$, then $\typctx''_1 = \typctx'_1$ and $\ltype'=\ltype$, and we have
				\[
				\infer{\typctx'_1 + \typctx_2 \hasstype[k_1+k_2+\Kr{\colrp}{\colrtwo}] \tm\cdot^{\colrtwo} \tmtwo :\ltype }
				{\typctx'_1 \hasstype[k_1] \tm : \mtype_j\typearrowp{\colrp}\ltype
					&
					\typctx_2 \hasstype[k_2] \tmtwo : \mtype_j
				}
				\]
				But now notice that
				\[
				\abs{(k_1+k_2+\Kr{\colrp}{\colrtwo}) - (k_1+k_2+\Kr{\colr}{\colrtwo})} =
				\abs{\Kr{\colrp}{\colrtwo}- \Kr{\colr}{\colrtwo}} = 1
				\]
				completing the argument. \qed
			\end{itemize}
		\end{itemize}

	\end{enumerate}

	\end{itemize}
\end{proof}

We formally define the following commutation property for the $\whiternegp 1$ and $\whiterposp 1$ relations. The property easily generalizes to $\whiternegp{k}$ and $\whiterposp{l}$ for any $k$ and $l$. Note that we only write down the case \(k=l=1\), from which the general case follows by induction. 

\ignore{
\noindent
\begin{minipage}[t]{0.63\textwidth}
	\vspace{0pt}
	
	\begin{lemma}
		\label{l:commutation-neg-and-pos}
		Let $\typctx,\madewhiterpos{\typctx}, \madewhiterneg{\typctx}$ and $\ltype,\madewhiterpos{\ltype},\madewhiterneg{\ltype}$ such that:
		\begin{itemize}
			\item  $\Pair{\madewhiterneg{\typctx}}{\madewhiterneg{\ltype}}\whiternegp 1 \Pair{\typctx}{\ltype}$ and;
			\item  $\Pair{\madewhiterpos{\typctx}}{\madewhiterpos{\ltype}}\whiterposp 1 \Pair{\typctx}{\ltype}$.
		\end{itemize}
		
		Then there exist $\madewhiternegandpos\typctx,\madewhiternegandpos\ltype$ such that:
		\begin{itemize}
			\item $ \Pair{\madewhiternegandpos\typctx}{\madewhiternegandpos\ltype}
			\whiterposp 1
			\Pair{\madewhiterneg{\typctx}}{\madewhiterneg{\ltype}}$
			\item $ \Pair{\madewhiternegandpos\typctx}{\madewhiternegandpos\ltype} 
			\whiternegp 1
			\Pair{\madewhiterpos{\typctx}}{\madewhiterpos{\ltype}}$
		\end{itemize}
	\end{lemma}
\end{minipage}
\hfill
\begin{minipage}[t]{0.35\textwidth}
	\vspace{10pt}
	\begin{tabular}{|c}
		\begin{tikzcd}
			{			\Pair{\madewhiterneg{\typctx}}{\madewhiterneg{\ltype}}} \arrow[dd, "\rotatebox{90}{$\whiterposp 1$}" description, dashed, no head] 
			&   
			{\Pair{\typctx}{\ltype}} \arrow[dd, "
			\rotatebox{90}{$\whiterposp 1$}" description, no head] \arrow[l, "
			\whiternegp 1" description, no head] 
			\\
			&  
			\\
			{\Pair{\madewhiternegandpos\typctx}{\madewhiternegandpos\ltype}}                
			&  
			{\Pair{\madewhiterpos{\typctx}}{\madewhiterpos{\ltype}}}  \arrow[l, "
			\whiternegp 1" description, dashed, no head]               
		\end{tikzcd}
	\end{tabular}
\end{minipage}

\begin{proof}
	The typing \(\Pair{\madewhiterneg{\typctx}}{\madewhiterneg{\ltype}}\) arises by repainting one negatively-occurring arrow in \(\Pair{\typctx}{\ltype}\) from black to white; similarly \(\Pair{\madewhiterpos{\typctx}}{\madewhiterpos{\ltype}}\) arises by repainting one positively occurring arrow in \(\Pair{\typctx}{\ltype}\). Construct \(\Pair{\madewhiternegandpos\typctx}{\madewhiternegandpos\ltype}\) by making both these repaintings simultaneously.\qed
\end{proof}
}

Note that the following proposition also holds for multiset types instead of linear types.

\gettoappendix{prop:multirepaint}

\begin{proof}\applabel{prop:multirepaint}
	We decompose $\Pair{\typctx'}{\ltype'} \whiternegp {k_1} \Pair{\typctx}{\ltype}$ in the following sequence, where $\Pair{\typctx'}{\ltype'}= \Pair{\typctx_{k_1}}{\ltype_{k_1}}$ and $\Pair{\typctx}{\ltype}=\Pair{\typctx_{0}}{\ltype_{0}}$ and reason by induction on $k_1$.
	
	\begin{tikzcd}[
		every matrix/.append style={name=mycd},
		execute at end picture={
			\draw[decorate, decoration={brace, amplitude=8pt}]
			(mycd-1-1.north) -- 
			(mycd-1-5.north) node[midway, yshift=15pt, font=\scriptsize]{$k_1$};
		}]
		{			\Pair{\typctx_{k_1}}{\ltype_{k_1}} } 
		&   
		{\Pair{\typctx_{k_1 -1}}{\ltype_{k_1-1}}} 
		\arrow[l, "
		\whiternegp 1" description, no head] 
		&
		{\cdots} 
		\arrow[l, "\whiternegp 1" description, no head]		
		& 
		{\Pair{\typctx_{1}}{\ltype_{1}}} 
		\arrow[l, "\whiternegp 1" description, no head] 
		& 
		{\Pair{\typctx_{0}}{\ltype_{0}}} 
		\arrow[l, "\whiternegp 1" description, no head]       
	\end{tikzcd}
	
	\begin{itemize}
		\item $k_1 =0$. The same derivation  $\typctx\hasstype[k] \tm : \ltype$ has the required property.
		
		\item $k_1 \mapsto k_1 + 1$.
		
		\ignore{
			We apply the repainting \reflemma{repaint} to the first negative repainting,  which entails that there exists a derivation of $\typctx'_1 \hasstype[k'] \tm:\ltype_1'$ such that  such that one of the following holds:
			\begin{itemize}
				\item  $\Pair{\typctx'_1}{\ltype'_1}\whiterposp 1 \Pair{\typctx_1}{\ltype_1}$ and $k=k'$; or
				\item  $\Pair{\typctx'_1}{\ltype'_1} = \Pair{\typctx_1}{\ltype_1}$ and $\abs{k'-k} = 1$. 
			\end{itemize} 
			
		}

		\begin{tikzcd}[
			every matrix/.append style={name=mycd},
			execute at end picture={
				\draw[decorate, decoration={brace, amplitude=8pt}]
				(mycd-1-1.north) -- 
				(mycd-1-5.north) node[midway, yshift=15pt, font=\scriptsize]{$k_1 +1$};
				\draw[decorate, decoration={brace, amplitude=8pt,mirror}]
				(mycd-3-1.south) -- 
				(mycd-3-4.south) node[midway, yshift=-15pt, font=\scriptsize]{$k_1$};
			}]
			{			\Pair{\typctx_{k_1+1}}{\ltype_{k_1+1}} } \arrow[dd, "\rotatebox{90}{$\whiterposp i$}" description, dotted, no head] 
			&   
			{\Pair{\typctx_{k_1}}{\ltype_{k_1}}} \arrow[dd, "
			\rotatebox{90}{$\whiterposp i$}" description, no head, dotted] \arrow[l, "
			\whiternegp 1" description, no head] 
			&
			{\cdots} \arrow[dd, "\cdots" description, no head, dotted] 
			\arrow[l, "\whiternegp 1" description, no head]		
			& 
			{\Pair{\typctx_{1}}{\ltype_{1}}} 
			\arrow[l, "\whiternegp 1" description, no head] 
			\arrow[dd, "\rotatebox{90}{$\whiterposp i$}" description, dashed, no head] 
			& 
			{\Pair{\typctx_{0}}{\ltype_{0}}} 
			\arrow[l, "\whiternegp 1" description, no head] 
			\\
			&
			\\
			{			\Pair{\typctx'_{k_1+1}}{\ltype'_{k_1+1}} } 
			&   
			{\Pair{\typctx'_{k_1 }}{\ltype'_{k_1}}} 
			\arrow[l, "\whiternegp 1" description, no head] 
			&
			{\cdots} 
			\arrow[l, "\whiternegp 1" description, no head]		
			& 
			{\Pair{\typctx'_{1}}{\ltype'_{1}}} 
			\arrow[l, "\whiternegp 1" description, no head] 
			& 
		\end{tikzcd}
		
		where the dashed arrow is obtained by the Repainting \reflemma{repaint}, where $i=0$ or $i=1$ depending on the outcome of the repainting and for which $\typctx_1'\ctypes {l} \tm \hastype \ltype_1'$. Moreover, if $i=1$ then $k=l$ and if $i=0$ then $\abs{l-k} = 1$; which can be summed up as $\abs{l-k} \leq 1 - i$ (where $1-i$ is always positive). 
		
		We then close the diagram with the dotted arrows by commutation (\reflemma{commutation-neg-and-pos}). 
		
		From there we can apply the inductive hypothesis, which entails that there exist $\typctx'' \ctypes{k'} \tm\hastype\ltype''$ such that $\Pair{\typctx''}{\ltype''}\whiterposp {k_2} \Pair{\typctx'_{k_1+1}}{\ltype'_{k_1+1}}$
		where $0 \leq k_2 \leq k_1$,  and $\abs{l- k'} \leq k_1 - k_2$.
		\ignore{
			\begin{tikzcd}
				{			\Pair{\typctx_{k_1+1}}{\ltype_{k_1+1}} } \arrow[dd, "\rotatebox{90}{$\whiterposp i$}" description, no head] 
				&   
				{\Pair{\typctx_{k_1}}{\ltype_{k_1}}} \arrow[dd, "
				\rotatebox{90}{$\whiterposp i$}" description, no head] \arrow[l, "
				\whiternegp 1" description, no head] 
				&
				{\cdots} \arrow[dd, "\cdots" description, no head] 
				\arrow[l, "\whiternegp 1" description, no head]		
				& 
				{\Pair{\typctx_{1}}{\ltype_{1}}} 
				\arrow[l, "\whiternegp 1" description, no head] 
				\arrow[dd, "\rotatebox{90}{$\whiterposp i$}" description, no head] 
				& 
				{\Pair{\typctx_{0}}{\ltype_{0}}} 
				\arrow[l, "\whiternegp 1" description, no head] 
				\\
				&
				\\
				{			\Pair{\typctx'_{k_1+1}}{\ltype'_{k_1+1}} } \arrow[dd, "\rotatebox{90}{$\whiterposp {k_2}$}" description, dashed, no head] 
				&   
				{\Pair{\typctx'_{k_1 }}{\ltype'_{k_1}}} 
				\arrow[l, "\whiternegp 1" description, no head] 
				&
				{\cdots} 
				\arrow[l, "\whiternegp 1" description, no head]		
				& 
				{\Pair{\typctx'_{1}}{\ltype'_{1}}} 
				\arrow[l, "\whiternegp 1" description, no head] 
				& 
				\\
				&
				\\
				{			\Pair{\typctx''}{\ltype''} } 
				&  
		\end{tikzcd}}
		
		We have that $\abs{k- k'} \leq \abs{k-l}+\abs{l-k'}\leq k_1 -k_2 + 1 -i$ and $\Pair{\typctx''}{\ltype''}\whiterposp {k_2+i} \Pair{\typctx'}{\ltype'}$ (by transitivity of $\whiterposp -$).	
		\qed

	\end{itemize}
	
\end{proof}

\gettoappendix{prop:apprepaint}


\begin{proof}\applabel{prop:apprepaint}
	We proceed by induction on the number of black-annotated arrows in the types $\mtype$ and $\mtypetwo$. 
	
	\begin{itemize}
		\item In the base case, if there are no black-annotated arrows in $\mtype$ or $\mtypetwo$, then $\mtype \whiterneg_{d} \mtypetwo$ or $\mtypetwo \whiterpos_{d} \mtype$ entails that $d= 0$ and $\mtype = \mtypetwo$. Then we can immediately use the application rule to conclude that $\typctx+\typctxtwo \hasstype[k+l+\Kr{\colr}{\colrtwo}] \tm \cdot^{\colrtwo} \tmtwo : \ltype$, and taking $d'=0$ we conclude. 
		
		\item 
		For the inductive step, we treat the case where $\mtypetwo \whiterpos_{d} \mtype$; the other case is similar and simpler. 
		
		We have $\mtypetwo \typearrowp{\colr}\ltype \whiterneg_{d} \mtype \typearrowp{\colr}\ltype$.   Applying \refprop{multirepaint} to $\tm$ we obtain a typing $\typctx' \hasstype[k'] \tm:\mtype' \typearrowp{\colrp}\ltype'$ where $\Pair{\typctx'}{\mtype' \typearrowp{\colrp}\ltype'} \whiterpos_{d_1} \Pair{\typctx}{\mtypetwo \typearrowp{\colr}\ltype}$ and $k' \leq k + d-d_1$.
		
		The $d_1$ recolorings that transform $\Pair{\typctx}{\mtypetwo \typearrowp{\colr}\ltype}$ into $\Pair{\typctx'}{\mtype' \typearrowp{\colrp}\ltype'}$ split among the components of the type so that 
		there are $d_2, d_3$ such that $\Pair{\typctx'}{\ltype'} \whiterpos_{d_2} \Pair{\typctx}{\ltype}$, $\mtype' \whiterneg_{d_3} \mtypetwo$ and $d_1 = d_2 + d_3 + \Kr{\colr}{\colrp}$. 
		\begin{itemize}
			\item Suppose $d_3 = 0$, hence $\mtype' = \mtypetwo$. In this case, we easily construct the typing derivation for $\appp\colrtwo\tm\tmtwo$ $\typctx'+\typctxtwo \ctypes{k' + l + \Kr{\colrp}{\colrtwo}} \appp\colrtwo\tm\tmtwo \hastype \ltype'$ such that $\Pair{\typctx'+\typctxtwo'}{\ltype'} \whiterposp {d_2 } \Pair{\typctx+\typctxtwo}{\ltype}$ and we check that  $k' + l + \Kr{\colrp}{\colrtwo}\leq k+l+ \Kr{\colr}{\colrtwo} + d -d_2$:
			
			\begin{eqnarray*}
				k' + l + \Kr{\colrp}{\colrtwo}   & \leq & k + d-d_1+l+ \Kr{\colrp}{\colrtwo} \\
				& = &k +l+ d + \Kr{\colrp}{\colrtwo} - d_2 - \Kr{\colr}{\colrp}  \\
				& \leq  & k+l+ \Kr{\colr}{\colrtwo} + d -d_2 \\
			\end{eqnarray*}
			 
			noting that $ \Kr{\colrp}{\colrtwo} - \Kr{\colr}{\colrp} \leq \Kr{\colr}{\colrtwo}$.
			\item Otherwise, $d_3 > 0$, which means that $\mtype'$ has strictly fewer black-annotated arrows than $\mtypetwo$, hence
			we can apply the inductive hypothesis to $\typctx' \hasstype[k'] \tm:\mtype' \typearrowp{\colrp}\ltype'$ and $\typctxtwo \hasstype[l] \tmtwo: \mtypetwo$. This gives 
			us a typing $\typctx'' + \typctxtwo' \hasstype[m] \tm \cdot^{\colrtwo} \tmtwo: \ltype''$ where $\Pair{\typctx''+\typctxtwo'}{\ltype''} \whiterpos_{d'} \Pair{\typctx'+\typctxtwo}{\ltype'}$ and
			$m \leq k' + l +  \Kr{\colrp}{\colrtwo} + d_3 - d'$.
			Observe that $\Pair{\typctx''+\typctxtwo'}{\ltype''} \whiterpos_{d'+d_2} \Pair{\typctx+\typctxtwo}{\ltype}$. To complete the proof we calculate
			\begin{eqnarray*}
				m   & \leq & k'+l+ \Kr{\colrp}{\colrtwo} + d_3 - d'  \\
				& \leq & k + d - d_1 + l+ \Kr{\colrp}{\colrtwo} + d_3 - d'  \\
				& =  & k + l + \Kr{\colrp}{\colrtwo} + d - d_2 - d_3 - \Kr{\colr}{\colrp} + d_3 - d'  \\
				& = & k + l+ \Kr{\colrp}{\colrtwo} -  \Kr{\colr}{\colrp} + d - (d'+d_2)\\
				& \leq & k + l+ \Kr{\colr}{\colrtwo} + d - (d'+d_2)
			\end{eqnarray*}
			noting that $ \Kr{\colrp}{\colrtwo} - \Kr{\colr}{\colrp} \leq \Kr{\colr}{\colrtwo}$. \qed
		\end{itemize} 
	\end{itemize}
	
\end{proof}

\gettoappendix{prop:leqpwc-contextual}

\begin{proof}\applabel{prop:leqpwc-contextual}
	Note that it is sufficient (by transitivity of $\lecolpos$) to prove for all $\tm,\tmtwo,\tmthree$ such that  $\tm \lecolpos \tmtwo$, we have that for all $\colr$:
	\begin{enumerate}[(i)]
		\item $\lap{\colr}\var\tm \lecolpos \lap{\colr}\var\tmtwo$; 
		\item $\appp\colr\tm\tmthree \lecolpos \appp\colr\tmtwo\tmthree$;
		\item $\appp\colr\tmthree\tm\lecolpos  \appp\colr\tmthree\tmtwo$.
	\end{enumerate}
	
	We now prove these three properties:
	
	\begin{enumerate}[(i)]
		\item For the abstraction case, note that $\Tr{\typctx}{\mtype \typearrowp{\colr}\ltype}{k} \in \semint{\lambda^{\colr} x. \tm}$ if and only if
		$\Tr{\typctx, x:\mtype}{\ltype}{k} \in \semint{\tm}$, and $\Pair{\typctx}{\mtype \typearrowp{\colr}\ltype} \whiterposp {d} \Pair{\typctx'}{\mtype' \typearrowp{\colr}\ltype'}$ if and only if $\Pair{\typctx,x:\mtype}{\ltype} \whiterposp d \Pair{\typctx', x:\mtype'}{\ltype'}.$ The result then follows directly.
		
		\item For the first application case, suppose $\Tr{\typctx}{\ltype}{k} \in \semint{\tm \cdot^{\colrtwo} \tmthree}$. Then we must have $\Tr{\typctx_1}{\mtype \typearrowp{\colr}\ltype}{k_1} \in \semint{\tm}$ and $\Tr{\typctx_2}{\mtype}{k_2} \in \semint{\tmthree}$ with $k = k_1 + k_2 + \Kr{\colr}{\colrtwo}$ and $\typctx = \typctx_1 + \typctx_2$. Since $\tm\lecolpos \tmtwo$ we can find $\Tr{\typctx'_1}{\mtype' \typearrowp{\clrp{}} \ltype'}{k'_1} \in \semint{\tmtwo}$ with $\Pair{\typctx'_1}{\mtype' \typearrowp{\clrp{}} \ltype'} \whiterposp {d} \Pair{\typctx_1}{\mtype \typearrowp{\colr}\ltype}$ and $k_1 \geq k'_1 + d$. This implies that there are $d_1$, $d_2$ such that $d = d_1 + d_2 + \Kr{\colr}{\clrp{}}$ and $\Pair{\typctx'_1}{\ltype'}\whiterposp {d_1} \Pair{\typctx_1}{\ltype}$ and $\mtype' \whiternegp {d_2} \mtype$. By Proposition~\ref{prop:apprepaint} there exists $\Tr{\typctx''_1+\typctx'_2}{\ltype''}{m} \in \semint{\tmtwo \cdot^{\colrtwo} \tmthree}$ with $\Pair{\typctx''_1+\typctx'_2}{\ltype''} \whiterposp {d'} \Pair{\typctx'_1+\typctx_2}{\ltype'}$ and $m \leq k'_1 + k_2 + \Kr{\clrp{}}{\colrtwo} +d_2-d'$. Then we have
		\[
		\Pair{\typctx''_1+\typctx'_2}{\ltype''} \whiterposp {d'} \Pair{\typctx'_1+\typctx_2}{\ltype'} \whiterposp {d_1} \Pair{\typctx_1+\typctx_2}{\ltype}
		\]
		and hence $\Pair{\typctx''_1+\typctx'_2}{\ltype''} \whiterposp {d'+d_1} \Pair{\typctx_1+\typctx_2}{\ltype} = \Pair{\typctx}{\ltype}$. It remains to show that $k \geq m+d'+d_1$:
		\begin{eqnarray*}
			m   & \leq & k'_1 + k_2 + \Kr{\clrp{}}{\colrtwo} + d_2 - d'  \\
			& \leq & k_1 - d + k_2 + \Kr{\clrp{}}{\colrtwo} + d_2 - d'  \\
			& =  & k_1 + k_2 + \Kr{\clrp{}}{\colrtwo} +d_2 - d' - d_1 - d_2 - \Kr{\colr}{\clrp{}} \\
			& = &  k_1 + k_2 + \Kr{\clrp{}}{\colrtwo} - \Kr{\colr}{\clrp{}} - d' - d_1 \\
			& \leq & k_1 + k_2 + \Kr{\colr}{\colrtwo}- d' - d_1 \\
			& = & k - d' - d_1
		\end{eqnarray*}
		as required.
		
		\item The other application case is handled symmetrically. Note that in this case we end up in the dual case of \refprop{apprepaint} where $\mtype \whiterposp {d_2} \mtype'$.\qed
	\end{enumerate}

\end{proof}
}
{}

\end{document}